\RequirePackage{etex}
\def\draft{0}
\def\sigconf{0}
\def\big{0}
\def\anon{0}
\def\shownomenclature{1}
\def\masterthesis{0}
\def\cryptology{0}
\def\llncs{0}  

\ifnum\big=1
    \documentclass[final,17pt]{extarticle} \usepackage{fullpage}
    
\else
    \ifnum\sigconf=1
        \documentclass[sigconf,final]{acmart}
    \else 
        \ifnum \llncs=1

            \documentclass[runningheads]{llncs}
        \else
            \ifnum\draft=1
                \documentclass[11pt,draft,a4paper]{article}
            \else
                \ifnum\masterthesis=1
                    \documentclass[11pt,final,a4paper,titlepage]{article}
                \else
                    \documentclass[11pt,final,a4paper]{article}
                \fi
            \fi
            \usepackage{lmodern}
        \fi
    \fi
\fi
\ifnum\sigconf=0\usepackage{authblk}\fi \usepackage{silence} %authbulk must loaded before the silence package!!
\usepackage[modulo]{lineno}

\usepackage{amsthm,amssymb,amsfonts,latexsym,color,paralist,url,ifdraft,mathrsfs,thm-restate,comment,bbold,booktabs}
\usepackage{footnote}%for savenotes environment
\usepackage[utf8]{inputenc} \usepackage[autostyle=false, style=english]{csquotes} \MakeOuterQuote{"}

\usepackage[T1]{fontenc}
\ifnum\sigconf=0
    \ifnum \shownomenclature=1
        \usepackage[refpage,prefix]{nomencl} \newcommand*{\nom}[3]{#2\nomenclature[#1]{#2}{#3}}
        \makenomenclature
    \fi
    \usepackage[draft=false,pageanchor]{hyperref}
    \usepackage[final]{graphicx}
\fi

\ifnum\draft=1
\usepackage{showlabels}
\fi
\usepackage [ lambda,advantage , operators , sets , adversary , landau , probability , notions , logic ,ff, mm,primitives , events , complexity , asymptotics , keys]{cryptocode}
\usepackage{doi}
\theoremstyle{plain}
\ifnum\llncs=0
    \ifnum\sigconf=0
        \newtheorem{theorem}{Theorem}
        \newtheorem{lemma}{Lemma}
        \newtheorem{corollary}{Corollary}
        \newtheorem{proposition}{Proposition}
        \newtheorem{definition}{Definition}

    \fi 
    
\fi    
\ifnum\llncs=0

      \newtheorem*{theorem*}{Theorem}
      \newtheorem*{lemma*}{Lemma}
      \newtheorem*{corollary*}{Corollary}
      \newtheorem*{proposition*}{Proposition}
      \newtheorem*{claim*}{Claim}

      \theoremstyle{remark}
      \newtheorem{remark}{Remark}
          \theoremstyle{plain}
      \newtheorem{openproblem}{Open Problem}
      \theoremstyle{plain}
      
      \else
          \spnewtheorem{openproblem}{Open Problem}{\bfseries}{\itshape}%llncs uses a spnewtheorem
        \fi

\ifnum\draft=1
    \newcommand{\authnote}[3]{{\color{#3} {\bf  #1:} #2}}
\else
    \newcommand{\authnote}[3]{}
\fi
\newcommand{\ket}[1]{\vert #1 \rangle}
\newcommand{\ketbra}[1]{\vert #1 \rangle \langle #1 \vert}

\newcommand{\tensor}{\otimes}

\DeclareMathOperator{\Tr}{Tr}

\ifnum\draft=1
    \linenumbers
\fi

\usepackage{algorithm,algorithmicx,algpseudocode}
\usepackage[normalem]{ulem}

\newcommand{\stamp}{
 {\mathchoice
  {\includegraphics[height=1.8ex]{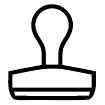}}
  {\includegraphics[height=1.8ex]{RubberStamp}}
  {\includegraphics[height=1.2ex]{RubberStamp}}
  {\includegraphics[height=0.9ex]{RubberStamp}}
 }
}
\newcommand{\tokengen}{\ensuremath{\mathsf{token\textit{-}gen}}}

\usepackage{tikz}
\usetikzlibrary{arrows,chains,matrix,positioning,scopes}
\makeatletter
\tikzset{join/.code=\tikzset{after node path={%
\ifx\tikzchainprevious\pgfutil@empty\else(\tikzchainprevious)%
edge[dashed, every join]#1(\tikzchaincurrent)\fi}}}
\makeatother
\tikzset{>=stealth', every on chain/.append style={join},
         every join/.style={->}}
\tikzstyle{labeled}=[execute at begin node=$\scriptstyle,
   execute at end node=$]
\ifnum\sigconf=0
 \ifnum\llncs=1
\usepackage[toc,page,titletoc]{appendix}

\usepackage[capitalise]{cleveref}

\crefname{appendix}{Supplement}{Supplements} 
\else

 \usepackage[capitalise]{cleveref} 

    \crefname{Game}{Game}{Games}
    \crefname{openproblem}{Open Problem}{Open Problems}
\fi
\usepackage{autonum}%must be loaded after cleveref. Sometimes causes bugs. The alternative is incomaptible with cleveref:
\expandafter\let\expandafter\savedflalignstar\csname flalign*\endcsname
\expandafter\let\expandafter\savedendflalignstar\csname endflalign*\endcsname
\AtBeginDocument{%
  \expandafter\let\csname flalign*\endcsname\savedflalignstar
  \expandafter\let\csname endflalign*\endcsname\savedendflalignstar
}

\newcommand{\ab}[1]{}  %%

\ifnum\masterthesis=1
    \usepackage{setspace}
    \usepackage{fancyhdr}
    \newcommand{\longunderline}{\underline{\ \ \ \ \ \ \ \ \ \ \ \ \ \ \ \ \ \ \ \ }\ }
    
    \newenvironment{changemargin}[2]{%
        \begin{list}{}{%
        \setlength{\topsep}{0pt}%
        \setlength{\leftmargin}{#1}%
        \setlength{\rightmargin}{#2}%
        \setlength{\listparindent}{\parindent}%
        \setlength{\itemindent}{\parindent}%
        \setlength{\parsep}{\parskip}%
        }%
    \item[]}{\end{list}}
\fi

\usepackage{tikz}
\usepackage{thm-restate}
\usepackage{multirow}%for the table
\newcounter{Game} 
\newenvironment{game}[1][htb]
    {
    \floatname{algorithm}{Game}% Update algorithm name
    \refstepcounter{Game}%to call the Game pointer. Currently the problem is that caption command by default calls the algorithm pointer. For the moment we need to use Game~~\ref{blah}
   \begin{algorithm}[#1]%
  }{\end{algorithm}}
\makeatother  
\usepackage{xspace}

\newcommand{\Noise}[1]{\ensuremath{\mathsf{Noise(#1)}}\xspace}
\newcommand{\tolerant}[1]{\ensuremath{#1\textit{-}\mathsf{noise}\textit{-}\mathsf{tolerant}}\xspace}
\newcommand{\oneunrest}[1]{\ensuremath{\CTMAC^{#1}_{\mathsf{1\textit{-}time}}}\xspace}
\newcommand{\Junrest}[1]{\ensuremath{\CTMAC^{#1}_{\mathsf{J\textit{-}times}}}\xspace}
\newcommand{\polyunrest}[1]{\ensuremath{\CTMAC^{#1}_{\mathsf{Poly\textit{-}times}}}\xspace}

\newcommand{\sigOrUnrest}[1]{\ensuremath{\widetilde{\CTMAC}^{#1}_{\mathsf{Poly\textit{-}times}}}\xspace}
\newcommand{\Bitspace}{\ensuremath{\mathsf{\{0,1\}}}\xspace}
\newcommand{\Stamp}{
 {\mathchoice
  {\includegraphics[height=1.8ex]{RubberStamp}}
  {\includegraphics[height=1.8ex]{RubberStamp}}
  {\includegraphics[height=1.2ex]{RubberStamp}}
  {\includegraphics[height=0.9ex]{RubberStamp}}
 }
}

\newcommand{\Miss}{\ensuremath{\mathsf{Miss}}\xspace}
\newcommand{\Bog}{\ensuremath{\mathsf{Bog}}\xspace}

\newcommand{\TMS}{\ensuremath{\mathsf{TMS}}\xspace}
\newcommand{\qa}{\ensuremath{\mathsf{q}}\xspace}

\newcommand{\Tmac}{\ensuremath{\mathsf{TMAC}}\xspace}
\newcommand{\Exly}{\ensuremath{\mathsf{Existentially}}\xspace}

\newcommand{\keygen}{\ensuremath{\mathsf{key\textit{-}gen}}\xspace}
\newcommand{\Vrfy}{\ensuremath{\mathsf{verify}}\xspace}

\newcommand{\QBitspace}{\ensuremath{\mathsf{\mathcal{BB}84}}\xspace}
\newcommand{\Sign}{\ensuremath{\mathsf{sign}}\xspace}
\newcommand{\Sabotage}{\ensuremath{\mathsf{Sabotage}}\xspace}

\newcommand{\QPT}{\ensuremath{\mathsf{QPT}}\xspace}
\newcommand{\str}{\ensuremath{\mathsf{Strong}}\xspace}

\newcommand{\NVR}{\ensuremath{\mathsf{\mathcal{VR}^\eta}}\xspace}
\newcommand{\NVRk}{\ensuremath{\mathsf{\mathcal{VR}^\eta_\mathnormal{k}}}\xspace}

\newcommand{\VR}{\ensuremath{\mathsf{\mathcal{VR}}}\xspace}

\newcommand{\PPT}{\ensuremath{\mathsf{PPT}}\xspace}
\newcommand{\Woof}{\ensuremath{\mathsf{UOWHF}}\xspace}

\newcommand{\Herm}{\ensuremath{\mathsf{Herm}}\xspace}
\newcommand{\Pos}{\ensuremath{\mathsf{Pos}}\xspace}
\newcommand{\Pd}{\ensuremath{\mathsf{Pd}}\xspace}
\newcommand{\cca}{\ensuremath{\mathsf{CCA}}\xspace}
\newcommand{\PrivKcca}{\ensuremath{\mathsf{PrivK^\cca}}\xspace}

\newcommand{\EncForge}{\ensuremath{\mathsf{Enc\textit{-}forge}}\xspace}
\newcommand{\Aenc}{\ensuremath{\mathsf{AENC}}\xspace}

\newcommand{\cma}{\ensuremath{\mathsf{CMA}}\xspace}
\newcommand{\cpa}{\ensuremath{\mathsf{CPA}}\xspace}
\newcommand{\Sec}{\ensuremath{\mathsf{Secret}}\xspace}
\newcommand{\ksig}{\ensuremath{\mathsf{ksig}}\xspace}

\newcommand{\Dec}{\ensuremath{\mathsf{dec}}\xspace}
\newcommand{\MM}{\ensuremath{\mathsf{PM}}\xspace}

\newcommand{\Unc}{\ensuremath{\mathsf{Unconditionally}}\xspace}

\newcommand{\OTL}{\ensuremath{\mathsf{OTL}}\xspace}
\newcommand{\OTO}{\ensuremath{\mathsf{OT1}}\xspace}
\newcommand{\OT}{\ensuremath{\mathsf{OT}}\xspace}
\newcommand{\OTM}{\ensuremath{\mathsf{OTM}}\xspace}
\newcommand{\wrap}{\ensuremath{\mathsf{wrap}}\xspace\xspace}

\newcommand{\UC}{\ensuremath{\mathsf{UC}}\xspace}
\renewcommand{\TM}{\ensuremath{\mathsf{TM}}\xspace}
\newcommand{\CTMAC}{\ensuremath{\mathsf{CTMAC}}\xspace}

\newcommand{\Cenc}{\ensuremath{\mathsf{CENC}}\xspace}
\newcommand{\Cmac}{\ensuremath{\mathsf{CMAC}}\xspace}
\newcommand{\SDP}{\ensuremath{\mathsf{SDP}}\xspace}
\newcommand{\Enc}{\ensuremath{\mathsf{enc}}\xspace}
\newcommand{\QECD}{\ensuremath{\mathsf{QECD}\xspace}}
\newcommand{\CVPRM}{\ensuremath{\mathsf{CVPRM}}\xspace}

\newcommand{\rand}{\ensuremath{\mathsf{rnd}}\xspace}

\newcommand{\Del}{\ensuremath{\mathsf{del}}\xspace}
\newcommand{\Badv}{\ensuremath{\mathsf{\mathcal{B}}}\xspace}
\newcommand{\Cadv}{\ensuremath{\mathsf{\mathcal{C}}}\xspace}
\newcommand{\Sadv}{\ensuremath{\mathsf{\mathcal{S}}}\xspace}
\newcommand{\Radv}{\ensuremath{\mathsf{\mathcal{R}}}\xspace}

\newcommand{\Tadv}{\ensuremath{\mathsf{\mathcal{T}}}\xspace}

\newcommand{\create}{\ensuremath{\mathsf{create}}\xspace}
\newcommand{\run}{\ensuremath{\mathsf{run}}\xspace}

\newcommand{\exWCerDel}{\ensuremath{\mathsf{WEAK\textit{-} DEL\textit{-}IND}}\xspace}

\newcommand{\mint}{\ensuremath{\mathsf{mint}}\xspace}
\newcommand{\Cons}{\ensuremath{\mathsf{Cons}}\xspace}

\newcommand{\MoneyForge}{\ensuremath{\mathsf{MONEY\textit{-}FORGE}}\xspace}

\newcommand{\Iset}{\ensuremath{\mathcal{I}}}

\newcommand{\UnfExp}{\ensuremath{\mathsf{Forge}}\xspace}
\newcommand{\OneToken}{\ensuremath{\mathsf{1\textit{-}\tokengen}}\xspace}
\newcommand{\JToken}{\ensuremath{\mathsf{j\textit{-}\tokengen}}\xspace}
\newcommand{\UnfDef}{\ensuremath{\mathsf{Unforgeable}}\xspace}
\newcommand{\SabDef}{\ensuremath{\mathsf{Unsabotagable}}\xspace}

\newcommand{\kaut}{\ensuremath{\mathsf{k_{\Aenc}}}\xspace}
\newcommand{\kenc}{\ensuremath{\mathsf{c}}\xspace}
\newcommand{\Fwrap}{\ensuremath{\mathcal{F}_\wrap}\xspace}
\newcommand{\FOTM}{\ensuremath{\mathcal{F}_\OTM}\xspace}

\usepackage{adjustbox} %for the diagram
\tikzset{
    pil/.style={
        ->,
       thick,
       shorten <=2pt,
       shorten >=2pt
    }
}
\usepackage{MnSymbol}
\makenomenclature

\begin{document}
\ifnum\masterthesis=0
    \title{Noise-Tolerant Quantum Tokens for MAC
    %\ifdraft{\\(working draft)}
    }
\fi

\ifnum\anon=0
    \ifnum\masterthesis=0
        \ifnum\sigconf=0
            \ifnum\cryptology=1
                \author{Amit Behera}
                \affil{Department of Computer Science, Ben-Gurion University of the Negev, Beersheba, Israel\\
                amitbehera1767@gmail.com}
                \author{Or Sattath}
                \affil{Department of Computer Science, Ben-Gurion University of the Negev, Beersheba, Israel\\
                sattath@post.bgu.ac.il}
                \author{Uriel Shinar}
                \affil{Department of Computer Science, Ben-Gurion University of the Negev, Beersheba, Israel\\
                shinaru@post.bgu.ac.il}
            \else
                \author[1]{Amit Behera}
                \author[1]{Or Sattath}
                \author[1]{Uriel Shinar}
                \affil[1]{Computer Science Department, Ben-Gurion University}
            \fi
        \else
            \author{Or Sattath}
            \affiliation{%
            \institution{Computer Science Department, Ben-Gurion University}
            \country{Israel}}
        \fi
    \fi
\else
    \ifnum\sigconf=0
        \author{}
        \ifnum\llncs=1
             \institute{}
        \fi

    \fi
\fi

\ifnum\masterthesis=0
    \ifnum\sigconf=0
        \maketitle
    \fi
\fi

\ifnum\masterthesis=1
    
    \begin{titlepage}
        \centering
        { Ben-Gurion University of the Negev}
        
        {The Faculty of Natural Sciences}
        
        {\small The Department of Computer Science}
        
        \vspace{2cm}
        
        {\Large \bfseries Template of Thesis}
        
        \vspace{2cm}
        
        {\small Thesis submitted in partial fulfillment of the requirements for the Master of Sciences degree}
        
        \vspace{1cm}
        
        {\bfseries Name of Student}
        
        {Under the supervision of Dr. Or Sattath}
        
        \vspace{2cm}
        
        \today
    \end{titlepage}
    
    \begin{titlepage}
        \centering
        { Ben-Gurion University of the Negev}
        
        {The Faculty of Natural Sciences}
        
        {\small The Department of Computer Science}
        
        \vspace{2cm}
        
        {\Large \bfseries Template of Thesis}
        
        \vspace{2cm}
        
        {\small Thesis submitted in partial fulfillment of the requirements for the Master of Sciences degree}
        
        \vspace{1cm}
        
        {\bfseries Name of Student}
        
        {Under the supervision of Dr. Or Sattath}
        
        \vspace{1cm}
        
        {\small Signature of student: \longunderline Date: \longunderline}
        
        \vspace{0.5cm}
        
        {\small Signature of supervisor: \longunderline Date: \longunderline}
        
        \vspace{0.5cm}
        
        \begin{changemargin}{-1cm}{-1cm}
        \centering
            {\small Signature of the committee for graduate studies: \longunderline Date: \longunderline}
        \end{changemargin}
        \vspace{2cm}
        
        \today
    \end{titlepage}
\fi

\ifnum\masterthesis=1
    \pagenumbering{roman}
    \begin{center}
        {\large \bfseries Title of Thesis}
        
        \vspace{0.5cm}
        
        {\bfseries Name of Student}
        
        \vspace{0.5cm}
        
        Thesis submitted in partial fulfillment of the requirements for the Master of Sciences degree
        
        \vspace{0.25cm}
        
        {Ben-Gurion University of the Negev}
        
        \vspace{0.25cm}
        
        \today
        
        \vspace{2cm}
        
        {\bfseries Abstract}
        
        \vspace{0.5cm}
    \end{center}
\else    
    \begin{abstract}
\fi
Message Authentication Code, or MAC, is a well-studied cryptographic primitive that is used to authenticate communication between two parties who share a secret key. A Tokenized MAC or TMAC  is a related cryptographic primitive, introduced by Ben-David \& Sattath (QCrypt'17), which allows limited signing authority to be delegated to third parties via the use of single-use quantum signing tokens. These tokens can be issued using the secret key, such that each token can be used to sign \emph{at most} one document.

We provide an elementary construction for TMAC based on BB84 states. Our construction can tolerate up to $14\%$ noise, making it the first noise-tolerant TMAC construction. 
The simplicity of the quantum states required for our construction, combined with its noise tolerance, make it practically more feasible than the previous TMAC construction.  

The TMAC presented is existentially unforgeable against adversaries with signing and verification oracles (i.e., it is analogous to EUF-CMA security for MAC), assuming that post-quantum one-way functions exist. 

\ifnum\masterthesis=0
    \end{abstract}
\fi

\ifnum\sigconf=1
    \keywords{}
    \maketitle
\fi

\ifnum\masterthesis=1
    \pagebreak
    \pagenumbering{arabic}
\fi

\ifnum\masterthesis=1
    \subsection*{Acknowledgments}
\fi

\ifnum\masterthesis=1
    \setcounter{tocdepth}{3}
    \tableofcontents
    
    \listoffigures
\fi
\begin{flushright}
 In memory of Stephen Wiesner, 1942–-2021.
\end{flushright}
\ifnum\llncs=0
 \tableofcontents
\fi
%-----------------------------------------------------------------------------%

\section{Introduction} % (fold)
\label{sec:introduction}
%-----------------------------------------------------------------------------%

The discovery of Wiesner's quantum money protocol~\cite{Wie83} initiated the study of quantum cryptographic primitives based on the no-cloning theorem. Ben-David and Sattath~\cite{BS16a} introduced one such primitive called a Tokenized Private Digital Signature scheme, or alternatively, Tokenized Message Authentication Code $(\Tmac)$.
Traditional $\mac$ schemes allow Alice to communicate with Bob in an authenticated manner by sharing a secret key. She can sign a document and send it along with its signature to Bob. Bob would then use the shared secret key to verify the signature and could respond back in the same manner. %Despite the multitude of applications of $\mac$s in cryptography, there can be scenarios where they are not enough.  
The motivation for $\Tmac$ arises from the following scenario in which a $\mac$ scheme is not enough. Suppose it so happens that Alice will be temporarily absent for a short time, and she would like Charlie to sign a few urgent documents on her behalf. The na\"ive thing that Alice can do is to send her key to Charlie. However, there are two main drawbacks with this approach: (i) The secret key would allow Charlie to forever sign an arbitrary number of documents on Alice's behalf, and (ii) even if Charlie is completely trusted, Charlie's computer can get hacked --- and the secret key copied --- without him even noticing. These two issues can be circumvented using a $\Tmac$ scheme. A $\Tmac$ consists of four algorithms: $\keygen$, $\tokengen$, $\Sign$, and $\Vrfy$. Alice can run the algorithm $\keygen$ to generate a \emph{classical} secret key $k$, which can then be shared with Bob, as in the case of $\mac$. Next, she can run $\tokengen_k$ to generate a \emph{quantum} token, denoted\footnote{The icon $\Stamp$ represents a rubber stamp.} as \nom{a_T}{$\ket{\Stamp}$}{ A quantum signing token}, that she would then send to Charlie. If needed, she could repeat this procedure \nom{a_r}{$r$}{The number of signing tokens given to the adversary} times to allow Charlie to sign $r$ documents. Charlie can then use the algorithm $\Sign_{\ket{\Stamp}}(m)$ to generate a \emph{classical} signature $\sigma$ for a \emph{classical} message $m$. Bob can verify the authenticity of the document Charlie sends by running $\Vrfy_k(m,\sigma)$.

The security guarantee of $\Tmac$ ensures that Charlie cannot produce $r+1$ signed documents (signature, document pair) given the $r$ signing tokens he was provided with. Hence, even if Charlie's computer is hacked, the damage would be limited to only a fixed number of documents, depending on the number of tokens Charlie had been given. One may wonder why cannot Charlie run the signing algorithm once with $m_1$, and then run the signing algorithm again with another document $m_2$. This is because the signing algorithm applies a (destructive) measurement on multiple qubits, and therefore, the token is consumed during the signing.  
 
Alice can even confirm after her return to the office, that Charlie has not kept any of the tokens or signed any unapproved documents. If Alice gave $r$ tokens to Charlie, and he claims to have used $n$ tokens for signing, then Alice would ask for those $n$ signed documents, and additionally, send Charlie $r-n$ fresh random documents to sign, in order to consume the remaining tokens. She would then verify all the $r$ signatures on those $r$ documents for confirmation. This property is known as revocability (see~\cite{BS16a} for a more rigorous definition). Note that the revocation procedure mentioned above, also allows Alice to learn the $n$ documents that Charlie signed when she was away.

\paragraph{Notions of Security.} \label{pg:notions_of_security}
The security notion that we consider for $\Tmac$ schemes is inspired by the $n$ to $n+1$ unforgeability of quantum money schemes and the $\cma$ security of vanilla $\mac$ schemes. This notion asserts that a quantum polynomial time adversary who is given $n\in \poly$ many tokens and has classical access\footnote{Classical here means that the adversary can only send classical documents and signatures to the oracles and, hence, cannot query in superposition.} to both a signing oracle and a verification oracle cannot produce valid signatures for $n+1$ distinct documents that have not been queried to the signing oracle (except with negligible probability). We say such a scheme is $\UnfDef^{\tokengen,\widetilde{\Sign},\Vrfy}$, where the superscripts denote the oracles available to the adversary.
 
One might also consider a strengthened notion of security, such as the one analogous to strong unforgeability (i.e., a \emph{different} signature for a document which was signed by the signing oracle also counts as forgery), as well as quantum access to the verification oracle.
Our construction satisfies neither of the two stronger security notions mentioned above, as shown in \cref{sec:Drawbacks Of the Conjugate Tmac}.
An extensive discussion on the different notions of security for $\Tmac$ is given in \cref{sec: notions of security}.

\paragraph{Our Contributions.} We consider an IID noise model $\Noise{p}$ for some $0\leq p\leq 1$, meaning the qubit at each coordinate gets corrupted\footnote{We say a qubit gets corrupted to mean that an arbitrary dimension-preserving CPTP map acts on it. Since we deal with BB84 states in our construction, the map in our case which (negatively) affects verification the most is the Pauli operator $Y$.} with a probability of $p$, and remains undeterred with a probability of $(1-p)$, independent of other qubits. A $\Tmac$ is $\tolerant{\delta}$ for some $0\leq \delta \leq 1$, if correctness holds for the scheme up to negligible error in the noise model $\Noise{\delta}$ (\cref{def:noise-model}). Our main result is as follows.
\nomenclature[A]{$\alpha$}{Approximately $\cos^2(\frac{\pi}{8})$}
\begin{restatable}{theorem}{MainResult}
Let\footnote {The value of $\alpha$ was computed numerically, the numerical error was relatively large and differed between engines (of magnitude $10^{-5}$ on the "sedumi" solver and $10^{-3}$ on the "SDPT3" solver), leaving some room for doubt if $\alpha$ is exactly $\cos^2(\frac{\pi}{8})$, or a slightly larger value. Regardless, the exact value of $\alpha$ does not affect the results of this work, besides the precise amount of noise tolerance.} $\alpha\approx \cos^2(\pi/8)$. 
Assuming post-quantum one-way functions exist, for every constant $\delta < 1-\alpha$, there exists an \\$\UnfDef^{\tokengen,\Vrfy,\widetilde{\Sign}}$  $\tolerant{\delta}$ $\Tmac$ (see \cref{def:Unforgeability,def:noise tolerance}) scheme $\Pi_\delta$  (see \cref{alg:noise_tolerant Conjugate Tmac,prp:oneRtofb}) based on conjugate coding states.
\label{thm:main_result}
\end{restatable}
The formal proof is given in \cref{sec:main_results}, on \cpageref{pf:thm:main_result}.

A previous construction for $\Tmac$~\cite{BS16a} achieves the same security under the same assumption.\footnote{\cite{BS16a} actually assumes the slightly stronger assumption of collision resistant hash functions. However using a universal one-way hash function (\Woof) instead, which is equivalent to one-way functions, would achieve the same result, as shown in \cref{sec:Security proof for full blown}.} However, the construction is not known to be noise-tolerant, and it also requires entangled states as the tokens. 

In comparison, our construction is noise-tolerant and only requires simple tensor product states, which improves the practicality of our construction.
Hence, the implementation of our $\Tmac$ scheme requires a non-perfect quantum channel capable of transmitting a BB84 state with less than a threshold percentage of error, and a quantum memory for the third party to store the BB84 token states for the duration it is granted to sign. Transmission of BB84 states or Wiesner money states for long distances with relatively low noise has already been demonstrated in quantum key distribution experiments~\cite{HRP+06,KLH+15}. Hence, the main practical challenge that remains is that of a quantum memory capable of storing BB84 states for reasonable periods of time. The main vectors used for transmission of quantum information are photons that are difficult to store with high fidelity, and the qubit lifespan is short even on quantum computers.  However, for tokenized private signatures, a short-term storage may very well be enough, as the tokens are expected to be temporary by nature.

As an application for $\Tmac$, we also show that any $\Tmac$ scheme can be used to construct a classically verifiable quantum money scheme, see \cref{sec:quantum money}. The verification can even be made non-interactive, at the cost of requiring temporary memory dependence.\footnote{This can be viewed as a small database that needs to be maintained for a short time frame, at the end of which the database can be deleted.} To the best of our knowledge, none of the known private quantum money schemes based on similarly simple states has non-interactive verification. We remark that these are the same constructions used in \cite{BS16a} in the context of public tokenized signature and public quantum money.  The fact that a secure $\Tmac$ implies secure private quantum money, combined with a recent result~\cite{Aar20} which shows that unconditionally secure private quantum money schemes do not exist, implies that no  $\Tmac$ scheme can be unconditionally unforgeable if the adversary is allowed to request a polynomial number of tokens (see \cref{thm: impossibility of unconditionally secure TMAC scheme.}). Since the existence of post-quantum one-way functions is one of the weakest possible computational assumptions in cryptography, this shows that our construction is close to optimal in terms of the assumptions required. 

Another application of our result is that it resolves an open problem of~\cite{BGZ21} regarding the construction of one-time memories in the stateless hardware model. This is a non-standard model in which it is possible to implement any program with classical input, as a stateless black-box that can be queried classically. One-time memories are an ideal modeling of $1$ out of $2$ non-interactive oblivious transfer, which are impossible to realize in the standard model. In~\cite{BGZ21}, the authors presented a conjugate coding-based construction of one-time memories in the stateless hardware model, and conjectured it to be a universally composable realization of one-time memories against unbounded malicious receivers in that model. In~\cite{BGZ21}, the authors prove this conjecture partially; they require the added assumption that the receiver makes fewer than $c\secpar$ (classical) queries to the stateless hardware where $c<0.114$ and $\secpar$ is the security parameter. In \cref{sec:One-Time-Memories from Stateless Hardware}, we show that our results imply that the result above can be extended to any polynomial number of classical queries.\footnote{One cannot hope for an unbounded number of queries as the entire input-output behavior of the stateless hardware can be extracted. In addition, as covered in~\cite{BGZ21}, quantum queries to the stateless hardware also result in insecure schemes.} Moreover, our work proves that this construction of one-time memories remains secure even in a noisy setting. We note that~\cite{CGL19} also presented a construction of one-time memories from stateless hardware, with a full security proof, but it was previously unknown whether this could be done with a conjugate coding-based scheme. 

%%%%%%%%%%%%%%%%%%%%%%%%%%%%%%%%%%%%%%%%%%%%%%%
\paragraph{Construction.}\label{pg:construction}
%%%%%%%%%%%%%%%%%%%%%%%%%%%%%%%%%%%%%%%%%%%%%%%
Our construction has two main steps. First, we construct a one-restricted $\Tmac$, meaning only single-bit documents can be signed. The construction is unconditionally unforgeable against single-token attacks in the following sense. Given a single token and access to a verification oracle, which can be queried only with classical strings and only a polynomial number of times,  a computationally unbounded adversary that is given one signing token, cannot produce valid signatures for both $0$ and $1$ (except with negligible probability). We say such a scheme is $\Unc$ $\UnfDef^{\OneToken,\Vrfy}$, where $\OneToken$ represents the single token available to the adversary, and $\Vrfy$ the verification oracle.

Assuming post-quantum one-way functions exist, any such single bit $\Tmac$ that is $\Unc$
$\UnfDef^{\OneToken,\Vrfy}$ can be lifted in a noise-tolerance preserving manner using standard techniques similar to the ones used in~\cite{BS16a} to a scheme that can sign documents of any length, and that is $\UnfDef^{\tokengen,\widetilde{\Sign},\Vrfy}$\footnote{Unforgeability against Quantum Polynomial Time  adversaries suffices for this lift as well.} (see ``Notions of security'', \cpageref{pg:notions_of_security}). The approach is reminiscent of the mini-scheme to full-scheme lift used for quantum money \cite{AC13}.
The above-mentioned lifting is the only place where we need (standard) computational hardness assumptions in our entire construction. This lifting is described in detail in \cref{sec:Security proof for full blown}.
    
The main challenge in our work is to construct a single-bit $\Tmac$ that is  $\Unc$ $\UnfDef^{\OneToken,\Vrfy}$. Our construction draws inspiration from the classically verifiable variant of Wiesner's quantum money~\cite{MVW12} (see also~\cite{PYJ+12}). In this variant of Wiesner's quantum money, the bill is a tensor product of random conjugate coding states (also known as BB84 states), such as $\ket{1}\tensor\ket{-}\tensor\ket{0}\tensor\ket{-}\tensor\ldots\tensor\ket{0}$, the classical representation of which is kept secret by the bank. The bank sends a uniformly random challenge bit per qubit when a customer approaches the bank to verify the money. An honest customer should measure each qubit either in the computational basis or in the Hadamard basis, according to the challenge received, and should send the result as a proof. The bank accepts if the result is consistent with the corresponding qubit at the coordinates where the challenge string agrees with the secret string representing the conjugate state. The money scheme can be made noise-tolerant by relaxing the verification to accept even if the consistency check of the result with respect to the challenge fails at a small fraction of the coordinates. The same security guarantees hold for the noise-tolerant variant, as shown in~\cite{PYJ+12,MVW12}.

In our construction for the single-bit noise-sensitive $\Tmac$ scheme (\cref{alg:Conjugate Tmac}), the quantum money state for the above-mentioned scheme serves as the token, and its classical representation serves as the secret key. The challenges are no longer random, but instead correspond to the (single-bit) document to be signed, i.e., an honest signer measures all the token's qubits in the computational basis in order to sign $0$, or measures them in the Hadamard basis in order to sign $1$. In both cases, the measurement outcome is the signature. Verification of an alleged signature for the document $0$ (respectively, $1$) is done by checking if the signature string is consistent with the secret key at the coordinates that had computational basis (respectively, Hadamard basis) states.
The construction is made noise-tolerant in a way similar to the quantum money scheme, i.e., we relax the verification to accept even if the consistency check of the signature string for the document with the secret key fail at a small fraction of the coordinates. 

\ifnum\llncs=0
Forgery against the noise-sensitive (respectively, noise-tolerant) variant of the $\Tmac$ scheme --- i.e., producing valid signatures of both $0$ and $1$ using a single token --- corresponds to passing two verifications in the noise-sensitive (respectively, noise-tolerant) variant of the money scheme using a single money state. However, there is a caveat that unlike the money scheme, the challenges are chosen by the adversary. We prove that despite the mentioned caveat, the probability to forge a signature for both $0$ and $1$ decreases exponentially in the number of qubits in both the noise-sensitive and noise-tolerant variants. Moreover, the security guarantee needs to hold even against adversaries that have classical access to a verification oracle (which is not the case for Wiesner's quantum money). For these reasons, we cannot reduce the $\Tmac$ unforgeability to the quantum money unforgeability, and provide a complete new proof instead.
\fi
\ifnum\llncs=0
 The complete scheme is shown in \cref{alg:Full construction}.

 \begin{algorithm} %\label{alg:TOM}
    \caption{${\CTMAC^{0.07}}$ - A $14\%$ noise-tolerant $\Tmac$ scheme.}
    \label{alg:Full construction}
    \textbf{Assumes:} \nom{Enc_Aut}{$\Aenc$}{Authenticated encryption scheme} is a post-quantum classical-queries strong authenticated encryption scheme, $\{h_r : \Bitspace^\ast \rightarrow \Bitspace^{\ell'(|r|)}\}_{r\in\Bitspace^\ast}$ is a universal one-way hash function family  with indexing function $I$, $\ell(\secpar)=\ell'(\secpar)+\secpar$.
    \begin{algorithmic}[1] % The number tells where the line numbering should start
        \Procedure{$\keygen$}{$1^\secpar$}
            \State \textbf{Return} $ \Aenc.\keygen(1^\secpar)$.
        \EndProcedure
    \end{algorithmic}
    \begin{algorithmic}[1] % The number tells where the line numbering should start
        \Procedure{$\tokengen_k$}{}
            \For  {$j=1,\ldots,\ell(\secpar)$}
            \State $a^i,b^j\sample\Bitspace^\secpar$.
            \EndFor
            \State Denote $\kappa=(a^1,b^1),\ldots,(a^{\ell(\secpar)},b^{\ell(\secpar)})$.
         \State    Compute $\ket{\widetilde{\Stamp}_j}=H^{b^j}\ket{a^j}$.
           \State \textbf{Return} $(\otimes^{\ell(\secpar)}_{j=1}\ket{\widetilde{\Stamp}_j}, \Aenc.\enc_{k}(\kappa))$.
        \EndProcedure

    \end{algorithmic}
	\begin{algorithmic}[1] % The number tells where the line numbering should start
        \Procedure{$\Sign_{\ket{\Stamp}}$}{$m$}
            \State Interpret $\ket{\Stamp}$ as $(\otimes^{\ell(\secpar)}_{j=1}\ket{\widetilde{\Stamp}_j},e)$.
            \State $r_1\gets I(\secpar) $, $r_2\sample \Bitspace^\secpar$. Let $m'\equiv (h_{r_1}(m||r_2)||r_1)$
            \For  {$j=1,\ldots,\ell(\secpar)$}
            \State Let $m'_j$ denote the $j^{th}$ bit of $m'$ and measure $({H^{m'_j}})^{\otimes\secpar}\ket{\widetilde{\stamp}_j}$ in the computational basis to obtain the classical string $s^j$.
            \EndFor
            \State \textbf{Return:}$((s^1,s^2,\ldots,s^{\ell(\secpar)}),e,r_1,r_2)$.
        \EndProcedure

    \end{algorithmic}
  	\begin{algorithmic}[1] % The number tells where the line numbering should start
    \Procedure{$\Vrfy_k$}{$m,\sigma$}
            \State Interpret $\sigma$ as $((s^1,s^2,\ldots,s^{\ell(\secpar)}),e,r_1,r_2)$.
            \State Compute $\Aenc.\Dec_{k}(e)$ to obtain $\kappa\equiv(a^1,b^1),\ldots,(a^{\ell(\secpar)},b^{\ell(\secpar)})$ and denote $m'\equiv (h_{r_1}(m||r_2)||r_1)$, let $m'_j$ be the $j^{th}$ bit of $m'$.
            \For{$j=1,\ldots,\ell(\secpar)$}
                \State Let $a^j_i,b^j_i,s^j_i$ denote the $i^{th}$ bit of $a^j,b^j,s^j$ respectively.
               \State  Construct the subset $\Cons_{m'_j}=\{i\in[\secpar]\mid b^j_i=m'_j\}$.
                \State Construct the subset
                $\Miss_{m'_j,s^j}=\{i\in \Cons_{m'_j}\mid s^j_i\neq a^j_i\}$.
                \If{$\abs{\Miss_{m'_j,s^j}}>0.07\secpar$}
                    \State \textbf{Return} $0$.
                \EndIf
            \EndFor
            \State \textbf{Return} $1$.
    \EndProcedure
    \end{algorithmic}
\end{algorithm}
 \fi
 
\ifnum\llncs=0 
\paragraph{Proof Techniques.}
We first give the proof sketch for the security of the noise-sensitive scheme, and then explain how to reduce the security of the noise-tolerant scheme to that of the noise-sensitive one.

First, we prove unforgeability against single-token attacks for the noise-sensitive one-restricted $\Tmac$ scheme (see the paragraph ``Construction'' on \cpageref{pg:construction}). 
The main challenge in the proof is to tackle the verification oracle, and to understand what information the adversary can procure from the verification oracle. In order to do so, we define a new game for our specific scheme where we strengthen the adversary by replacing the verification oracle with a stronger oracle. This oracle not only answers verification queries but also outputs a part of the secret key after a successful query for a document. The part of the secret key models all the information that the adversary might learn using the verification oracle.
Let $m$ be a document for which a successful query was made in the game. We argue that in such a game, if a successful query is made with respect to a document $m$, then it is redundant to make verification queries corresponding to $m$, thereafter. This is because the adversary can simulate such queries on her own, using the part of the secret key that was received after the first successful query for $m$.

In the next step, an adversary in the new game that we described in the last paragraph is reduced to an adversary against a Quantum Encryption with Certified Deletion scheme. In a Quantum Encryption with Certified Deletion ($\QECD$) scheme (\cite{BI20}, see also \cref{def:qunatum encryption with certified deletion}), a user can provide proof of deletion for a quantum cipher-text such that the proof can be verified using a separate algorithm. The security guarantee is that if the adversary submits a proof of deletion for a cipher-text that passes verification, then they could not learn the message later, even if the secret key is provided. 

We map any adversary against $\Tmac$ in the intermediate game described in the last paragraph to an adversary against a $\QECD$ scheme in a weaker security game that is a relaxation of the certified deletion game in~\cite{BI20} (see \cref{sec: reduction to weak certified deletion}). We should point out that this is a scheme-specific reduction, and we do not know if it can be made generic.

Lastly, standard semi-definite programming techniques are used to provide a bound on the success probability of an adversary in the weak certified deletion game against the $\QECD$ scheme mentioned above--see \cref{sec: analysis of weak certified deletion scheme}.

The proof of unforgeability for the noise-tolerant variant (see the paragraph ``Construction'' on \cpageref{pg:construction})  follows along similar lines. We consider an intermediate scheme-specific game for the noise-tolerant scheme with a strengthened adversary, similar to what we did for the noise-sensitive scheme. In the next step, however, instead of mapping the strengthened adversary to an adversary in a certified deletion game, the adversary is mapped to a strengthened adversary in the corresponding game\footnote{This is the same intermediate game that we discussed above in the analysis of the noise-sensitive scheme.} for the noise-sensitive scheme with a smaller security parameter. In our analysis for the noise-sensitive scheme, we have shown that the winning probability of the strengthened adversary against the noise-sensitive scheme in the later game is exponentially small. The reduction then gives an exponentially small bound on the winning probability of the strengthened adversary against the noise-tolerant scheme.
 
 \fi

 \paragraph{Related Works.}\label{pg:Related Works}
 In~\cite{BS16a}, the authors also introduced the public variant of tokenized MACs called Public Tokenized Digital Signatures that allow anyone to verify the validity of a signature using a public key. The authors provide a construction for the private variant and a candidate construction for the public variant based on hidden sub-spaces that were originally used by Aaronson and Christiano to construct public quantum money~\cite{AC13}. A variant of the construction for public tokenized digital signature construction in~\cite{BS16a} was proven to be secure based on post-quantum  Indistinguishability Obfuscation (IO) in~\cite{CLLZ21}. The main advantage of the private scheme in~\cite{BS16a}, as well as the public variant in~\cite{CLLZ21}, over our construction is that it is secure even against adversaries performing superposition queries to the verification oracle. Another advantage is that the $\Tmac$ scheme in~\cite{BS16a} is a strongly unforgeable $\Tmac$, i.e., given a token, the adversary cannot even produce two different signatures for the same document. The analysis of the construction in~\cite{BS16a} is based on query-complexity theoretic lower bounds, whereas the analysis of our construction involves techniques based on semi-definite programming. 
 However,~\cite{BS16a}  does not achieve noise tolerance. 

$\Tmac$ is an enhanced version of the much studied cryptographic primitive called Message Authentication Codes ($\mac$) in terms of functionality. 
There are two kinds of $\mac$ schemes that are of interest to us: post-quantum $\mac$s for classical documents, and $\mac$s for quantum documents. There have been quite a few works~\cite{BZ13,GLL+16,SY17} on post-quantum $\mac$s, starting with the work of Boneh and Zhandry~\cite{BZ13}. Among other results, the authors in~\cite{BZ13} showed that quantum secure Pseudorandom functions imply an existentially unforgeable $\mac$ that is secure against quantum-chosen message attacks, i.e., even against quantum adversaries who can query the $\mac$ on superposition of documents. The notion of authentication of quantum data was first studied by Barnum et al. in~\cite{BCG+02}. Later, there were subsequent works including~\cite{BW16,AM17,GYZ17} that strengthened the security definitions and the constructions achieving them, although these results are less relevant for this work since we concentrate on the authentication of \emph{classical} documents.

Tokenized digital signatures are not the only cryptographic primitives that achieve the task of revocable signature delegation. 
One such primitive that achieves such delegation is called the ``one-shot signature'', introduced recently by Amos et al.~\cite{AGK+20}. 
In one-shot signatures, everyone can generate a quantum signing token and a classical verification key that can verify alleged signatures produced by the token.
The security guarantees that one cannot produce two signatures that pass verification with the same verification key. One-shot signatures allow an owner to delegate signing authority to a signer only using classical communication.\footnote{The owner can sign a verification key $vk$ of the signer, using a $\mac$ or a digital signature, and send the signature of $vk$ to the signer. The signer then can sign one document using its token and append the signature of $vk$ that he received to it. 
The signature of $vk$ made by the owner validates all signatures that pass verification with respect to $vk$. The security offered by the one-shot signatures restricts the signer to only produce one signature that passes verification with respect to $vk$.}
This is the main advantage of  one-shot signatures over $\Tmac$ since a $\Tmac$ requires the owner to send the quantum tokens to the signer over a quantum channel.
However, the obvious limitation is that one-shot signatures lack efficient construction. The construction given in~\cite{AGK+20} is based on an oracle, for which we have no efficient instantiation as of yet. In contrast, tokenized digital signatures exist under standard assumptions. 
Another related notion is the bolt-to-certificate transformation introduced in~\cite{CS20}. In~\cite{CS20}, the authors used Zhandry's quantum lightning~\cite{Zha21} to construct such a transformation.  
Note that the lightning bolts are quantum states that are similar to tokens for digital signatures in two ways. First, both of these are unclonable, i.e., it is not possible to make two bolts or two tokens from a single bolt or token. 
Second, they are revocable; bolts can be converted to certificates (as shown in~\cite{CS20}), and tokens can be spent by signing a bit. However, the main difference between the two is that unlike lightning bolts, there are at least two revocation procedures for tokens: one can either sign the document with $0$ or $1$ in order to spend it. Note that in order to use the bolts for signing bits, it is crucial to have two different revocation procedures: one for signing the bit $0$ and the other for $1$. %However, we could not find a way to employ two procedures for the quantum lightning-to-certificate process.
However, lightning-to-certificate provides only a single revocation method. Nevertheless, the transformation of bolt-to-certificate has been proven useful in other areas such as in the bitcoin scalability problem~\cite{CS20} and for semi-quantum money~\cite{RS20}. There is no construction of quantum lightning based on standard assumption, and the security of the existing scheme~\cite{Zha21} was put into question by Roberts~\cite{Rob21}.

 $\Tmac$s also imply other important cryptographic primitives that are well studied in the literature. One such example is unforgeable private-key quantum money~\cite{Wie83,AC13}, which we discuss in more detail in \cref{sec:quantum money}. Similarly, it is also possible to construct unforgeable public quantum money~\cite{AC13,Zha21,FGH+12} from public tokenized digital signatures (as shown in~\cite{BS16a}), which is a stronger and much harder primitive to construct. 
 
 Another application of $\Tmac$ is that of disposable cryptographic backdoors~\cite{CGL19}. For a cryptographic primitive such as an encryption scheme, the task is to give one-time backdoor access to secret information, such as the message hidden under the cipher, in case of an encryption. This task is achieved by defining a disposable backdoor variant of encryption, where the key-generation algorithm additionally outputs a quantum token or backdoor that can be used only once to learn the message from the cipher. In theory, the above variant can be constructed using exotic primitives called one-time programs~\cite{GKR08}, which are programs that can be executed once and then they become useless. However, these primitives cannot exist in the standard model, even in a quantum setting, see~\cite{GKR08,BGS13}. In~\cite{CGL19}, the authors show how to construct one-time programs from $\Tmac$s, relative to classical stateless hardware. It should be noted that all queries to the oracle in the stateless hardware model are classical, and hence the reduction in~\cite{CGL19} only requires the $\Tmac$ scheme to be unforgeable with respect to classical queries. Unforgeability against superposition queries for the underlying $\Tmac$ does not provide any additional utility. %Hence, the reduction in~\cite{CGL19} only requires the $\Tmac$ scheme to be unforgeable with respect to classical queries, which is precisely the security notion we use in this work. 
 This is yet another scenario where the unforgeability definition considered in our work (i.e., unforgeability with respect to classical queries), despite not being the strongest security notion, is the correct notion to consider under the circumstances.  
 
 The same result regarding one-time programs was considered in~\cite{BGZ21}, in which the authors directly use a conjugate coding state-based construction to achieve one-time programs. However, only a partial result regarding the security of the construction was provided. In \cref{sec:One-Time-Memories from Stateless Hardware}, the results of this work are used to prove a conjecture in~\cite{BGZ21}, and to strengthen their result.
 
 \paragraph{Organization.}
 \ifnum\llncs=0
We start with some preliminaries in \cref{sec:Preliminaries and Notations}, where we also discuss the definition, correctness, and security notions of $\Tmac$. An overview of the main result is given in \cref{sec:main_results}. The construction for a $1$-bit $\Tmac$ scheme that is unforgeable against single token attacks with verification oracle, $\CTMAC^\eta$ is given in \cref{sec:candidate_construction}, and its unforgeability is proved in \cref{sec:Security of Conjugate TMAC}.  \cref{sec:applications} covers applications of our result: a solution to an unsolved problem in~\cite{BGZ21} regarding a construction of one-time memories from stateless hardware, and a construction of private quantum money from $\Tmac$. In \cref{sec:open questions}, we discuss some open questions and future directions. The expansion of $\CTMAC^\eta$ to a full-blown scheme is covered in \cref{sec:Expansion To a Full-Blown Scheme}. In \cref{sec: analysis of weak certified deletion scheme}, we provide a detailed security analysis of the two $\QECD$ schemes used in the proof, which is left out of the main text. \cref{sec:Drawbacks Of the Conjugate Tmac} discusses drawbacks of $\CTMAC^\eta$,  \cref{sec:Comparison with Vanilla Unforgeability} shows that any $\Tmac$ satisfying a standard notion of unforgeability can also be used as a classical $\mac$ and \cref{sec:Fixed J Secure Tokens} discusses a lift to a (length-restricted) $\Tmac$ scheme which is unconditionally unforgeable for a fixed number of tokens given to the adversary. A nomenclature appears on the last page (\cpageref{pg:nomenclature}).
\else
We start with some preliminaries in \cref{sec:Preliminaries and Notations}, where we also discuss the definition, correctness, and security notions of $\Tmac$. An overview of the main result is given in \cref{sec:main_results}. The construction for a $1$-bit $\Tmac$ scheme that is unforgeable against single token attacks with verification oracle $\Tmac$, $\CTMAC^\eta$ is given in \cref{sec:candidate_construction}, and its unforgeability is proved in \cref{sec:Security of Conjugate TMAC}.   In \cref{sec:open questions}, we discuss some open questions and future directions. \cref{sec:applications} covers some applications of our result: a solution to an unsolved problem in~\cite{BGZ21} regarding a construction of one-time memories from stateless hardware, and a construction of private quantum money from $\Tmac$. In \cref{sec:extra proofs}, we cover some proofs left outside of the main text. The expansion of $\CTMAC^\eta$ to a full-blown scheme is covered in \cref{sec:Expansion To a Full-Blown Scheme}, and \cref{sec:Drawbacks Of the Conjugate Tmac} discusses drawbacks of $\CTMAC^\eta$,  \cref{sec:Comparison with Vanilla Unforgeability} shows that any $\Tmac$ satisfying a standard notion of unforgeability can also be used as a classical $\mac$. A nomenclature appears on the last page (\cpageref{pg:nomenclature}).
\fi
%-----------------------------------------------------------------------------%
\section{Notations and Definitions}\label{sec:Preliminaries and Notations}
%-----------------------------------------------------------------------------%
%-----------------------------------------------------------------------------%
\subsection{Notations}\label{sec:notations}
%-----------------------------------------------------------------------------%
We assume that the reader is familiar with classical cryptography (see~\cite{KL14,Gol04}), as well as quantum computing (see~\cite{NC11}).

The operation of the Hadamard gate on a single qubit is defined by the following matrix:

For a string $b\in \Bitspace^n$, we define $H^b$ as the $n$-qubit operator $H^{b_1}\otimes\cdots\otimes H^{b_n}$, where \ifnum\llncs=1 $H$ is the Hadamard gate.\else

\[H\equiv\frac{1}{\sqrt{2}}\begin{pmatrix}
1 &1\\
1 &-1
\end{pmatrix}.\]
\fi \nomenclature[H]{$H$}{Hadamard gate} 
We use the notation \nom{brakN}{$[n]$}{The set $\{1,\ldots,n\}$} to denote the set $\{1,\ldots,n\}$. For any string $x=(x_1,\ldots,x_n)$ and a subset $\mathcal{J}\subseteq [n]$, let $x|_{\mathcal{J}}$ denote the string $x$ restricted to the bits with indices in $\mathcal{J}$. We occasionally also use the notation $s(i)$ to denote the $i^{th}$ bit of the string $s$, and $x||y$ to denote the concatenation of $x$ and $y$. For every $n\in \NN$ and $j\leq n$, let \nom{brakNChoose}{$\binom{[n]}{j}$}{The set of all $j$-sized subsets of $[n]$} denote the set of all $j$-sized subsets of $[n]$.

We write $y \gets \mathsf{alg}(x)$ to denote the probabilistic process in which $y$ is sampled according to the distribution $\mathsf{alg}(x)$. The sampler $\mathsf{alg}$ takes the string $x$ as input and outputs the string $y$. To define a variable $x$ as the value $y$, we use the notation $x\equiv y$. For a finite set of $S$, the notation $x\sample S$ denotes that $x$ is sampled uniformly at random from $S$. 

When discussing a scheme, $\mathcal{S}$, which is dependent on a security parameter $\secpar$, we assume that $\secpar$ is known to all algorithms in this scheme.   To avoid ambiguity, when discussing multiple schemes $\mathcal{S}_1, \mathcal{S}_2,\ldots$ for some cryptographic primitive, we use the notation $\mathcal{S}_1.\mathsf{alg}, \mathcal{S}_2.\mathsf{alg}$ to differentiate between the algorithm $\mathsf{alg}$ of the respective schemes. 

As per usual convention, security is defined by a game consisting of two sides: a challenger $\Cadv$ and an adversary $\adv$. Unless explicitly stated otherwise, it is assumed that $\Cadv$ always refers to the challenger, and likewise, $\adv$ always refers to the adversary.

We discuss both computationally bounded and computationally unbounded adversaries. A computationally unbounded adversary is a family of quantum circuits $\{\mathcal{C}_\secpar\}_{\secpar\in\mathbb{N}}$.
A Quantum Polynomial Time (\nom{QPT}{$\QPT$}{Quantum Polynomial Time}) adversary is a family of uniform quantum circuits $\{\mathcal{C}_\secpar\}_{\secpar\in\mathbb{N}}$ such that for every $\secpar$, circuit $\mathcal{C}_\secpar$ has $p(\secpar)$ input nodes and size $S(\secpar)$ for both $q(\secpar), S(\secpar)\in \poly$. Likewise, a Probabilistic Polynomial Time (\nom{PPT}{$\PPT$}{Probabilistic Polynomial Time}) adversary denotes a probabilistic Turing Machine running in polynomial time.

A function is negligible in $\secpar$ if it decreases to $0$ faster than an inverse polynomial, with $\secpar \rightarrow \infty$. 
We also use the shorthand $f(\lambda)\leq \negl$ to state that $f$ is negligible. 
For a function $f: \Bitspace^n\rightarrow \Bitspace^m$, we define the unitary operator $\mathsf{U}_f$ as the linear operator such that
$\mathsf{U}_f(\ket{x}\ket{y})=\ket{x}\ket{f(x)\oplus y}$ where $x\in \Bitspace^n,y\in \Bitspace^m$. The notation $\QBitspace$ is used to denote the set of quantum states $\{\ket{0},\ket{1},\ket{+}, \ket{-}\}$. 

%-----------------------------------------------------------------------------%
\subsection{Tokens for MAC}\label{sec:Conjugate TMAC}
%-----------------------------------------------------------------------------%
%-----------------------------------------------------------------------------%
\subsubsection{Definition and Correctness}
%-----------------------------------------------------------------------------%
We define tokenized $\mac$, also known as a tokenized private digital signature, as in~\cite{BS16a}:

\begin{definition}
[Tokenized MAC] \label{def: TMAC scheme} A Tokenized Message
Authentication Code, or $\Tmac$ scheme consists of four $\QPT$ algorithms: $\keygen,\, \tokengen,\, \Sign$, and $\Vrfy$, with the following syntax.
\begin{enumerate}
    \item Upon input $1^\secpar$,
    where $\secpar$ is the security parameter, $\keygen$ outputs a classical key $k$ known as the secret key.
    \item The algorithm $\tokengen$ receives $k$ as input and generates some quantum state $\ket{\Stamp}$, a signing token. We stress that, in the general case, if $\tokengen$ is called $r$ times, it may output different states $\ket{\Stamp_1},\ldots,\ket{\Stamp_r}$.
    \item The algorithm $\Sign$ receives $\ket{\Stamp}$ and a classical  document \nom{a_m}{$m$}{A document or a message, that usually needs to be signed} and outputs the signature \nom{a_sig}{$\sigma$}{A signature of a document}, which is a classical string.
    \item The algorithm $\Vrfy$ receives $k$, a classical document $m$ and a classical  signature $\sigma$, and outputs a Boolean answer.
\end{enumerate}
We say a $\Tmac$ scheme is correct if for every document $m\in\Bitspace^*$, every $k$ in the range of $\keygen$, and every token $\ket{\Stamp}$ generated by $\tokengen_k$:
\[\Pr[\Vrfy_k(\Sign_{\ket{\Stamp}}(m))=1 ]=1.\]
\end{definition}
 
We also define a length-restricted version of the scheme:
\begin{restatable}[Length-Restricted $\Tmac$]{definition}{Length-Restricted TMAC}
\label{def: ell-restricted TMAC} A $\Tmac$ is $\ell$ restricted if the document space is restricted to $\Bitspace^\ell$ for some integer function $\ell(\secpar)$.
\end{restatable}
Next, we define the noise model we will be working with
as the following noise model $\Noise{p}$ for some $0\leq p\leq 1$.
\begin{definition}[Noise model $\Noise{p}$]\label{def:noise-model}
$\Noise{p}$ is the noise model in which the noise acts on each qubit as an IID Boolean random variable with a probability of $p$, i.e., at each coordinate, with a probability of $p$, the qubit gets corrupted, meaning an arbitrary CPTP map that preserves dimensions acts on the qubit, and with a probability of $1-p$, the qubit remains undeterred. 
Since we deal with BB84 states in our construction, the CPTP map in our setting that negatively affects verification the most is the Pauli operator $Y$.
\end{definition}
\begin{definition}[Noise-tolerant scheme and noise-tolerance preserving lift]\label{def:noise tolerance}
A $\Tmac$ is $\tolerant{\delta}$ for some $0\leq \delta \leq 1$ if correctness holds for the scheme up to a negligible error in the noise model $\Noise{\delta}$ (\cref{def:noise-model}), i.e., there exists a negligible function $\negl$ such that
%We say a $\Tmac$ scheme is correct if for every document $m\in\Bitspace^*$, every $k$ in the range of $\keygen$, and every token $\ket{\Stamp}$ generated by $\tokengen_k$:
\[\Pr[\Vrfy_k(\Sign_{\ket{\Stamp'}}(m))=1\mid \ket{\Stamp'}\gets \Noise{\delta}(\ket{\Stamp}); \ket{\Stamp}\gets \tokengen_k; k\gets \keygen(1^\secpar)]\geq 1-\negl.\]

A transformation that maps every $\tolerant{\delta}$ $\Tmac$ scheme to a $\tolerant{\delta}$ $\Tmac$ for every $0\leq \delta \leq 1$ is called a noise-tolerance preserving transformation or a  noise-tolerance preserving lift.
\end{definition}

\subsubsection{Notions of Security}\label{sec: notions of security}
In a vanilla $\mac$ scheme, we say forgery occurs when the adversary submits a valid signature (i.e., they pass verification) of a document that was not previously signed by a signing oracle. Clearly for $\Tmac$, this notion of forgery does not make sense since an adversary given $r$ tokens should be able to sign $r$ different documents. However, we expect that an adversary would not be able to sign more than $r$ documents. Moreover, we would like to give the adversary access to some kind of signing oracle in order to model the fact that the adversary might have access to other signed documents which were legitimately signed by others. It is natural to define forgery in the following manner that is similar in spirit to the forgery of quantum money: we say forgery occurs if an adversary who is given $r\in \poly$ tokens submits valid signatures for $r+1$ distinct documents, none of which were previously signed by the signing oracle.
However, there is a syntactic issue that needs to be taken care of. The signing procedure of a $\Tmac$ scheme receives a token as input, rather than a secret key. Therefore, a small change is needed to define the signing oracle. The analogue in $\Tmac$ for the signing oracle is the procedure $\widetilde{\Sign}_{k}(m)$, which first generates a token $\ket{\Stamp}$ by $\tokengen_k$, and then signs the document $m$ with $\Sign_{\ket{\Stamp}}(m)$.

Similarly to quantum money, an adversary that received $r$ tokens may attempt to submit any $w>r\ (w\in \poly)$ number of documents with the hope that at least $r+1$ fresh\footnote{Recall that a signature is considered fresh if the document was not signed by the $\widetilde{\Sign}_{k}$ oracle.} distinct documents pass verification. Hence, we allow the adversary to submit signatures for polynomially many documents in the unforgeability game, and they win if at least $r+1$ of them pass verification.

Furthermore, an adversary may wish to submit a signature, check if it gets accepted, and improve future guesses based on the result. This is modeled in the unforgeability game by giving the adversary access to a verification oracle. 
Readers familiar with $\mac$ might find it peculiar that we provide the adversary with a verification oracle.  Many $\mac$ schemes have deterministic signing (and in fact, this is w.l.o.g.--see~\cite[Section 6.1.5.2]{Gol04}). In such schemes, there is no need for a verification algorithm altogether: %verification is done by running the signing algorithm on the to-be-verified document.
Verification of a document $m$ and an alleged signature $\sigma$ is done by testing whether $\sigma \overset{?}{=
}\Sign_k(m)$. This type of verification is called \emph{canonical verification}.
It negates the need for a verification oracle, as explained in Ref.~\cite[Section 6.1.5.1]{Gol04}. Therefore, the verification oracle is often ignored.
Since the $\Tmac$ presented in this work does not have deterministic signing, the verification oracle is necessary to model the capabilities of the adversary.

There is some subtlety in determining to which of the oracles the adversary is allowed quantum access to, i.e., allowed to make superposition queries. While there is strong motivation to allow quantum access to the verification oracle,\footnote{The verification procedure could be potentially given as an obfuscated circuit~\cite{BGIRSVY12} in a potential lift to a public tokenized digital signatures scheme. A quantum adversary could then run the obfuscated circuit in superposition.} there is less so for quantum access to a signing oracle. In addition, even defining security with regard to quantum access to the signing oracle is a contrived task  (see also related works on~\cpageref{pg:Related Works}). On the other hand, security with regard to quantum access to verification is concisely defined for $\Tmac$s with deterministic verification, but likewise, this is complicated to define in the most general scenario. While our construction does have deterministic verification, it does not achieve this stronger security notion. For the sake of generality, we present here a definition allowing for non-deterministic or even quantum verification, and postpone the definition of unforgeability against adversaries with quantum access to verification to \cref{sec:A Quantum Superposition Attack}, along with the appropriate attack on our construction.

To ease notation, we provide the tokens to the adversary by the oracle $\tokengen_k$ that the adversary can query polynomially many times. However, it can be assumed, without loss of generality, that an adversary always makes exactly $r(\secpar)\in \poly$ calls to $\tokengen_k$ immediately after receiving the security parameters. In an intermediate step of our analysis, we will also discuss unforgeability against adversaries holding a single token, i.e., the adversary can access the $\tokengen$ oracle, but only once.

With the above motivations in mind, we define the following unforgeability game in Game~\ref{exp:UnfExp}. %The security game for the notion of security we described here is brought in Game~\ref{exp:UnfExp}.
\nomenclature[Forge]{$\UnfExp$}{Unforgeability game for a $\Tmac$}
  \begin{savenotes}
\begin{game}
 \caption{   $\UnfExp^{\tokengen,\widetilde{\Sign},\Vrfy}_{\adv,\Pi}(\secpar)$
} \label{exp:UnfExp} 
\begin{algorithmic}[1] 
    \State  $\Cadv$ creates a secret key by running $k \gets \keygen(1^\secpar)$ and sends $1^\secpar$ to $\adv$.
    
    \State $\adv$ is given classical oracle access to $\widetilde{\Sign}_k$, $\Vrfy_k$ and $\tokengen_k$, and can query each oracle a polynomial number of times.  Let $r$ denote the number of times $\tokengen$ is called, and let $Q$ denote the set of all queries that are made to $\widetilde{\Sign}_k$. 
    
    \State $\adv$ sends the challenger $w$ documents and their signature $(m_i,\sigma_i)_{w}$.
    %\State $\Cadv$ runs $\Vrfy_k$ on all the above pairs.
    \State $\Cadv$ runs $c_i\gets \Vrfy_k(m_i,\sigma_i)$.

\end{algorithmic}
     Let $S$ be the set of all $ i\in [w]$ such that $c_i= 1$, and let  $count=\abs{\{m_i\mid i\in S\bigwedge m_i\notin Q\}}$. Then the output of the game is $1$ if and only if $count\geq r+1$.

 \end{game}
\end{savenotes}

\begin{definition}\label{def:Unforgeability}
A $\Tmac$ scheme $\Pi = (\keygen,\tokengen, \Sign, \Vrfy)$ is said to be 
$\Exly$ $\UnfDef^{\tokengen,\Vrfy,\widetilde{\Sign}}$
if for any $\QPT$ adversary $\adv$ 
\[\Pr[\UnfExp^{\tokengen,\Vrfy,\widetilde{\Sign}}_{\adv,\Pi}(\secpar) = 1] \leq \negl.\]
Likewise, various other types of unforgeability can be defined by preventing access to some of the oracles, which is notated by removing the corresponding oracles from the superscript in the definition and security game. The notation $(\cdot)^{\OneToken}$ instead of $(\cdot)^{\tokengen}$ denotes that only a single access to the oracle $\tokengen_k$ is allowed. The functionality of the oracles is summarized in  \cref{table:oracle types}.
\end{definition}
\renewcommand\arraystretch{1.5}
 \begin{table}[ht] 
\centering
\begin{adjustbox}{max width=\textwidth}
\begin{tabular}{l l }

\toprule                                          Oracle        & Functionality                                                                                            \\

\midrule \multirow{2}{*}{}  $(\cdot)^\Vrfy$       & $\adv$ has classical access to $\Vrfy_k$                                                                                                                \\
                                                          $(\cdot)^{\widetilde{\Sign}}$        & $\adv$ has classical access to $\widetilde{\Sign}_k$    \\
                                                                                   $(\cdot)^{\tokengen}$  & $\adv$ has  access to $\tokengen_k$
                                                \\
                                                $(\cdot)^{\OneToken}$  &  $\adv$ can make at most one query to $\tokengen_k$     \\
 \bottomrule   
\end{tabular}
\end{adjustbox}
\caption{Possible oracles an adversary may have access to in Game~\ref{exp:UnfExp}}
\label{table:oracle types}
\end{table}
The definition above can easily be adjusted to allow weaker forms of unforgeability such as random unforgeability, selective unforgeability, and universal unforgeability. In this work, we only address existential unforgeability; hence, the prefix $\Exly$ is omitted for the remainder of the paper.

An $\Unc$ $\UnfDef^{(\cdot)}$ $\Tmac$ is defined similarly, but where the above holds even against computationally unbounded adversaries. It is emphasized that a computationally unbounded adversary could still only query the oracles a polynomial number of times.

\ifnum\llncs=0
\subsubsection{Sabotage Attacks}
Another notion often discussed in relation to $\mac$ is that of sabotage. An adversary might not be able to forge a fresh signature, but still harm other users in other ways. Sabotage is the act of creating a malicious signature, that successfully passes one verification, but might not pass a second verification. 

\begin{savenotes}
\begin{game}
 \caption{   $\Sabotage^{\tokengen,\widetilde{\Sign},\Vrfy}_{\adv,\Pi}(\secpar)$
} \label{exp:UnfSab} 
\begin{algorithmic}[1] 
    \State  $\Cadv$ creates a secret key by running $k \gets \keygen(1^\secpar)$ and sends $1^\secpar$ to $\adv$.
    
    \State $\adv$ is given classical oracle access to $\widetilde{\Sign}_k$, $\Vrfy_k$ and $\tokengen_k$ and,  and can query each oracle a polynomial number of times.
    
    \State $\adv$ sends the challenger a signed document $(m,\sigma)$
    %\State $\Cadv$ runs $\Vrfy_k$ on all the above pairs.
    \State $\Cadv$ runs  $\Vrfy_k(m,\sigma)$ \emph{twice}.

\end{algorithmic}
     The output of the game is $1$ if and only if the first verification was successful, but the second was not.

 \end{game}
\end{savenotes}

\begin{definition}\label{def:sabotage secure}
A $\Tmac$ scheme $\Pi = (\keygen,\tokengen, \Sign, \Vrfy)$ is said to be 
 $\SabDef^{\tokengen,\Vrfy,\widetilde{\Sign}}$
if for any $\QPT$ adversary $\adv$ 
\[\Pr[\Sabotage^{\tokengen,\Vrfy,\widetilde{\Sign}}_{\adv,\Pi}(\secpar) = 1] \leq \negl.\]
Likewise, various other types of sabotage security can be defined by preventing access to some of the oracles, which is represented in the same way as in \cref{exp:UnfExp}.
\end{definition}

If a party can use a token to create such a signature, the party does not gain anything from it, but can harm others: Say Bob makes a transaction with Charlie's company, Charlie's assistant Alice uses one of Charlie's tokens to maliciously sign the transaction on his absence, and to verify the signature Bob turns to the trusted Dave, who knows Charlie's key. Upon Charlie's return, Bob shows him the signature, but to his surprise verification fails, and the annoyed Charlie calls for security to deport Bob, while Alice enjoys her stolen loot. Of course, if Bob is allowed to verify his signature an unlimited amount of times, verification will eventually pass, as the signature is a classical string. However even then, Bob would not necessarily know the number of times he would need to verify the signature in order to succeed with high probability. One could imagine that the verifier would lose patience and become suspicious after $100$ failed verifications.

It is clear that sabotage is not an issue in any $\Tmac$ that has deterministic verification, since in that case, signatures would either be accepted or rejected with probability $1$. As the construction presented in this paper has deterministic verification,\footnote{Assuming that the lift in \cref{lem: unrestricted onetime implies full blown unrestricted} is instantiated with a deterministic verification authenticated encryption, as most practical authenticated encryptions are.} we will not discuss security against sabotage any further.
\fi
\section{Main Result}\label{sec:main_results} 
We say that a $\Tmac$ scheme is based on conjugate coding states if every token is a tensor product of states from $\{\ket{0},\ket{1},\ket{+},\ket{-}\}$, perhaps in addition to some classical string. 
Our main result is the following:
\MainResult*
\begin{proof}\label{pf:thm:main_result}
Given \nom{eta}{$\eta$}{Error threshold}$\geq 0$, we first construct a $1$-restricted (that is for $m\in \Bitspace$) $\tolerant{2c\eta}$ $\Tmac$ (see \cref{def:noise tolerance}) scheme $\CTMAC^\eta$ (see \cref{alg:noise_tolerant Conjugate Tmac,thm:correctness_noise tolerance}) based on conjugate states for any constant $0\leq c<1$. %see \cref{alg:Conjugate Tmac} in \cref{sec:candidate_construction}. %We name this scheme, (\cref{alg:Conjugate Tmac}) Conjugate $\Tmac$, or $\CTMAC$, which we show 
In \cref{sec:noise-tolerant unforgeability}, we show that the scheme $\CTMAC^\eta$ is  $\UnfDef^{\OneToken,\Vrfy}$ (see \cref{thm:Pi_W is secure}), for $0\leq \eta \leq \frac{1-\alpha}{2}$.

Hence, given $\delta<1-\alpha$, there exists $\widetilde{\eta}$ such that $\delta<\widetilde{\eta}<1-\alpha$. Let $\widetilde{c}=\frac{\delta}{\widetilde{\eta}}$, which is clearly smaller than $1$. Since $\frac{\widetilde{\eta}}{2}<\frac{1-\alpha}{2}$, by \cref{thm:Pi_W is secure}, $\CTMAC^{\frac{\widetilde{\eta}}{2}}$ is $\UnfDef^{\OneToken,\Vrfy}$ and by \cref{thm:correctness_noise tolerance} is $\tolerant{2\widetilde{c}\frac{\widetilde{\eta}}{2}}$, i.e., $\tolerant{\delta}$.
\nomenclature[CTMACFull]{$\sigOrUnrest{\eta}$}{The result of the expansion of $\CTMAC^\eta$ to a full blown scheme}
Next, we lift $\CTMAC^{\frac{\widetilde{\eta}}{2}}$ to a full-blown $\UnfDef^{\tokengen,\Vrfy,\widetilde{\Sign}}$ $\Tmac$ scheme $\sigOrUnrest{\frac{\widetilde{\eta}}{2}}$, while preserving the parameter of noise tolerance using the following proposition, which is proven in \cref{sec:Security proof for full blown} on \cpageref{pf:prp:oneRtofb}. 
\begin{proposition}
\label{prp:oneRtofb}

Assuming the existence of post-quantum one-way functions, there is a noise-tolerance preserving lift (see \cref{def:noise tolerance}) of any $1$-restricted $\UnfDef^{\OneToken,\Vrfy}$ $\Tmac$  (see \cref{def: ell-restricted TMAC,def:Unforgeability}) to an unrestricted  $\UnfDef^{\tokengen,\Vrfy,\widetilde{\Sign}}$ $\Tmac$. Moreover, if the $1$-restricted scheme is based on conjugate states then the same holds for the resulting unrestricted scheme.
 \end{proposition}
Hence, $\sigOrUnrest{\frac{\widetilde{\eta}}{2}}$ is the required $\Pi_\delta$.
 \end{proof}

In particular, since $0.14<1-\alpha$, we get the following corollary of \cref{thm:main_result}.
\begin{corollary}\label{cor:concrete_bound-main-result}
Assuming post-quantum one-way functions exist, there exists an $\UnfDef^{\tokengen,\Vrfy,\widetilde{\Sign}}$ $\tolerant{14\%}$ (see \cref{def:Unforgeability,def:noise tolerance}) $\Tmac$ scheme (see \cref{alg:noise_tolerant Conjugate Tmac,prp:oneRtofb}) based on conjugate coding states.
\end{corollary}

\section{Construction of a 1-Restricted TMAC Based on Conjugate Coding States}\label{sec:candidate_construction}

\subsection{A Noise-Sensitive Scheme}
In \cref{alg:Conjugate Tmac}, a $1$-restricted $\Tmac$ scheme (that is for $m\in \Bitspace$) based on the classical verification variant of Wiesner's money is described. We refer to this scheme as the Conjugate $\Tmac$, and write \nom{CTMAC}{$\CTMAC$}{$1$-bit noise-sensitive Conjugate $\Tmac$ scheme}.
\begin{algorithm*}
    \caption{$\CTMAC$ - The $1$-bit Conjugate $\Tmac$}
    \label{alg:Conjugate Tmac}
    \begin{algorithmic}[1] % The number tells where the line numbering should start
        \Procedure{$\keygen$}{$1^\secpar$}
            \State $a \sample \Bitspace^\secpar$, $b \sample \Bitspace^\secpar$.
            \State \textbf{Return} $k\equiv(a,b)$.
        \EndProcedure

    \end{algorithmic}
    \begin{algorithmic}[1] % The number tells where the line numbering should start
        \Procedure{$\tokengen_k$}{}
            \State Interpret $k=(a,b)$.
            \State \textbf{Return} $\ket{\Stamp}\equiv H^{b}\ket{a}$.
        \EndProcedure

    \end{algorithmic}
	\begin{algorithmic}[1] % The number tells where the line numbering should start
        \Procedure{$\Sign_{\ket{\Stamp}}$}{$m$} \Comment{ $m\in \Bitspace$}
        \State Measure $(H^m)^{\otimes\secpar}\ket{\Stamp}$
         in the standard basis to obtain a string $\sigma$. 
            \State \textbf{Return} $\sigma$.
        \EndProcedure

    \end{algorithmic}
  	\begin{algorithmic}[1] % The number tells where the line numbering should start
        \Procedure{$\Vrfy_k$}{$m,\sigma$}
            \State Interpret $k$ as $(a,b)$.
            \State  Construct the subset \nom{cons_c}{$\Cons_m$}{The subset  $\{i\in[\secpar]\mid b_i=m\}$, with regards to (a,b)}$=\{i\in[\secpar]| b_i=m\}$.
            \If{$a|_{\Cons_m}=\sigma|_{\Cons_m}$}
                \State \textbf{Return} $1$.
            \Else
                \State \textbf{Return} $0$.
            \EndIf
        \EndProcedure
    \end{algorithmic}
\end{algorithm*}

\ifnum\llncs=0
We first prove the correctness of $\CTMAC$.
\fi

\begin{proposition}[Correctness of $\CTMAC$]
\label{prp:correctness-one_restricted}
 The $1$-restricted Conjugate $\Tmac$, $\CTMAC$, given in \cref{alg:Conjugate Tmac}, is correct.
\end{proposition}
\ifnum\llncs=1
The proof of \cref{prp:correctness-one_restricted} is trivial and hence is omitted.
\fi
\ifnum\llncs=0

\begin{proof}\label{pf:correctness-one_restricted} 
Let $k\equiv(a,b)\gets \keygen(1^\secpar)$ be the secret key, where $\secpar$ is the security parameter, fixed arbitrarily. Hence, the token $\ket{\Stamp}$ generated by $\tokengen_k$ is $H^b\ket{a}$. Fix $m\in \{0,1\}$. 
Let, $\sigma\gets\Sign_{\ket{\Stamp}}(m)$. By definition, $\sigma$ is the (standard basis) measurement outcome of \[(H^m)^{\otimes\secpar}\ket{\Stamp}= (H^m)^{\otimes\secpar}(H^b\ket{a}).\]  
Let $\Cons_m\equiv\{i\in [\secpar]\quad |\quad b_i=m\}$. Due to the action of the Hadamard gate twice, for every $i\in \Cons_m$, the $i^{th}$ qubit of the state $(H^m)^{\otimes\secpar}(H^b\ket{a})$ remains $\ket{a_i}$. Hence, with probability $1$, $\sigma|_{\Cons_m}=a|_{\Cons_m}$. 
Since, $\Vrfy_k(m,\sigma)$ simply checks if $\sigma|_{\Cons_m}=a|_{\Cons_m}$, the verification accepts it with certainty. Therefore, for every $k$ in the range of $\keygen$, and every token $\ket{\Stamp}$ generated by $\tokengen_k$ 
\[\Pr[\Vrfy_k(m, \Sign_{\ket{\Stamp}}(m))=1 ]=1,\]
for all $m\in\Bitspace$.
Since $\secpar$ and $m$ were fixed arbitrarily, we conclude that $\CTMAC$ is correct.
\end{proof}
\fi

\begin{restatable}{theorem}{NoiseSensSecurity}
\label{thm:Pi_W is secure}
 The conjugate $\Tmac$ scheme $\CTMAC$ given in \cref{alg:Conjugate Tmac}, is $\Unc$ $\UnfDef^{\OneToken,\Vrfy}$.
\end{restatable}
\cref{thm:Pi_W is secure} is proven in \cref{sec: reduction to weak certified deletion}. In \cref{sec:Drawbacks Of the Conjugate Tmac}, it is shown that $\CTMAC$ does not satisfy some stronger notions of unforgeability, namely, strong unforgeability and unforgeability against adversaries with quantum verification queries.
\subsection{A Noise-Tolerant Scheme}\label{sec: A Noise-Tolerant Scheme}
In this section, we follow ideas similar to the constructions in~\cite{PYJ+12,MVW12} to extend our $1$-restricted $\Tmac$ scheme $\CTMAC$ (\cref{alg:Conjugate Tmac}) scheme to a noise-tolerant variant which we call \nom{CTMACeta}{$\CTMAC^\eta$}{$1$-bit noise-tolerant Conjugate $\Tmac$ scheme} (see \cref{alg:noise_tolerant Conjugate Tmac}), where $\eta$ represents an error threshold.
The construction of $\CTMAC^\eta$ is the same as $\CTMAC$, but with a lenient verification procedure where we accept an alleged signature for a bit $b$ even if the signature is not consistent with the secret key, up to a constant fraction of the coordinates. The full construction is given in \cref{alg:noise_tolerant Conjugate Tmac}.

\begin{algorithm*}
    \caption{$\CTMAC^\eta$ - $\tolerant{2c\eta}$ Conjugate $\Tmac$, for all constant $0\leq c<1$}
    \label{alg:noise_tolerant Conjugate Tmac}
    \begin{algorithmic}[1] % The number tells where the line numbering should start
        \Procedure{$\keygen$}{$1^\secpar$}
            \State The same as $\CTMAC.\keygen$.
        \EndProcedure

    \end{algorithmic}
    \begin{algorithmic}[1] % The number tells where the line numbering should start
        \Procedure{$\tokengen_k$}{}
        \State The same as $\CTMAC.\tokengen_k$.
        \EndProcedure

    \end{algorithmic}
	\begin{algorithmic}[1] % The number tells where the line numbering should start
        \Procedure{$\Sign_{\ket{\Stamp}}$}{$m$}
        \State  The same as $\CTMAC.\Sign_{\ket{\Stamp}}$.
        \EndProcedure

    \end{algorithmic}
  	\begin{algorithmic}[1] % The number tells where the line numbering should start
        \Procedure{$\Vrfy_k$}{$m,\sigma$}
            \State Interpret $k$ as $(a,b)$.
            \State  Construct the subset $\Cons_m=\{i\in[\secpar]\mid b_i=m\}$.
            \State Construct the subset
            \nom{cons_miss}{$\Miss_{m,\sigma}$}{The subset $\{i\in \Cons_m\mid \sigma_i\neq a_i\}$ with regards to (a,b)}$=\{i\in \Cons_m\mid \sigma_i\neq a_i\}$.
            \If{$\abs{\Miss_{m,\sigma}}\leq\eta\secpar$}
                \State \textbf{Return} $1$.
            \Else
                \State \textbf{Return} $0$.
            \EndIf
        \EndProcedure
    \end{algorithmic}
\end{algorithm*}

\begin{theorem}\label{thm:correctness_noise tolerance}
The $1$-restricted (see \cref{def: ell-restricted TMAC}) $\Tmac$ scheme $\CTMAC^\eta$ is correct. Moreover, for any constant $0\leq c<1$, $\CTMAC^\eta$ is $\tolerant{2c\eta}$ (see \cref{def:noise tolerance}).  
\end{theorem} 
\begin{proof}\label{pf:thm:correctness_noise tolerance}
The correctness follows from the correctness of the noise-sensitive scheme $\CTMAC$ (see \cref{alg:Conjugate Tmac}) \ifnum\llncs=0 because the two schemes have the same $\keygen,\tokengen$, and $\Sign$ algorithms, and the verification\footnote{Note that the verifications of both $\CTMAC$ (\cref{alg:Conjugate Tmac}) and $\CTMAC^\eta$ (\cref{alg:noise_tolerant Conjugate Tmac}) are deterministic.} is more lenient in $\CTMAC^\eta$, i.e., for the same secret key $k$, \[\CTMAC^\eta.\Vrfy_k(m,\sigma)=1\implies \CTMAC.\Vrfy_k(m,\sigma)=1,\] for any alleged document-signature pair $(m,\sigma)$.\else.\fi

Next, we show noise tolerance for $\CTMAC^\eta$. Let $0\leq c<1$ be an arbitrary constant. 
We will consider the noise model $\Noise{2c\eta}$ (\cref{def:noise-model}) from here on. Fix $m\in \{0,1\}$.
Let $k\gets \CTMAC^\eta.\keygen(1^\secpar)$ and $\ket{\Stamp}=\tensor_{j=1}^\secpar \ket{\Stamp_j}\gets \CTMAC^\eta. \tokengen_k$. Suppose some qubits of $\ket{\Stamp}$ got corrupted due to noise, and let the resulting state be $\ket{\Stamp'}$.

For every $j\in [\secpar]$, let $W_j,Y_j$ be Boolean random variables such that $Y_j=1$ \ifnum\llncs=0 if and only if \else iff \fi the $j^{th}$ qubit of $\ket{\Stamp}$ was corrupted by noise, and $W_j=1$ \ifnum\llncs=0 if and only if  the quantum state of the $j^{th}$ qubit of $\ket{\Stamp}$ is in the relevant basis with respect to $m$, i.e.,
\[\ket{\Stamp_j}\in \{H^m\ket{0},H^m\ket{1}\}.\]\else iff the $j^{th}$ qubit of $\ket{\Stamp}$ is in the set $\{H^m\ket{0},H^m\ket{1}\}$. \fi
Let the Boolean random variable $X_j$ be defined as \ifnum\llncs=1 $X_j=W_j\cdot Y_j$. \else  the product of $W_j$ and $Y_j$.
\[X_j=W_j\cdot Y_j.\]\fi

By definition of the $\CTMAC^\eta.\keygen$ algorithm, $\{W_j\}_{j\in[\secpar]}$ are IID random variables with parameter $\frac{1}{2}$ meaning,
$\forall j \in [\secpar],\Pr[W_j=1]=\frac{1}{2},$
moreover, $\{Y_j\}_{j\in [\secpar]}$ are IID with parameter $2c\eta$,  by definition of the noise model $\Noise{2c\eta}$ (\cref{def:noise-model}).
Clearly, for every $j\in [\secpar]$, $W_j$ and $Y_j$ independent random variables\ifnum\llncs=0 because the noise acts on the qubits independent of the $\keygen$ algorithm\fi.
Therefore, by definition, $\{X_j\}_{j\in [\secpar]}$ are IID random variables with parameter $\frac{1}{2}\cdot 2c\eta=c\eta$. 
Let $X\equiv\sum_{i=j}^\secpar X_j$.

Note that \ifnum\llncs=0 $X_j=1$ denotes the event that the qubit at the $j^{th}$ coordinate is in the relevant basis with respect to the message $m$ and that the qubit was corrupted by the noise. Hence, \fi the quantum states of the qubits of $\ket{\Stamp'}$ at the relevant coordinates with respect to $m$, differ from that of $\ket{\Stamp}$ at exactly $X$ coordinates. Since there are a total of $\secpar$ qubits in each token, and $\{X_j\}_{j\in [\secpar]}$ are IID with parameter $c\eta$,
$E[X]=E[\sum_{i=1}^\secpar X_i]=c\eta\secpar.$

\ifnum\llncs=0
Moreover, using a standard Chernoff-bound 
argument, we get,
\[\Pr[X>\eta \secpar]=%\Pr[X>\left(1-c\right)\eta \secpar + E(X)]\leq e^{-2\left(1-c\right)^2\eta^2 \secpar}=
d^\secpar
%\left(\frac{e^{\frac{1-c}{c}}}{(\frac{1}{c})^{\frac{1}{c}}}\right)^{E(x)}= \left(\frac{e^{\frac{1-c}{c}}}{(\frac{1}{c})^{\frac{1}{c}}}\right)^{c\eta}=d^{c\eta}
 \implies \Pr[X\leq\eta \secpar]\geq 1-d^\secpar,\]
where $d\equiv e^{-2\left(1-c\right)^2\eta^2}<1$.
\else
Next, we use a  standard Chernoff-bound 
argument to conclude that  $\Pr[X\leq\eta \secpar]\geq 1-d^\secpar$ where $d\equiv e^{-2\left(1-c\right)^2\eta^2}<1$.
\fi
Let $\sigma'\gets\CTMAC^\eta.\Sign_{\ket{\Stamp'}}( m)$. Conditioned on the event that $X\leq \eta \secpar$, the signature $\sigma'$ would be inconsistent with the secret key at less than $\eta\secpar$ of the relevant coordinates with respect to $m$, and hence, $\CTMAC^\eta.\Vrfy_k(\sigma', m)=1$ with certainty. Therefore, \[\Pr[\CTMAC^\eta.\Vrfy_k(\sigma', m)=1]\geq \Pr[X\leq\eta \secpar]\geq  1-d^\secpar\ifnum\llncs=0.\else,\fi\]
\ifnum\llncs=0 Since $d^{\secpar}$ is negligible in $\secpar$, and $m\in \Bitspace$ was arbitrary, we conclude that $\CTMAC^\eta$ is correct up to a negligible error in the noise model, $\Noise{2c\eta}$, and is hence $\tolerant{2c\eta}$. \else which is negligibly close to $1$. \fi Since $0\leq c<1$ was arbitrary, we conclude that $\CTMAC^\eta$ is $\tolerant{2c\eta}$ for every $0\leq c<1$.
\end{proof}

\begin{restatable}{theorem}{NoiseSecurity}
  
\label{thm:Pi^t_W is secure}
The $1$-restricted (\cref{def: ell-restricted TMAC}) scheme $\CTMAC^\eta$ (\cref{alg:noise_tolerant Conjugate Tmac}) is $\UnfDef^{\OneToken,\Vrfy}$ for any $0<\eta<\frac{1-\alpha}{2}$, where $\alpha= \cos^2(\frac{\pi}{8})$. 
\end{restatable}

The proof of \cref{thm:Pi^t_W is secure} is given in \cref{sec:noise-tolerant unforgeability}.

In particular, since $0.07<\frac{1-\alpha}{2}$, we get the following corollary of \cref{thm:correctness_noise tolerance,thm:Pi^t_W is secure}.

\begin{corollary}\label{cor:concrete_bound-onerstricted}
The $1$-restricted scheme (see \cref{def: ell-restricted TMAC})  $\CTMAC^{0.07}$ (see \cref{alg:noise_tolerant Conjugate Tmac}), which is tolerant to noise up to $14\%$ noise, is $\Unc$ $\UnfDef^{\OneToken,\Vrfy}$.
\end{corollary}

%------------------------------------------------------------------------------------------------------------------------------%
\section{Proving Single Token Unforgeability for  \texorpdfstring{$\CTMAC^\eta$}{One-Restricted TMAC}}\label{sec:Security of Conjugate TMAC}
%------------------------------------------------------------------------------------------------------------------------------%
In this section, we prove that  $\CTMAC^\eta$, the Conjugate $\Tmac$ (\cref{alg:Conjugate Tmac}), is $\Unc$  $\UnfDef^{\OneToken,\Vrfy}$ (see \cref{def:Unforgeability}).

Recall that in the security game $\UnfExp^{\OneToken,\Vrfy}_{\adv,\Pi}(\secpar)$ for a $1$-restricted scheme (see \cref{def: ell-restricted TMAC}), the adversary wins if it submits valid signatures for both $0$ and $1$, given a single token $\ket{\Stamp}$ and oracle access to $\Vrfy$. 

A na\"ive attempt for a proof would be to follow the blueprint used to prove the security for classically verifiable private quantum money schemes in Refs.~\cite{MVW12,PYJ+12} because the bill, challenges, and answers to the challenges there roughly correspond to the tokens, documents, and signatures in our construction, respectively. The blueprint proceeds as follows. In the first step, unforgeability is proven against $1$-to-$2$ counterfeiting (called "simple counterfeiting" in~\cite{MVW12}), meaning an adversary with a single bill (without any access to the verification oracle) cannot submit two bills that pass verification. Next, a reduction is shown from an adversary with oracle access to verification, submitting polynomially many money-states to a  $1$-to-$2$ adversary.  %The first step poses no issues in $\Tmac$: 
We can define $1$-to-$2$ unforgeability in a manner similar to $\UnfDef^{\OneToken}$ with a small change that the adversary is forced to submit only two documents, and hence, we can follow the first step of the blueprint with respect to $\CTMAC^\eta$ without an issue. The analysis would almost be the same as in~\cite[Section 4.2]{MVW12}. 
However, there is an issue in mimicking the second step in the blueprint, i.e., it is not clear in case of $\CTMAC^\eta$, how to simulate an adversary with access to verification oracles using an adversary mentioned without the verification oracle. In order to shed more light on this issue, we briefly review the reduction used in~\cite{MVW12}. 
The reduction constructs a $1$-to-$2$ adversary (without any oracle access) that simulates the adversary augmented with a verification oracle by guessing the responses of the oracle queries. The decrease in success probability is only polynomial because it can be assumed that at most two of the queries of the augmented adversary were successful. This property holds in the context of~\cite{MVW12} because the challenger chooses the challenges (representing measurements) which means with overwhelming probability, the challenges are all distinct, and all queries to the oracle would be with regard to distinct challenges.
 
 On the contrary, in the setting of a $1$-restricted $\Tmac$, the adversary in the unforgeability game (corresponding to $\UnfExp^{\OneToken,\Vrfy}(\secpar)$) gets to choose the "challenges" themselves, and hence the challenges need not be random. In particular, the adversary might query the signature of the same document successfully queried multiple times with different signatures. Hence, a uniformly random guess for the set of successful queries would be correct only with exponentially small probability.\footnote{In fact, there are $\Tmac$ schemes which are $\UnfDef^{\OneToken}$, but are not $\UnfDef^{\OneToken,\Vrfy}$. In order to see this, consider an $\UnfDef^{\OneToken,\Vrfy}$ $\Tmac$ such as our construction or~\cite{BS16a}, and modify its verification procedure to also accept valid signatures appended with a prefix of the secret key. Repeated tries and fails of verification could then uncover the key. }
It is easy to see that the issue discussed above cannot occur in a strong $\Tmac$ (see \cref{def:strong TMAC}) scheme because it is not possible to successfully query different signatures of the same document in a strong $\Tmac$, except with negligible probability.
However, our construction is not a strong $\Tmac$.
For instance, if we consider the  noise-sensitive  scheme $\CTMAC^0$, then we can construct an adversary that given one token and access to verification oracle, produces two different signatures of a single bit, as described in \cref{sec:A break of Strong Unforgeability}. The main idea is that an adversary with a verification oracle could easily measure the token in one of the two bases (computational or Hadamard), and obtain a valid signature as described in \cref{sec:A break of Strong Unforgeability}. Then she could perform queries to the oracle changing one bit at a time to uncover the quantum states of all coordinates in the chosen basis which would enable her to create multiple signatures of the bit corresponding to the chosen basis.

Ergo, we argue for security against attacks augmented with verification oracle via a scheme specific approach.
The crucial point to observe is that in the attack mentioned above, the adversary only learns the quantum state of the coordinates that are in one particular basis. Informally, we claim that this is all the information the adversary can learn using the verification oracle. This statement is made formal in \cref{sec:red_supplied_basis_unforgeability}.

The structure of the proof is given in \cref{fig:proof idea}.  We first consider the security game $\UnfExp_{\adv,\CTMAC^\eta}^{\OneToken,\NVR}(\secpar)$ that is the same as the security game for $\UnfExp_{\adv,\CTMAC^\eta}^{\OneToken,\Vrfy}(\secpar)$, but with a different verification oracle that along with the result, also returns additional information after a successful query. The additional information that the oracle provides, represents all the information that the adversary could have obtained by making repeated successful queries to the same message. Hence, there is no need for the adversary to make repeated successful queries to the same message, and the game is equivalent in power if the adversary is not allowed to do so, which is the definition of the game $\UnfExp_{\adv,\CTMAC^\eta}^{\OneToken,\NVR^*}(\secpar)$.
$\UnfExp_{\adv,\CTMAC^\eta}^{\OneToken,\NVR}(\secpar)$ clearly gives more power to the adversary compared to the security game corresponding to $\UnfExp_{\adv,\CTMAC^\eta}^{\OneToken,\Vrfy}(\secpar)$. 
Hence, if the success probability for any adversary in $\UnfExp_{\adv,\CTMAC^\eta}^{\OneToken,\NVR^*}(\secpar)$ is negligible, then the same holds in $\UnfExp_{\adv,\CTMAC^\eta}^{\OneToken,\Vrfy}(\secpar)$. 
Next, it is shown that for $\eta=0$, an adversary in $\UnfExp_{\adv,\CTMAC^\eta}^{\OneToken,\NVR^*}(\secpar)$  can be reduced to an adversary in the "weak" certified deletion game for a tailor-made quantum encryption with certified deletion scheme~\cite{BI20}. 

Finally, we conclude by showing a reduction from $\UnfExp_{\adv,\CTMAC^\eta}^{\OneToken,\NVR^*}(\secpar)$ with respect to the noise-tolerant scheme to the same game but with respect to the noise-sensitive scheme (i.e $\eta=0$), with a decrease in the security parameter. The reduction is unusual in the sense that the success probability decreases by a multiplicative exponential factor, but the bound on the success probability of the adversary in the reduced game is small enough that we still get a meaningful bound on the success probability of the adversary in the original game.

The arguments used in all the steps of the proof are information-theoretic and hence, we conclude that $\CTMAC^\eta$ is  $\Unc$ $\UnfDef^{\OneToken,\Vrfy}$. This is a stronger result in the sense that we require unforgeability only against $\QPT$ adversaries to lift the scheme to an $\UnfDef^{\tokengen,\Vrfy,\widetilde{\Sign}}$ $\Tmac$. 
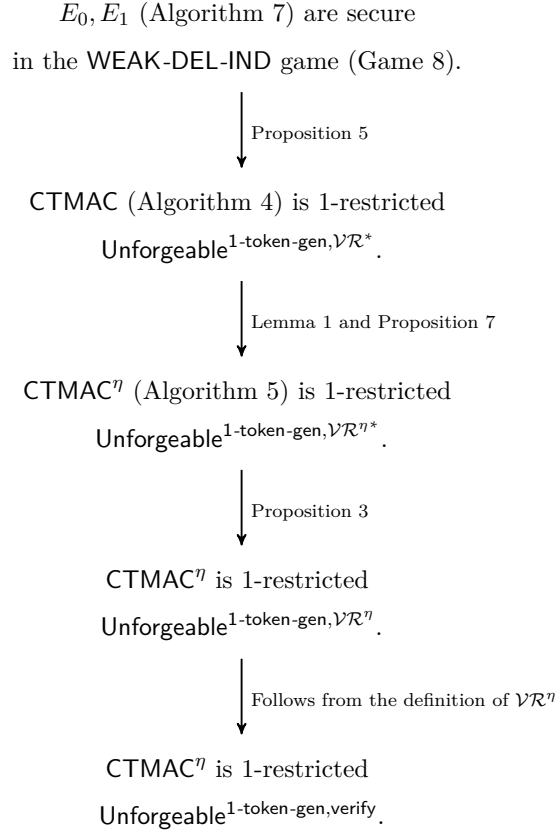
\begin{figure*}[htp]
\centering
\begin{adjustbox}{max width=0.6 \textwidth,center}
    \begin{tikzpicture}%[scale=0.5, align = left]
  \matrix (m) [matrix of nodes, row sep=3em, column sep=0em, align=center]
    {$\begin{matrix} \text{$E_0,E_1$ (\cref{alg: E_m}) are secure}\\ \text{in the $\exWCerDel$ game (Game~\ref{exp: weak certified deletion}).} \end{matrix}$
    \\
    $\begin{matrix}\text{$\CTMAC$ (\cref{alg:Conjugate Tmac}) is $1$-restricted}\\ \text{ $\UnfDef^{\OneToken,{\VR}^*}$.}\end{matrix}$
\\
$\begin{matrix}\text{$\CTMAC^\eta$ (\cref{alg:noise_tolerant Conjugate Tmac}) is $1$-restricted}\\ \text{ $\UnfDef^{\OneToken,{\NVR}^*}$.}\end{matrix}$
\\
$\begin{matrix}\text{$\CTMAC^\eta$ is $1$-restricted}\\ \text{ $\UnfDef^{\OneToken,\NVR}$.}\end{matrix}$
\\
      $\begin{matrix}\text{$\CTMAC^\eta$  is $1$-restricted}\\ \text{ $\UnfDef^{\OneToken,\Vrfy}$.}\end{matrix}$
\\
     }; 
  
 \path[->,draw,thick]
   (m-1-1) edge [auto] node[labeled]  [scale=1] 
   {\text{$\begin{matrix}\text{\cref{prp:weak certififed deletion security implies weak basis supply security}}
   \end{matrix}$
    }} (m-2-1) ;
     
\path[->,draw,thick]
   (m-2-1) edge [auto] node[labeled]  [scale=1] 
   {\text{$\begin{matrix}\text{\cref{lem:noise reduction lemma,prp: reduction bound}}
   \end{matrix}$
    }} (m-3-1) ;
    
\path[->,draw,thick]
   (m-3-1) edge [auto] node[labeled]  [scale=1] 
   {\text{$\begin{matrix}\text{\cref{prp:weak_to_strong_supplied basis}}
   \end{matrix}$
    }} (m-4-1) ;
    
\path[->,draw,thick]
   (m-4-1) edge [auto] node[labeled]  [scale=1] 
   {\text{$\begin{matrix}\text{Follows from the definition of $\NVR$}
   \end{matrix}$
    }} (m-5-1) ;    
%hacky-fix    
     
%   { [start chain] \chainin (m-1-2);
%     \chainin (m-7-1) [join={node[right,labeled] {} }];
    
%      }
\end{tikzpicture}
\end{adjustbox} 

\caption[Stages In Proving  \cref{thm:Pi^t_W is secure}]{
The above diagram summarizes the main stages of the proof of \cref{thm:Pi^t_W is secure}.}
\label{fig:proof idea}
\end{figure*}
 %-----------------------------------------------------------------------------%
\subsection{The Revealed Basis Setting} \label{sec:red_supplied_basis_unforgeability}
%-----------------------------------------------------------------------------%

\nomenclature[vr]{$\NVR$}{Verify and reveal oracle with parameter $\eta$ } 
In this section, we consider a modified version of the security game $\UnfExp^{\OneToken, \Vrfy}$ that we get by replacing the verification oracle $\Vrfy$ with an oracle $\NVRk$ that provides more information to the adversary. 

For a key $k=(a,b)$ representing an $n$ qubit length state $\ket{\Stamp}=H^b\ket{a}$, the oracle $\NVRk$ (Verify and Reveal) receives input of the form $(m,\sigma)$, runs the verification of $\CTMAC^\eta$ and returns a bit $0$ or $1$ accordingly, called the result bit. In addition, if the verification is successful, then along with the result bit, it also responds with the sets ${\Cons_m}=\{i\in [n]| b_i=m\}$, and the set $\Miss_{m,\sigma}=\{i\in \Cons_m\mid\sigma_i\neq a_i\}$. Queries whose result bit is $0$ are referred to as failed queries, and the others are referred to as successful queries.

The corresponding security game $\UnfExp^{\OneToken,\NVR}_{\adv,\CTMAC^\eta}(\secpar)$, described in \cref{exp: supplied_basis noise-tolerant} is defined by replacing the verification oracle in $\UnfExp^{\OneToken,\Vrfy}_{\adv,\CTMAC^\eta}(\secpar)$ with $\NVRk$. %to the adversary instead of the verification oracle, in addition to a single token.
 As $\NVRk$ is scheme specific, $\UnfExp^{\OneToken,\VR}_{\adv,\CTMAC^\eta}(\secpar)$ is defined only in the context of $\CTMAC^\eta$.

\begin{game}
\caption{
$\UnfExp^{\OneToken,\NVR}_{\adv,\CTMAC^\eta}(\secpar)$:}\label{exp: supplied_basis noise-tolerant}
\begin{algorithmic}[1] 
    \State $\Cadv$ generates a secret key $k$ by $\CTMAC^\eta.\keygen(\secpar)$.
    \State  $\adv$ is given input $1^\secpar$ as well as  a single access to $\CTMAC^\eta.\tokengen_k$  and classical oracle access to $ \NVRk$, which it can query a polynomial number of times. $\adv$ then outputs $(\sigma_0,\sigma_1)$.
\end{algorithmic}
The output of the game is defined to be $1$ if and only if $(1)$ $\CTMAC^\eta.\Vrfy_k(0,\sigma_0)=1 $ and  $(2)$ $\CTMAC^\eta.\Vrfy_k(1,\sigma_1)=1$.
\end{game}

We also define a stricter variant of the game in \cref{dfn: weak basis break}. 
\nomenclature[ForgeStar]{$\UnfExp^{\OneToken,{\NVR^*}}_{\adv,\CTMAC^\eta}$}{ A variant of the security game $\UnfExp$, specific to the scheme $\CTMAC^\eta$}
 \begin{definition} \label{dfn: weak basis break}
The security game, $\UnfExp^{\OneToken,{\NVR^*}}_{\adv,\CTMAC^\eta}(\secpar)$ is defined to be the same as $\UnfExp^{\OneToken,\NVR}_{\adv,\CTMAC^\eta}(\secpar)$, except for the added restriction that the value of the game is $0$ if the adversary performs two successful queries to $\NVRk$ for the same document.
\end{definition}

Next, we will show that an adversary in $\UnfExp^{\OneToken,{\NVR^*}}_{\adv,\CTMAC^\eta}(\secpar)$ does not gain additional information by repeated successful queries with respect to the same message, and hence, the two games are equivalent.
% The motivation behind the equivalence is that. This is formalized in the proposition below.

\begin{proposition}\label{prp:weak_to_strong_supplied basis}
For every (computationally unbounded) $\adv$ that wins $\UnfExp^{\OneToken,\NVR}_{\adv,\CTMAC^\eta}(\secpar)$ with probability $\epsilon(\secpar)$ using at most $q(\lambda)$ queries to $\NVRk$, there exists a (computationally unbounded) adversary $\Badv$ winning $\UnfExp^{\OneToken,\NVR^*}_{\Badv,\CTMAC^\eta}(\secpar)$ with probability $\epsilon(\secpar)$ making at most $q(\secpar)$ queries to $\NVRk$.
\end{proposition}

\begin{proof}\label{pf:prp:weak_to_strong_supplied basis}
Let $\adv$ be an adversary who makes at most $q(\secpar)$ queries to the oracle $\NVRk$ in the security game $\UnfExp^{\OneToken,\NVR}_{\adv,\CTMAC^\eta}(\secpar)$, and wins the game with probability $\epsilon(\secpar)$. %and makes $q(\secpar)$ queries to $\VR_k$.
We construct an adversary $\Badv$  to $\UnfExp^{\OneToken,\NVR}_{\adv,\CTMAC^\eta}(\secpar)$, who makes at most $q(\secpar)$ queries. %We will show that %We construct an adversary $\Badv$ for $\UnfExp^{\OneToken,\VR^*}_{\adv,\CTMAC}(\secpar)$ 
%$\Badv$ simulates $\adv$ as follows. 
%On receiving, 
$\Badv$ runs $\adv$ on  the security parameter she receives. $\Badv$ answers any oracle query $(m,\sigma)$ that $\adv$ makes to $\NVR$, as follows. 
  \begin{itemize}
  \item If $m$ has not been queried in any previous successful oracle query, $\Badv$ queries its own oracle with $(m,\sigma)$ and passes the answer to $\adv$. If the query is successful, $\Badv$ stores the result $\Cons_{m}$ and $\Miss_{\sigma,m}$, and uses these sets to uncover $x_m$, a string describing the quantum state of the token at the coordinates $\Cons_{m}$.
  \item Else, $\Badv$ uses the stored result to check if $\CTMAC^\eta.\Vrfy_k(m,\sigma)$ would pass (checking that $\sigma|_{\Cons_{m}}=x_m$ for all but $\eta\secpar$ indices), and answers accordingly to $\adv$, re-forwarding the stored result if necessary.
  \end{itemize}
       The view of $\adv$ is clearly identical to that in $\UnfExp^{\OneToken,\NVR}_{\adv,\CTMAC^\eta}(\secpar)$. Note that, once $\Badv$ successfully queries the oracle for a particular document, it would not query for the same document again. Hence, $\Badv$ wins $\UnfExp^{\OneToken,\NVR^*}_{\adv,\CTMAC^\eta}(\secpar)$ if and only if $\adv$ wins $\UnfExp^{\OneToken,\NVR}_{\adv,\CTMAC^\eta}(\secpar)$. Therefore, $\Badv$ wins $\UnfExp^{\OneToken,\NVR^*}_{\adv,\CTMAC^\eta}(\secpar)$ with a probability of $\epsilon(\secpar)$.
\end{proof}
Note that any adversary in $\UnfExp^{\OneToken,\Vrfy}_{\adv,\CTMAC^\eta}(
\secpar)$ (Game~\ref{exp:UnfExp}) is also an adversary in $\UnfExp^{\OneToken,\NVR}_{\adv,\CTMAC^\eta}(
\secpar)$. Therefore, if the winning probability for any $\QPT$ adversary $\adv$ in $\UnfExp^{\OneToken,\NVR}_{\adv,\CTMAC^\eta}(
\secpar)$ is bounded by $\epsilon(\secpar)$, then the same would hold for the security game $\UnfExp^{\OneToken,\Vrfy}_{\adv,\CTMAC^\eta}(
\secpar)$.

By taking the contra-positive of the \cref{prp:weak_to_strong_supplied basis}, we reach the following corollary:

\begin{corollary}

\label{cor: noise-tolerant weak  basis supply security implies noise-tolerant security}
If for every $\secpar$ and for every computationally unbounded $\Badv$ making at most $q(\secpar)$ queries to $\NVRk$ 
\[\Pr[\UnfExp^{\OneToken,{\NVR}^*}_{\Badv,\CTMAC^\eta}(\secpar) = 1] \leq \epsilon(\secpar),\]
then for every $\secpar$ and for any computationally unbounded adversary $\adv$ making at most $q(\secpar)$ queries to $\Vrfy_k$ \[\Pr[\UnfExp^{\OneToken,\Vrfy}_{\adv,\CTMAC^\eta}(\secpar) = 1] \leq \epsilon(\secpar),\]
where $\UnfExp^{\OneToken,\Vrfy}_{\adv,\CTMAC^\eta}(\secpar)$ and $\UnfExp^{\OneToken,{\NVR}^*}_{\Badv,\CTMAC^\eta}(\secpar)$ are the security games defined in Game~\ref{exp:UnfExp} and \cref{dfn: weak basis break}, respectively.
\end{corollary}

%-----------------------------------------------------------------------------%
\subsection{Reduction to Certified Deletion in a Noiseless Setting}\label{sec: reduction to weak certified deletion}
%-----------------------------------------------------------------------------%
In the previous section, we saw that \cref{cor: noise-tolerant weak  basis supply security implies noise-tolerant security}
 implies that proving unforgeability can be reduced to finding a bound on 
\[\Pr[\UnfExp^{\OneToken,\NVR^*}_{\adv,\CTMAC^\eta}(\secpar) = 1],\]
where $\UnfExp^{\OneToken,\NVR^*}$ is the game described in \cref{dfn: weak basis break}. We begin by proving the bound for the special case $\eta=0$, meaning the underlying scheme is the simpler $\CTMAC$. 
The bound for the special case is achieved by a reduction to a weaker notion of security of a Quantum Encryption with Certified Deletion scheme (QECD). Hence, we take a short detour to QECD~\cite{BI20}. 
%To achieve a bound, we take a short detour inspired by:

\begin{definition}[Quantum Encryption with Certified Deletion,~\cite{BI20}] \label{def:qunatum encryption with certified deletion}A quantum encryption with certified deletion scheme,
 or $\QECD$, consists of five $\QPT$ algorithms: $\keygen,\Enc, \Dec, \Del$ and $ \Vrfy$ fulfilling the following:\footnote{We defer here from the definition given previously. In~\cite{BI20}, the length of the encrypted message, $n$, is an integer independent of $\secpar$. In addition,~\cite{BI20} uses density operator formalism, while we use pure states to preserve consistency with the rest of the paper.}
\begin{enumerate}
    \item On input $1^\secpar$, where $\secpar$ is the security parameter, the algorithm $\keygen$ outputs a classical key $k$.
    \item $\Enc_k(a)$ receives a classical string $a\in\Bitspace^\secpar$ and outputs a quantum state $\ket{c}$ denoting its encryption.
    \item $\Dec_k(\ket{c})$ receives a quantum state $\ket{c}$ and outputs a classical string $a'\in \Bitspace^\secpar$.
    \item $\Del(\ket{c})$ receives a quantum state and returns a classical string $cer$, which we will refer to as the deletion certificate.
    \item $\Vrfy_k(cer)$ outputs a Boolean answer.
\end{enumerate}
\end{definition}

The motivation for a quantum encryption with certified deletion scheme is that an adversary in possession of some cipher-text could supply a trusted authority with a certificate that the adversary has deleted the cipher, which the authority can then verify. If verification of the certificate succeeds, then the trusted authority is assured that even if the key is later leaked, or intentionally given to the adversary, the adversary cannot recover the data previously encrypted. %Informally, the security game is similar in spirit to the adversary in $\UnfExp^{\OneToken,\VR^*}_{\adv,\CTMAC}$. 
The two $\QECD$ schemes $E_m$, for $m\in \{0,1\}$ that are of interest to us are given in \cref{alg: E_m}.
\begin{algorithm*}
    \caption{$E_m$ - A $\QECD$ scheme (for $m\in \Bitspace$)}
    \label{alg: E_m}
    \begin{algorithmic}[1] % The number tells where the line numbering should start
        \Procedure{$\keygen$}{$1^\secpar$}
                 \State $k\sample\Bitspace^\secpar$.
            \State \textbf{Return} $k$.
        \EndProcedure
    \end{algorithmic}
    \begin{algorithmic}[1] % The number tells where the line numbering should start
        \Procedure{$\Enc_k$}{$a$}
            \State \textbf{Return} $H^k(\ket{a})$.
        \EndProcedure

    \end{algorithmic}
	\begin{algorithmic}[1] % The number tells where the line numbering should start
        \Procedure{$\Dec_k$}{$\ket{c}$}
            \State Compute $H^k\ket{c}$ to obtain $\ket{cc}$.
            \State Measure $\ket{cc}$ to obtain $a'$.
            \State \textbf{Return} $a'$.
                \EndProcedure

    \end{algorithmic}
  	\begin{algorithmic}[1] % The number tells where the line numbering should start
        \Procedure{$\Del$}{$\ket{c}$}
            \State Compute $(H^m)^{\otimes\secpar}\ket{c}$to obtain $\ket{cc}$.
            \State Measure $\ket{cc}$ to obtain $cer$.
            \State \textbf{Return} $cer$.
            \EndProcedure
    \end{algorithmic}
      	\begin{algorithmic}[1] % The number tells where the line numbering should start
        \Procedure{$\Vrfy_k$}{$cer$}
            \State Define $\Cons_m=\{i\in [\secpar]| k_i=m\}$.
            \If{$cer=a|_{\Cons_m}$}
            \State \textbf{Return} $1$.
            \Else
                \State \textbf{Return} $0$.
            \EndIf
            \EndProcedure
    \end{algorithmic}
\end{algorithm*}

 As per the original definition in~\cite{BI20}, a QECD scheme is secure if no adversary can distinguish between the encryption of a message of his choice and that of the $0^\secpar$ string, even if she is provided with the secret key after it is verified that she deleted the cipher. These schemes are not certified deletion secure in that sense.\footnote{A simple adversary could choose a string differing by $1$ bit from the $0$ string as $m_1$: measuring honestly results in a certificate $r$ that is then provided to the adversary. If the key is $1$ at the corresponding coordinate (meaning the corresponding qubit was in the Hadamard basis), then the adversary outputs $1$, or else they output a uniformly random guess. It is easy to see that the attacker wins with a probability of $\frac{3}{4}$.} However, for our needs, it is sufficient to provide a bound on the winning probability in a stricter security game (Game~\ref{exp: weak certified deletion}) that is tailor-made for our needs.
\begin{savenotes}
\begin{game}
\caption{Weak Certified Deletion Game $\exWCerDel_{\adv,\Pi}(\secpar)$:}\label{exp: weak certified deletion}  
\begin{algorithmic}[1]
    \State The challenger $\Cadv$ runs $\keygen(1^\secpar)$ to generate $k$, and uniformly samples $a\in\Bitspace^\secpar$. $\Cadv$ then sends $\ket{c}=\Enc_k(a)$.
   \State $\adv$ sends $\Cadv$ some string $cer$.
    \State $\Cadv$ computes  $V\gets\Vrfy_k(cer)$, and  then sends $k$ to the adversary.
    \State $\adv$ outputs $a'$.
\end{algorithmic}
 We say the output of the game is 1 if and only if $V=1$ and $a=a'$.
\end{game}
\end{savenotes}

In  \cref{sec: analysis of weak certified deletion scheme}, we prove the following proposition, based on semi-definite programming methods.

\begin{restatable}{proposition}{weakCertSecure}\label{cor:weak_cert_secure}
The optimal success probability for an adversary in\\ $\exWCerDel_{\adv,E_0}(\secpar)$ (see  Game~\ref{exp: weak certified deletion} and \cref{alg: E_m}), as well as in $\exWCerDel_{\adv,E_1}(\secpar)$ is  $\alpha^\secpar$, where $\alpha= \cos^2(\frac{\pi}{8})$.
\end{restatable}

\nomenclature[vr]{$\VR$}{$\VR^0$}
Next, we exhibit a reduction from the game $\UnfExp^{\OneToken,{\VR^0}^*}_{\adv,\CTMAC}(\secpar)$ to the weak certified deletion game (Game~\ref{exp: weak certified deletion}) against the schemes $E_0$, $E_1$ (\cref{alg: E_m}), to conclude the bound on %the and the game $\UnfExp^{\OneToken,\VR^*}_{\adv,\CTMAC}(\secpar)$ suffices to find a bound on
the winning probability in $\UnfExp^{\OneToken,{\VR^0}^*}_{\adv,\CTMAC}(\secpar)$.

Note that for $\eta=0$, $\Miss_{m,\sigma}$ is always the empty set and can be disregarded.
For the sake of succinctness, we will omit $\eta$ from the superscript $\NVR$ and use $\VR$ to denote  $\VR^0$ from now onward.  

\begin{proposition}
\label{prp:weak certififed deletion security implies weak basis supply security}
For every infinite set of integers $S$ and a  computationally unbounded $\adv$ making $q(\secpar)$  queries to $\VR_{k}$, such that for all $\secpar\in S$
\[\Pr[\UnfExp^{\OneToken,{\VR}^*}_{\adv,\CTMAC}(\secpar) = 1]= \epsilon(\secpar),\]
where 
$\UnfExp^{\OneToken,{\NVR}^*}_{\Badv,\CTMAC^\eta}(\secpar)$ is the game defined in
 \cref{dfn: weak basis break}, either there is a computationally unbounded adversary $\Badv_0$ such that for infinitely many $\secpar\in S$,
 \[\Pr[\exWCerDel_{\Badv_0,E_0}(\secpar) = 1] \geq \frac{1}{2\binom{q(\secpar)+2}{2}}\epsilon(\secpar),\]
or there is a computationally unbounded adversary $\Badv_1$ such that for infinitely many $\secpar\in S$,
 \[\Pr[\exWCerDel_{\Badv_1,E_1}(\secpar) = 1] \geq \frac{1}{2\binom{q(\secpar)+2}{2}}\epsilon(\secpar).\]
\end{proposition}
\begin{proof}
Let $\adv$ be an arbitrary adversary for the $\UnfExp^{\OneToken,\VR^*}_{\adv,\CTMAC}(\secpar)$ game (see \cref{dfn: weak basis break}) who makes $q(\secpar)$ queries to $\VR$, and  wins with probability $\epsilon(\secpar)$. Since verification is deterministic, without loss of generality, it can be assumed that if $\adv$ wins $\UnfExp^{\OneToken,\VR^*}_{\adv,\CTMAC}(\secpar)$, it always makes two successful queries for different documents, perhaps at the cost of forcing $\adv$ to make at most two more queries.  
Let $A_m(\secpar)$ for $m\in \Bitspace$ be the event that the adversary sent two successful queries corresponding to two different documents and the first of the two successful queries made by $\adv$ is for the document $m$, and let $p_m(\secpar)\equiv\Pr[A_m(\secpar)]$. Since $p_0(\secpar)+p_1(\secpar)=\epsilon(\secpar)$, there exists a bit $\hat{m}$, such that $p_{\hat{m}}(\secpar)\geq \frac{1}{2}\epsilon(\secpar)$ for infinitely many $\secpar\in S$. Denote the corresponding set of infinitely many $\secpar$ as $G$.
We construct an adversary $\Badv_{\hat{m}}$ against $E_{\hat{m}}$ as follows. Given an encryption $\ket{c}=H^k\ket{a}$ for message $a\in\Bitspace^{\secpar}$ and key  $k\in\Bitspace^{\secpar}$, $\Badv_{\hat{m}}$ would pass $\ket{c}$ to $\adv$ as a token. $\Badv_{\hat{m}}$ would then simulate $\adv$, and uniformly guess the two indices of its two successful queries $i<j$. On the $i^{th}$ query, $\Badv_{\hat{m}}$ will forward the queried signature $\sigma_{\hat{m}}$ to the challenger as a certificate of deletion. Let $\Cons_{\hat{m}}\equiv\{i\in[\secpar]| k_i={\hat{m}}\}$. The challenger will check that $a|_{\Cons_{\hat{m}}}=\sigma_{\hat{m}}|_{\Cons_{\hat{m}}}$, which is exactly what is expected of a valid signature for ${\hat{m}}$. Hence, for a correct guess of the indices, the signature would be deemed valid, and $\Badv_{\hat{m}}$ would get the key $k$ from the challenger. $\Badv_{\hat{m}}$ can then use $k$ to construct $\Cons_{\hat{m}}$ by itself, and answer $\adv$'s successful oracle query. $\Badv_{\hat{m}}$ will continue to reject all other queries made by $\adv$ until the $j^{th}$ query. On the $j^{th}$ query, it will accept the sent string $\sigma_{1-\hat{m}}$. $\Badv_{\hat{m}}$ could then reassemble the string $a'$ as follows.
\[a'(i)=\begin{cases}\sigma_{\hat{m}}(i) & \mbox{if } i\in \Cons_{\hat{m}}\\
\sigma_{1-{\hat{m}}}(i) & \mbox{otherwise, }
\end {cases}\]
and send $a'$ to the challenger.

$\Badv_{\hat{m}}$ guesses the indices correctly with a probability greater than $\frac{1}{\binom {q(\secpar)+2} {2}}$  (since up to two oracle queries may have been added). In such an event, the view of $\adv$ is the same as in $\UnfExp^{\OneToken,\VR^*}_{\adv,\CTMAC}(\secpar)$, meaning  for all $\secpar\in G$,  with a probability greater than $\frac{1}{2}\epsilon(\secpar)$, $\adv$ outputs two successfully verified signatures, and the first successful query is for the document $\hat{m}$. 

In that case, $a|_{\Cons_0}=\sigma_0|_{\Cons_0}$ and $a|_{\Cons_1}=\sigma_1|_{\Cons_1}$ and, hence, $a'=a$ (that is, $\Badv$ wins). Therefore, $\Badv_{\hat{m}}$ could win the $\exWCerDel_{\Badv_{\hat{m}},E_{\hat{m}}}(\secpar)$ with a probability of  at least $\frac{1}{2\binom {q(\secpar)+2} {2}}\epsilon(\secpar)$ for all $\secpar\in G$.
\end{proof}

An immediate result is the following:
\begin{restatable}{proposition}{weaksb}

\label{prp: weak supplied basis security}
For every computationally unbounded $\adv$ making $q(\secpar)$ queries to $\VR_{k}$, there exists $\secpar_0$ such that for all $\secpar>\secpar_0$:
\[\Pr[\UnfExp^{\OneToken,\VR^*}_{\adv,\CTMAC}(\secpar) = 1]\leq 2\alpha^\secpar\binom{q(\secpar)+2}{2},\] 
where $\alpha= \cos^2(\frac{\pi}{8})$ and $\UnfExp^{\OneToken,{\NVR}^*}_{\Badv,\CTMAC^\eta}(\secpar)$ is the game defined in
 \cref{dfn: weak basis break}.
\end{restatable}

\begin{proof}
Assume towards contradiction that this is not so, then by \cref{prp:weak certififed deletion security implies weak basis supply security}, there exists  $\Badv_0$ winning $\exWCerDel_{\Badv_0,E_0}$ with a probability greater than $\alpha^\secpar$ for infinitely many $\secpar$, or there exists  $\Badv_1$ winning $\exWCerDel_{\Badv_1,E_1}$ with a probability greater than $\alpha^\secpar$ for infinitely many $\secpar$. In either case, this contradicts \cref{cor:weak_cert_secure}.
\end{proof}
The bound above also holds against any adversary $\adv$ that makes \emph{at most} $q(\secpar)$ queries to $\VR_k$
% \[\Pr[\UnfExp^{\OneToken,\VR^*}_{\adv,\CTMAC}(\secpar)=1]\leq 2\alpha^\secpar\binom{q(\secpar)+2}{2},\]
because $\adv$ can always be assumed to make exactly $q(\secpar)$ queries by forcing $\adv$ to submit extra failing queries. If $q(\secpar)$ is polynomial, this bound is negligible in $\secpar$. Combined with \cref{cor: noise-tolerant weak  basis supply security implies noise-tolerant security}, this bound suffices to prove that $\CTMAC$ is $\UnfDef^{\OneToken,\Vrfy}$ (\cref{thm:Pi_W is secure}).
%---------------------------------------------------------------------------------------------------------------------------------------------------------------------------------------------------------------------------------------------------------------------%
\subsection{Noisy to Noiseless Reduction in a Revealed Basis Setting}\ifnum\llncs=0\label{sec: Noisy to Noiseless Reduction}\fi
\ifnum\llncs=1\label{sec:noise-tolerant unforgeability}\fi

%---------------------------------------------------------------------------------------------------------------------------------------------------------------------------------------------------------------------------------------------------------------------%

We extend the results to the noise-tolerant setting by giving a reduction from the noise-tolerant revealed basis setting to the noise-sensitive revealed basis setting. 

\begin{restatable}{lemma}{NoiseReductionLemma}
\label{lem:noise reduction lemma}
Let $f:\mathbb{N}\rightarrow \mathbb{N}$ be a function such that $f(\secpar)< \secpar$, and let  $\adv$ be an adversary making $q(\secpar)$ queries to $\NVRk$ and denote
\[W_{\eta}(\secpar)\equiv\Pr[\UnfExp^{\OneToken,\NVR^*}_{\adv,\CTMAC^\eta}(\secpar) = 1], \]
where 
$\UnfExp^{\OneToken,{\NVR}^*}_{\Badv,\CTMAC^\eta}(\secpar)$ is the game defined in
 \cref{dfn: weak basis break}.
There exists an adversary $\Badv$ making at most $q(\secpar)+2$ queries to $\VR_k$ such that for every $\secpar$ and any $0\leq\eta\leq1$,
\[\Pr[\UnfExp^{\OneToken,\VR^*}_{\Badv,\CTMAC^\eta}(f(\secpar))= 1] \geq \left(1-\frac{2\eta \secpar}{\secpar-f(\secpar)}\right)^{f(\secpar)}W_{\eta}(\secpar).\]
\end{restatable}

\ifnum\llncs=0
\begin{proof}
Let $\adv$ be an adversary for $\UnfExp^{\OneToken,\NVR^*}_{\adv,\CTMAC^\eta}(\secpar)$ (\cref{dfn: weak basis break}), as described in \cref{lem:noise reduction lemma}. For simplicity of description, we assume that that $\adv$ always makes at least two queries to $\NVR$, and that $\adv$ always submits the first two successful queries it makes to $\NVR$. We also assume that if $\adv$ wins, it always makes two successful queries for different documents, perhaps at the cost of forcing $\adv$ to make at most two more queries. We construct an adversary $\Badv$ for the corresponding noise-sensitive security game $\UnfExp^{\OneToken,\VR^*}_{\Badv,\CTMAC}(f(\secpar))$ (\cref{dfn: weak basis break}) as follows.

\nomenclature[Iset]{$\Iset$}{A set of coordinates chosen by the adversary in the proof of \cref{lem:noise reduction lemma}}
Let $\Iset\in \binom{[\secpar]}{f(\secpar)}$ be a set of indices, and also denote $\overline{\Iset}\equiv[\secpar]\backslash \Iset$. For every $\Iset\in \binom{[\secpar]}{f(\secpar)}$, we define $\Badv^\Iset$, an adversary for the noise-sensitive scheme $\CTMAC^\eta$ as follows. Upon receiving a $f(\secpar)$-qubit token $\ket{\Stamp}$ from $\Cadv$, $\Badv^\Iset$ generates another token $\ket{\widetilde{\Stamp}}$ consisting of $\secpar-f(\secpar)$ qubits. This is done by running $\CTMAC^\eta.\keygen(1^{\secpar-f(\secpar)})$ to obtain a key $\widetilde{k}=(\widetilde{a},\widetilde{b})$ followed by running $\CTMAC^\eta.\tokengen_{\widetilde{k}}$ to get the token. $\Badv^\Iset$ then assembles a $\secpar$-length token $\ket{\Stamp'}$ by placing the qubits of  $\ket{\Stamp}$ at coordinates $\Iset$, and the qubits of $\ket{\widetilde{\Stamp}}$ at coordinates $\overline{\Iset}$.  This also defines a fictional key $k'\equiv(a',b')$, composed of the strings $a', b'$, which are derived by placing $\widetilde{a}$,$\widetilde{b}$ at coordinates $\overline{\Iset}$ and $a$,$b$ at coordinates $\Iset$. $\Badv^\Iset$  simulates $\adv$ on the token $\ket{\Stamp'}$.

 For every alleged query $(m,\sigma)$ made to $\NVR$ by $\adv$ in the simulated game of $\Badv^\Iset$, let \[{\Cons'}^{\Iset}_m\equiv\{i\in [\secpar] \mid b'_i=m\},\] \[{\Miss'}^{\Iset}_{m,\sigma}\equiv\{i\in {\Cons'}^{\Iset}_m\mid a'_i\neq\sigma_i\},\] which are defined with respect to the fictional key $k'$. Additionally, let
 \[\widetilde{\Cons}^{\overline{\Iset}}_{m}\equiv\{i\in \overline{\Iset}\mid {b'_i}=m\},\]
 \[\widetilde{\Miss}^{\overline{\Iset}}_{m,\sigma}\equiv\{i\in \widetilde{\Cons}^{\overline{\Iset}}_{m}\mid {a'_i}\neq \sigma_i\}.\]
 We say a query $(m,\sigma)$ to the oracle $\NVR$ in the true game $\UnfExp^{\OneToken,\NVR^*}_{\adv,\CTMAC^\eta}(\secpar)$ is "good" if $\abs{\Miss_{m,\sigma}}\leq\eta\secpar$. Similarly, we say such a query in the simulation of $\adv$ within $\Badv^\Iset$ is "good" if $\abs{{\Miss'}^{\Iset}_{m,\sigma}}\leq\eta\secpar$.
 
We further describe the behavior of $\Badv^\Iset$: whenever $\adv$ queries $\NVR$, with a query $(m,\sigma)$, $\Badv^\Iset$ runs $\CTMAC.\Vrfy_{\widetilde{k}}(m|_{\overline{\Iset}},\sigma|_{\overline{\Iset}})$ and constructs the corresponding sets $\widetilde{\Cons}^{\overline{\Iset}}_{m}$,  $\widetilde{\Miss}^{\overline{\Iset}}_{m,\sigma}$.
If $\abs{\widetilde{\Miss}^{\overline{\Iset}}_{m,\sigma}}>\eta \secpar$, meaning there are more than $\eta\secpar$ errors at coordinates in $\overline{\Iset}$ that $\Badv^\Iset$ rejects. Otherwise, it queries $\VR$ with $(m|_{\Iset},\sigma|_{\Iset})$, and accepts if and only if $\VR$ accepts. %The challenger for $\Badv^\Iset$ is with respect to the noisy scheme, and hence
Note that if $\VR$ accepts, $\Badv^\Iset$ would receive the set $\Cons^{\Iset}_{m}\equiv\{i\in \Iset | b'_i=m\}$, and since the challenger for $\Badv^\Iset$ is with respect to the noise-sensitive scheme, $\sigma|_{\Iset}$ must be consistent on $\Cons^{\Iset}_{m}$. $\Badv^{\Iset}$ could thus reliably accept the query and forward to $\adv$ the sets ${\Cons'}^{\Iset}_{m}=\widetilde{\Cons}^{\overline{\Iset}}_{m}\cup \Cons^{\Iset}_{m}$ and ${\Miss'}^{\Iset}_{m,\sigma}=\widetilde{\Miss}^{\overline{\Iset}}_{m,\sigma}$.
Finally, when $\adv$ submits two signed documents $(\sigma_0,\sigma_1)$, $\Badv^{\Iset}$ would submit to its challenger $({\sigma_0}|_{\Iset},{\sigma_1}|_{\Iset})$. The important observation is that $\Badv^\Iset$ would never accept a query $(m,\sigma)$, unless it is "good", even though it might reject queries that are "good".
The adversary $\Badv$, is defined as the adversary which uniformly selects $\Iset\in \binom{[\secpar]}{f(\secpar)}$, and runs $\Badv^\Iset$. Informally, we claim that in events where $\adv$ wins $\UnfExp_{\adv,\CTMAC^\eta}^{\OneToken,\NVR}(\secpar)$, the view of $\adv$ is identical to the view when run within the simulation of $\Badv^\Iset$, for sufficiently many choices of $\Iset$.
  
 \begin{definition}
The random variables $M_1$ and $S_1$ denote the document and signature, respectively, in the first "good" query, in the true game $\UnfExp^{\OneToken,\NVR^*}_{\adv,\CTMAC^\eta}$. If such a query does not exist, $M_1,S_1$ are the document and signature, respectively, in the next-to-last query made by $\adv$ . Similarly, $M_2, S_2$ denote the document and signature, respectively, in the second "good" query. If such a query does not exist, $M_2,S_2$ are the document and signature, respectively, in the last query made by $\adv$. The random variable $J_1$ represents the indices in which $S_1$ is inconsistent with the secret key, that is, the set $\Miss_{M_1,S_1}$, and likewise, $J_2$ is the set $\Miss_{M_2,S_2}$. The random variable $T$ represents the token $\ket{\Stamp'}$ that $\adv$ receives from the challenger.
 In addition, for any $\Iset\in \binom{[\secpar]}{f(\secpar)}$, we define corresponding random variables $M^\Iset_1,S^\Iset_1,M^I_2,S^\Iset_2,J^\Iset_1,J^\Iset_2,T^\Iset$, denoting the same, but when $\adv$ is simulated within $\Badv^\Iset$, with $J^\Iset_i$  defined with regard to ${\Miss'}^{\Iset}_{M^\Iset_i,S^\Iset_i}$.  
  \end{definition}
  
 We emphasize, that the value of $M^\Iset_1,S^\Iset_1,M^\Iset_2,S^\Iset_2,J^\Iset_1,J^\Iset_2$ is defined with regard to "good" queries, and regardless of the response of $\Badv^\Iset$ to the query, which could potentially reject a "good" query.
%  if the set ${\Miss'}^I_{m,\sigma}$ is not empty.
  
 Denote:
 \[\mathcal{E}_{good}=\{\adv \text{ makes two "good" queries in } \UnfExp_{\adv,\CTMAC^\eta}^{\OneToken,\NVR^*}\},\]
   \[\mathcal{E}^{j_1,j_2,b,\ket{\Stamp}}_{view}=\{J_1=j_1, J_2=j_2, M_1=b, M_2=1-b,T=\ket{\Stamp} \}, \]
   \[\mathcal{\overset{\Iset}{F}}_{good}=\{\adv \text{ makes two "good" queries in the simulation within $\Badv^I$}\}\]
   \[\mathcal{\overset{\Iset}{F}}^{j_1,j_2,b,\ket{\Stamp}}_{view}=\{J^\Iset_1=j_1, J^\Iset_2=j_2, M^\Iset_1=b, M^\Iset_2=1-b,T^\Iset=\ket{\Stamp} \}. \]
  \begin{lemma}\label{lem:view lemma}
  For all $\Iset\in \binom{[\secpar]}{f(\secpar)}$, $j_1,j_2\subseteq \mathcal{\overline{\Iset}}$, $b\in \Bitspace$ and $\ket{\Stamp}\in \QBitspace^\secpar$ such that $\abs{j_1}\leq\eta\secpar,\abs{j_2}\leq\eta\secpar$:
  \begin{align}\Pr[\mathcal{E}_{good}\wedge \mathcal{E}^{j_1,j_2,b,\ket{\Stamp}}]
  =\Pr[\mathcal{\overset{\Iset}{F}}_{good}\wedge \mathcal{\overset{\Iset}{F}}^{j_1,j_2,b,\ket{\Stamp}}_{view}].\end{align}
  \end{lemma}
\begin{proof} 
Fix $\Iset\in \binom{[\secpar]}{f(\secpar)}$ arbitrarily and $j_1,j_2$, as described in the lemma.
 Conditioned on the events where $T=\ket{\Stamp}$ and $T^\Iset=\ket{\Stamp}$ for some  fixed value $\ket{\stamp}$, in the true game $\UnfExp^{\OneToken,\NVR^*}_{\adv,\CTMAC^\eta}(\secpar)$ and in the simulation within $\Badv^\Iset$, respectively, it is easy to argue by induction  that for every "step" until a  "good" query occurs, the view remains the same in both the true game and the simulation. We can thus derive that for any $j_1,b$: \begin{align}&\Pr[J_1^\Iset=j_1,M_1^\Iset=b \mid T^\Iset=\ket{\Stamp}]=\Pr[J_1=j_1,M_1=b \mid T=\ket{\Stamp}].\end{align}

Continuing with similar reasoning, conditioned on $J_1=j_1,M_1=b, T=\ket{\Stamp}$ in the true game and $J_1^\Iset=j_1,M_1^\Iset=b,T^\Iset=\ket{\Stamp}$ in the simulation, in addition to  conditioning on a fixed index of the first "good" query (if such a query exists) in both, we can argue that:

\begin{align}\Pr&[J_2^\Iset=j_2, M_2^\Iset=1-b\mid J_1^\Iset=j_1,M_1^\Iset=b, T^\Iset=\ket{\Stamp}]\\=&\Pr[J_2=j_2, M_2=1-b\mid J_1=j_1,M_1=b, T=\ket{\Stamp}].\end{align}
This is because the conditioning on $j_1$ ensures that the first "good" query (if such exists) will be answered positively even in the simulation: $j_1\subseteq \overline{\Iset}$, meaning the signature corresponding to the query is consistent with $b'$ on $\Iset$, and contains less than $\eta\secpar$ errors on $\overline{\Iset}$ ; hence, the same inductive reasoning  on the view applies.

Similarly 
\begin{align}\Pr&[\mathcal{\overset{\Iset}{F}}_{good}\mid  J_2^\Iset=j_2, M_2^\Iset=1-b, J_1^\Iset=j_1,M_1^\Iset=b, T^\Iset=\ket{\Stamp}]\\=&\Pr[\mathcal{E}_{good}\mid  J_2=j_2, M_2=1-b, J_1=j_1,M_1=b, T=\ket{\Stamp}].\end{align}

Finally, observe that the token $\ket{\Stamp}$ that $\adv$ receives is distributed uniformly in both the true game and in the simulation. Using standard conditional probability, we break down the probability of the event $\mathcal{E}_{good}\wedge \mathcal{E}^{j_1,j_2,b,\ket{\Stamp}}_{view}$ into simpler terms, and then reassemble them.
\begin{align}\label{eq: view conditioning}
\Pr&[\mathcal{E}_{good}\wedge \mathcal{E}^{j_1,j_2,b,\ket{\Stamp}}_{view}]\\
%
%   &=\Pr[\mathcal{E}_{good},  J_2=j_2, M_2=1-b,J_1=j_1,M_1=b,T=\ket{\Stamp}]\\
%   &=\Pr[\mathcal{E}_{good},  J_2=j_2, M_2=1-b,J_1=j_1,M_1=b\mid T=\ket{\Stamp}]\cdot\Pr[T=\ket{\Stamp}]\\
%   &=\Pr[\mathcal{E}_{good},  J_2=j_2, M_2=1-b\mid J_1=j_1,M_1=b, T=\ket{\Stamp}]\\&\cdot\Pr[J_1=j_1,M_1=b \mid T=\ket{\Stamp}]\cdot\Pr[T=\ket{\Stamp}]\\
%
  &=\Pr[\mathcal{E}_{good}\mid  J_2=j_2, M_2=1-b, J_1=j_1,M_1=b, T=\ket{\Stamp}]\\&\cdot\Pr[J_2=j_2, M_2=1-b\mid J_1=j_1,M_1=b, T=\ket{\Stamp}]\cdot\Pr[J_1=j_1,M_1=b \mid T=\ket{\Stamp}]\cdot\Pr[T=\ket{\Stamp}]\\
&=\Pr[\mathcal{\overset{\Iset}{F}}_{good}\mid  J_2^\Iset=j_2, M_2^\Iset=1-b, J_1^\Iset=j_1,M_1^\Iset=b, T^\Iset=\ket{\Stamp}]\\&\cdot\Pr[J_2^\Iset=j_2, M_2^\Iset=1-b\mid J_1^\Iset=j_1,M_1^\Iset=b, T^\Iset=\ket{\Stamp}]\cdot\Pr[J_1^\Iset=j_1,M_1^\Iset=b \mid T^\Iset=\ket{\Stamp}]\cdot\Pr[T^\Iset=\ket{\Stamp}]\\
%
%   &=\Pr[\mathcal{\overset{\Iset}{F}}_{good},J_2^\Iset=j_2, M_2^\Iset=1-b\mid J_1^\Iset=j_1,M_1^\Iset=b, T^\Iset=\ket{\Stamp}]\\&\cdot\Pr[J_1^\Iset=j_1,M_1^\Iset=b \mid T=^\Iset\ket{\Stamp}]\cdot\Pr[T^\Iset=\ket{\Stamp}]\\
%   &=\Pr[\mathcal{\overset{\Iset}{F}}_{good},J_2^\Iset=j_2, M_2^\Iset=1-b, J_1^\Iset=j_1,M_1^\Iset=b\mid T^\Iset=\ket{\Stamp}]\cdot\Pr[T^\Iset=\ket{\Stamp}]\\
%   &=\Pr[\mathcal{\overset{\Iset}{F}}_{good},J_2^\Iset=j_2, M_2^\Iset=1-b, J_1^\Iset=j_1,M_1^\Iset=b, T^\Iset=\ket{\Stamp}]\\
  %
  &=\Pr[\mathcal{\overset{\Iset}{F}}_{good}\wedge \mathcal{\overset{\Iset}{F}}^{j_1,j_2,b,\ket{\Stamp}}_{view}],
\end{align}
which proves \cref{lem:view lemma}.
\end{proof}
Notice that for any $\Iset,j_1,j_2,b,\ket{\Stamp}$, as described in \cref{lem:view lemma},
\[\mathcal{\overset{I}{F}}_{good}\wedge \mathcal{\overset{I}{F}}^{j_1,j_2,b,\ket{\Stamp}}_{view}\subseteq \{\UnfExp^{\OneToken,\Vrfy}_{\Badv^\Iset,\CTMAC}(\secpar)=1\},\]
which implies, 
\begin{equation}\label{eq:union of goods contained in simulation win}
\bigcupdot_{\mathclap{\substack{j_1,j_2\subseteq \overline{\Iset}\\ \abs{j_1},\abs{j_2}\leq \eta\secpar\\ b\in\Bitspace\\\ket{\Stamp}\in \QBitspace^\secpar}}}\quad\left(\mathcal{\overset{I}{F}}_{good}\wedge \mathcal{\overset{I}{F}}^{j_1,j_2,b,\ket{\Stamp}}_{view}\right)\subseteq\{ \UnfExp^{\OneToken,\Vrfy}_{\Badv^\Iset,\CTMAC}(\secpar)=1\}.
\end{equation}
%
\begin{comment}% A detailed explanation
% The above holds due to the following reason. The event $\mathcal{\overset{I}{F}}_{good}\wedge \mathcal{\overset{I}{F}}^{j_1,j_2,b,\ket{\Stamp}}_{view}$ implies that $\adv$ has made two "good" queries for 2 distinct documents within the $\Badv^\Iset$ simulation, with the errors contained in the set $\overline{\Iset}$. Hence, $\Badv^\Iset$ would answer each of those $2$ queries with acceptance. Therefore, in the event $\mathcal{\overset{I}{F}}_{good}\wedge \mathcal{\overset{I}{F}}^{j_1,j_2,b,\ket{\Stamp}}_{view}$, $\adv$ would submit to $\Badv^I$ the $2$ signatures $S^{\Iset}_1,S^{\Iset}_2$, corresponding to distinct documents, such that none of the two signatures contain errors in the set $\Iset$. $\Badv^{\Iset}$ will then submit the signatures $S^{\Iset}_1|_{\Iset},S^{\Iset}_2|_{\Iset}$, which do not contain any errors. Hence, $\Cadv$ would accept and $\Badv^\Iset$ wins in the event $\mathcal{\overset{I}{F}}_{good}\wedge \mathcal{\overset{I}{F}}^{j_1,j_2,b,\ket{\Stamp}}_{view}$. 
% The disjoint union in \cref{eq:union of goods contained in simulation win} holds since $\mathcal{\overset{I}{F}}^{j_1,j_2,b,\ket{\Stamp}}_{view}$ are disjoint for distinct choices of $(j_1,j_2,b,\ket{\Stamp})$.
\end{comment}
In addition, by the assumption that if $\adv$ wins, it always makes two successful queries for different documents:
\begin{equation}\label{eq:Union of good is a win}
\{\UnfExp^{\OneToken,\NVR}_{\adv,\CTMAC^\eta}(\secpar)=1\}=\bigcupdot_{\mathclap{\substack{ b\in\Bitspace \\ \ket{\Stamp}\in \QBitspace^\secpar\\\abs{j_1},\abs{j_2}\leq \eta\secpar\\ j_1,j_2\subseteq[\secpar]}}}\left( \mathcal{E}_{good}\wedge \mathcal{E}^{j_1,j_2,b,\ket{\Stamp}}\right)
\end{equation}
Observe that the success probability of $\Badv$, is the average over the winning probabilities of $\Badv^\Iset$ for all choices of $\Iset\in \binom{[\secpar]}{f(\secpar)}$. Combining this with \cref{eq:union of goods contained in simulation win}, and \cref{lem:view lemma}, we get,
\begin{align}
\Pr&[\UnfExp^{\OneToken,\VR^*}_{\Badv,\CTMAC}(\secpar)=1]\\
&=\frac{1}{\binom{\secpar}{f(\secpar)}}\sum_{\Iset\in \binom{[\secpar]}{f(\secpar)}}\Pr[\UnfExp^{\OneToken,\VR^*}_{\Badv^\Iset,\CTMAC}(\secpar)=1]\\
&\geq\frac{1}{\binom{\secpar}{f(\secpar)}}\sum_{\Iset\in \binom{[\secpar]}{f(\secpar)}}\Pr[\bigcupdot_{\mathclap{\substack{j_1,j_2\subseteq \overline{\Iset}\\ \abs{j_1},\abs{j_2}\leq \eta\secpar\\ b\in\Bitspace\\\ket{\Stamp}\in \QBitspace^\secpar}}}\mathcal{\overset{I}{F}}_{good}\wedge \mathcal{\overset{I}{F}}^{j_1,j_2,b,\ket{\Stamp}}_{view}]& \text{by \cref{eq:union of goods contained in simulation win}}\\
&=\frac{1}{\binom{\secpar}{f(\secpar)}}\sum_{\Iset\in \binom{[\secpar]}{f(\secpar)}}\quad\quad\sum_{\mathclap{\substack{j_1,j_2\subseteq \overline{\Iset}\\ \abs{j_1},\abs{j_2}\leq \eta\secpar\\ b\in\Bitspace\\\ket{\Stamp}\in \QBitspace^\secpar}}}\Pr[\mathcal{\overset{I}{F}}_{good}\wedge \mathcal{\overset{I}{F}}^{j_1,j_2,b,\ket{\Stamp}}_{view}]\\
&=\frac{1}{\binom{\secpar}{f(\secpar)}}\sum_{\Iset\in \binom{[\secpar]}{f(\secpar)}}\quad\quad\sum_{\mathclap{\substack{j_1,j_2\subseteq \overline{\Iset}\\ \abs{j_1},\abs{j_2}\leq \eta\secpar\\ b\in\Bitspace\\\ket{\Stamp}\in \QBitspace^\secpar}}}\Pr[\mathcal{E}_{good}\wedge \mathcal{E}^{j_1,j_2,b,\ket{\Stamp}}] &\text{By   \cref{lem:view lemma},}\\
&=\frac{1}{\binom{\secpar}{f(\secpar)}}\quad\sum_{\mathclap{\substack{ b\in\Bitspace \\ \ket{\Stamp}\in \QBitspace^\secpar\\\abs{j_1},\abs{j_2}\leq \eta\secpar}}}\quad\quad\quad\sum_{\mathclap{\substack{\Iset\in \binom{[\secpar]}{f(\secpar)} \\ \text{ s.t } \\\quad  \Iset \cap (j_1\cup j_2) =\emptyset }}}\Pr[\mathcal{E}_{good}\wedge \mathcal{E}^{j_1,j_2,b,\ket{\Stamp}}].
\label{eq:reduction analysis}
\end{align}
For any $j_1,j_2\subseteq[\secpar]$ such that $\abs{j_1},\abs{j_2}\leq \eta\secpar$, we denote by $C_{j_1,j_2}$, the number of subsets $ \Iset\in \binom{[\secpar]}{f(\secpar)}$ such that $\Iset \cap (j_1\cup j_2) =\emptyset$. Since $\abs{j_1\cup j_2}\leq 2\eta\secpar$, $C_{j_1,j_2}$ is at least $\binom{\secpar-2\eta\secpar}{f(\secpar)}$. Plugging this bound on $C_{j_1,j_2}$, as well as \cref{eq:Union of good is a win}, to \cref{eq:reduction analysis} achieves the required bound: 
\begin{align}
\Pr&[\UnfExp^{\OneToken,\VR^*}_{\Badv,\CTMAC}(\secpar)=1]\\
&\geq\frac{1}{\binom{\secpar}{f(\secpar)}}\quad\sum_{\mathclap{\substack{ b\in\Bitspace \\ \ket{\Stamp}\in \QBitspace^\secpar\\\abs{j_1},\abs{j_2}\leq \eta\secpar}}} C_{j_1,j_2}\Pr[\mathcal{E}_{good}\wedge \mathcal{E}^{j_1,j_2,b,\ket{\Stamp}}]\\
&\geq\frac{1}{\binom{\secpar}{f(\secpar)}}\quad\sum_{\mathclap{\substack{ b\in\Bitspace  \\ \ket{\Stamp}\in \QBitspace^\secpar\\\abs{j_1},\abs{j_2}\leq \eta\secpar}}}\quad\binom{\secpar-2\eta\secpar}{f(\secpar)}\Pr[\mathcal{E}_{good}\wedge \mathcal{E}^{j_1,j_2,b,\ket{\Stamp}}]\\
&=\frac{\binom{\secpar-2\eta\secpar}{f(\secpar)}}{\binom{\secpar}{f(\secpar)}}\quad\Pr[\bigcupdot_{\mathclap{\substack{ b\in\Bitspace \\ \ket{\Stamp}\in \QBitspace^\secpar\\\abs{j_1},\abs{j_2}\leq \eta\secpar}}} \mathcal{E}_{good}\wedge \mathcal{E}^{j_1,j_2,b,\ket{\Stamp}}]\\
&=\frac{\binom{\secpar-2\eta\secpar}{f(\secpar)}}{\binom{\secpar}{f(\secpar)}}\quad\Pr[\UnfExp^{\OneToken,\NVR}_{\adv,\CTMAC^\eta}(\secpar)=1] &\text{By \cref{eq:Union of good is a win}}\\
&=\frac{\binom{\secpar-2\eta\secpar}{f(\secpar)}}{\binom{\secpar}{f(\secpar)}}W_{\eta}(\secpar)&\text{By definition of $W_{\eta}(\secpar)$}\\
&=\frac{\frac{(\secpar-2\eta \secpar)!}{(\secpar-2\eta \secpar-f(\secpar))!}}{\frac{\secpar!}{(\secpar-f(\secpar))!}}W_{\eta}(\secpar)\\
&=\frac{\prod^{f(\secpar)}_{i=1} \left(\secpar-2\eta \secpar-f(\secpar)+i\right)}{\prod^{f(\secpar)}_{i=1} (\secpar-f(\secpar)+i)}W_{\eta}(\secpar)\\
&=\prod_{i=1}^{f(\secpar)}\left(1-\frac{2\eta \secpar}{\secpar-f(\secpar)+i}\right)W_{\eta}(\secpar)\\
&\geq \left(1-\frac{2\eta \secpar}{\secpar-f(\secpar)}\right)^{f(\secpar)}W_{\eta}(\secpar)
\end{align}
which completes the proof of \cref{lem:noise reduction lemma}.

\end{proof}
\fi
\ifnum\llncs=1
The idea of this lemma is to construct $\Badv$ that creates $\secpar-f(\secpar)$ qubits of his own to complete the  $f(\secpar)$ qubits he received from the challenger to $\secpar$ qubits. The qubits are placed in a randomly chosen set of coordinates $\Iset$, and sent to $\adv$. When receiving a verification query, $\Badv$  rejects if more than $\eta\secpar$ errors occurred in the qubits he injected. 
However, for the $f(\secpar)$ true qubits,  $\Badv$ has to rely on the verification procedure of the challenger, which only accepts if none of the errors occurred in those qubits. Moreover,  if the query was accepted, then $\Badv$ is able to derive the correct sets $\Cons_m,\Miss_{m,\sigma}$ to respond to $\adv$. 
$\adv$ wins the $\UnfExp$ game when the number of errors is at most $\eta \lambda$.   Since $\adv$ is oblivious to the location of the coordinates that were injected, there is a probability that none of $\adv$'s errors occurs in the injected coordinates, and in this case $\Badv$ also wins.  
The full proof of  \cref{lem:noise reduction lemma} is provided in \cref{sec:Noisy to Noiseless Reduction proof}.

\fi
\ifnum\llncs=0
%---------------------------------------------------------------------------------------------------------------------------------------------------------------------------------------------------------------------------------------------------------------------%

\subsection{Choosing an Appropriate Function for the Reduction}\label{sec:noise-tolerant unforgeability}

%---------------------------------------------------------------------------------------------------------------------------------------------------------------------------------------------------------------------------------------------------------------------%
\fi
 We achieve an exponential bound  on $\Pr[\UnfExp^{\OneToken,{\NVR}^*}_{\Badv,\CTMAC^\eta}(\secpar)]$ by instantiating \cref{lem:noise reduction lemma}  with an appropriate $f$ as follows.
\begin{proposition}
\label{prp: reduction bound}
Let $0\leq\epsilon<1$, and $\eta\equiv\epsilon\frac{1-\alpha}{2}$. For every computationally unbounded $\adv$ making $q(\secpar)$ queries to $\VR_{k}$, there exists $\secpar_0$ such that for all $\secpar>\secpar_0$
\[\Pr[\UnfExp^{\OneToken,\NVR^*}_{\adv,\CTMAC^\eta}(\secpar) = 1]\leq \left(\frac{1}{k^z}\right)^{\secpar}k\hat{q}(\secpar),\]
where $k\equiv \frac{(1+\alpha)-\epsilon(1-\alpha)}{2\alpha}$, $z\equiv \frac{1-\epsilon}{1+\epsilon}$, $\hat{q}(\secpar)=2\binom{q(\secpar)+4}{2}$, $\alpha= \cos^2(\frac{\pi}{8})$, and $\UnfExp^{\OneToken,\NVR^*}$ is the game described in \cref{dfn: weak basis break}.

\end{proposition}

\begin{proof}
Let $\adv$ be as described in  \cref{prp: reduction bound}, and
denote 
\[ W_{\eta}(\secpar)\equiv\Pr[\UnfExp^{\OneToken,\NVR^*}_{\adv,\CTMAC^\eta}(\secpar) = 1].\]
Let $f(\secpar)\equiv \floor{z\secpar}$ (notice $0<z\leq 1$, for any $0\leq \epsilon<1$), then by \cref{lem:noise reduction lemma}, there exists an adversary $\Badv$ that makes $q(\secpar)+2$ queries to $\VR$ for which \begin{equation} \label{eq: adversary specific}
\Pr[\UnfExp^{\OneToken,\VR^*}_{\Badv,\CTMAC}(\floor{z\secpar})]\geq\left(1-\frac{2\eta \secpar}{\secpar-\floor{z\secpar}}\right)^{\floor{z\secpar}}W_{\eta}(\secpar).\end{equation}  
Recall that by 
\cref{prp: weak supplied basis security},
for every computationally unbounded $\Badv$ making at most $q(\secpar)+2$ queries to $\VR$ there exists $\secpar_0$ such that for all $\secpar>\secpar_0$,
\[\Pr[\UnfExp^{\OneToken,\VR^*}_{\Badv,\CTMAC}(\floor{z\secpar}) = 1]\leq 2\alpha^{\floor{z\secpar}}\binom{q(\secpar)+4}{2}=\alpha^{\floor{z\secpar}}\hat{q}(\secpar),\] %\weaksb*
Combined with \cref{eq: adversary specific}, we get $\left(1-\frac{2\eta \secpar}{\secpar-\floor{z\secpar}}\right)^{\floor{z\secpar}}W_{\eta}(\secpar)\leq \alpha^{\floor{z\secpar}} \hat{q}(\secpar).$
Moreover $\left(1-\frac{2\eta \secpar}{\secpar-\floor{z\secpar}}\right) \geq 0$, since
$1-\frac{\floor{z\secpar}}{\secpar}\geq 1-\frac{z\secpar}{\secpar}=1-z=\frac{2}{1+
\epsilon}\epsilon\geq \epsilon 
> \epsilon(1-\alpha)=2\eta.$
Hence,
\begin{align}
W_{\eta}(\secpar)
&\leq\left(\frac{\alpha}{\left(1-\frac{2\eta}{1-\frac{\floor{z\secpar}}{\secpar}}\right)}\right)^{\floor{z\secpar}}\hat{q}(\secpar)
&\leq&\left(\frac{\alpha}{\left(1-\frac{2\eta}{1-\frac{\floor{z\secpar}}{\secpar}}\right)}\right)^{z\secpar - 1}\hat{q}(\secpar)\\
&\leq\left(\frac{\alpha}{(1-\frac{2\eta}{1-\frac{z\secpar}{\secpar}})}\right)^{z\secpar - 1}\hat{q}(\secpar)
&\stackrel{\eta=\epsilon\frac{1-\alpha}{2}}{=} &\left(\frac{\alpha}{(1-\frac{2\epsilon\frac{1-\alpha}{2}}{1-z})}\right)^{z\secpar - 1}\hat{q}(\secpar)\\&\stackrel{z=\frac{1-\epsilon}{1+\epsilon}}{=}\left(\frac{\alpha}{1-\left(\frac{2\epsilon\frac{1-\alpha}{2}}{1-\frac{1-\epsilon}{1+\epsilon}}\right)}\right)^{z\secpar - 1}\hat{q}(\secpar)
&=&\left(\frac{\alpha}{1-\frac{(1-\alpha)(1+\epsilon)}{2}}\right)^{z\secpar - 1}\hat{q}(\secpar)\\
&=\left(\frac{1}{\frac{2-(1-\alpha)(1+\epsilon)}{2\alpha}}\right)^{z\secpar - 1}\hat{q}(\secpar)&=&\left(\frac{1}{\frac{(1+\alpha)-\epsilon(1-\alpha)}{2\alpha}}\right)^{z\secpar - 1}\hat{q}(\secpar)\\&=\left(\frac{1}{k}\right)^{z\secpar-1} \hat{q}(\secpar)&=&\left(\frac{1}{k^z}\right)^{\secpar} k\hat{q}(\secpar).
\end{align}
The second inequality holds since $\left(\frac{\alpha}{\left(1-\frac{2\eta}{1-\frac{\floor{z\secpar}}{\secpar}}\right)}\right)\leq \left(\frac{\alpha}{\left(1-\frac{2\eta}{1-\frac{z\secpar}{\secpar}}\right)}\right)=\frac{1}{k}\leq 1$ and $\floor{z\secpar}\geq z\secpar-1$.
\end{proof}

 Note that for every $0\leq \epsilon < 1$, $k>1$ and $z>0$. \cref{prp: reduction bound} thus gives an inverse exponential bound. The bound gets stronger as $\epsilon$ decreases. Indeed,  when $\epsilon =0$, we get the strongest bound $\left(\frac{2\alpha}{1+\alpha}\right)^{\secpar-1}  \hat{q}(\secpar)$. However, as $\epsilon\rightarrow 1$, then $k\rightarrow 1$, $z\rightarrow 0$ and the bound $\left(\frac{1}{k^z}\right)^{\secpar} k\hat{q}(\secpar)\rightarrow \hat{q}(\secpar)$, thus becoming a pointless bound.

The bound above holds for adversaries making \emph{exactly} $q(\lambda)$ queries to $\NVRk$. As in \cref{sec: reduction to weak certified deletion}, this implies that the same bound holds for any $\adv$ making \emph{at most} $q(\secpar)$ queries.
Combining this with  \cref{cor: noise-tolerant weak basis supply security implies noise-tolerant security}, proves \cref{thm:Pi^t_W is secure}. %this then proves \cref{thm:Pi^t_Wis secure}.

\ifnum\llncs=0

\section{Applications}\label{sec:applications}

%  %-----------------------------------------------------------------------------%
\subsection{One-Time-Memories From Stateless Hardware}
\label{sec:One-Time-Memories from Stateless Hardware}
%-----------------------------------------------------------------------------%

In this section, we resolve an open problem presented in \cite{BGZ21}. In \cite{BGZ21}, a construction of a primitive called one-time-memories, similar to non-interactive oblivious transfer, was suggested to be possible when assuming a trusted setup of stateless hardware.\footnote{\cite{BGZ21} uses the term "stateless hardware tokens", as this usage of the term "tokens" conflicts with our own, we refrain from using it.}
%Unbeknownst to us, their construction effectively consisted of using the $\Tmac$ $\CTMAC$ to ensure a single use of the stateless hardware. 
Although their terminology differs, in their protocol, they used what we denote here as $\CTMAC$. However, the authors did not provide a full security proof for their construction. For completeness, we repeat their construction, and show that its security is implied immediately by $\CTMAC^\eta$ being $\UnfDef^{\OneToken,\Vrfy}$, for $0\leq\eta<\frac{1-\alpha}{2}$ (see \cref{thm:Pi_W is secure}).

One-time-memories ($\OTM$) is an ideal modeling of one out of two non-interactive oblivious transfer, as described in \cref{fun: F_OTM}. This ideal functionality allows one machine, called a transmitter, to communicate one, and only one, of two possible messages $s_0,s_1$ to another machine called the receiver, with the transmitter completely oblivious to which of the messages was requested by the receiver. It is emphasized that one-time-memories are not realizable in the standard model, even in a quantum setting \cite{GKR08,BGS13}.

The functionality of stateless hardware is described in \cref{fun: F_wrap}. This ideal functionality allows a transmitter to create an entity computing a stateless program of its choice, ensuring that the receiver can only query the program in a black-box manner. Although the environment and algorithm discussed in the security notion are unbounded, we stress that $\Fwrap$ may only be queried a polynomial number of times, and may only be queried classically. A protocol that uses $\Fwrap$ is said to be in the stateless-hardware-model.

\begin{algorithm*}
    \caption{Ideal Functionality $\FOTM$}
    \label{fun: F_OTM}

    \textbf{Create:} Upon input $(s_0,s_1)$ from the transmitter, with $s_0,s_1\in \Bitspace$, send $\create$ to the receiver and store  $(s_0,s_1)$.
    
   \textbf{Execute:} Upon input $m\in \Bitspace$ from the receiver, send $s_m$ to the receiver. Delete the instance $(s_0,s_1)$.
\end{algorithm*}

\begin{algorithm*}
    \caption{Ideal Functionality $\Fwrap$}
    \label{fun: F_wrap}
    The functionality is parameterized by a polynomial $p(\cdot)$, and an implicit security parameter $\secpar$.

    \textbf{Create:} Upon input $(\create,M) $ from the transmitter, where $M$ is a Turing machine, send $\create$ to the receiver and store  $M$.
    
   \textbf{Execute:} Upon input $(\run,msg)$ from the receiver, execute $M(msg)$ for at most $p(\secpar)$ steps, and let $out$ be the response. Let $out:=\bot$ if $M$ does not halt in $p(\secpar)$ steps. Send $out$ to the receiver.
\end{algorithm*}

The main question explored in \cite{BGZ21} is: can $\OTM$ be implemented in the stateless hardware model? 

They gave a (partial) positive answer:
\begin{theorem}[{\cite[Main Theorem]{BGZ21}} ]
\label{thm:main BGZ21}
There exists a protocol $\mathcal{P}$ in the stateless-hardware-model, which is based on conjugate coding, and implements the $\OTM$ functionality with statistical security in the universal composability framework against a corrupted receiver making at most $c\secpar$ queries to the stateless hardware for any $c<0.114$.
\end{theorem}
Note that in their construction, the adversary can only perform some bounded number of stateless hardware queries. In an earlier version of their manuscript, a stronger result was claimed, where the receiver could make a polynomial number of stateless hardware queries. Unfortunately, their early work \cite{BGZ18} "was withdrawn due to an error in the main security proof", and the superseding work (\cite{BGZ21}) reproduced the same result except for the limitation on the linear number of queries. Additionally, an open question was left unanswered as to whether noise tolerance could be introduced to the construction.

Here, we complete their work by proving the stronger result and solving the open problem by using $\CTMAC^\eta$ (instead of $\CTMAC$).

The protocol $\mathcal{P}^\eta$ described in \cite{BGZ21} is reformulated in \cref{fun:protocol P}, replacing many of the details with the formalism of the scheme $\CTMAC^\eta$. We defer from the original construction by using $\CTMAC^\eta$ instead of the plain noise-sensitive $\CTMAC$.
\begin{algorithm*}
    \caption{ $\mathcal{P^\eta}$: an $\OTM$ protocol}
    \label{fun:protocol P}
    $\Tadv$ Input: $s_0,s_1\in \Bitspace$ 
    
    $\Radv$ Input: $m\in \Bitspace$ 

    \begin{algorithmic}[1]
        \State The transmitter $\Tadv$ calls $k\gets\CTMAC^\eta.\keygen(1^\secpar)$ and $\ket{\stamp}\gets \CTMAC^\eta.\tokengen_{k}$.
        \State $\Tadv$ prepares the program $M(s_0,s_1,k)$, as described in \cref{alg:program for stateless hardware}.
        \State $\Tadv$ sends $(\create,M)$ to the functionality $\Fwrap$, and sends the token $\ket{\stamp}$ to $\Radv$.
        \State $\Radv$ runs $\sigma\gets \CTMAC^\eta.\Sign(\ket{\stamp},m)$ and sends $(\run,\sigma)$ to $\Fwrap$, receiving $out$. 
    \end{algorithmic}

\end{algorithm*}
\begin{algorithm*}
    \caption{$\mathcal{M}$:Program for Stateless Hardware}
   \label{alg:program for stateless hardware}
    Hard-coded values: $s_0,s_1\in \Bitspace,k$.
    
    Inputs: $m\in \Bitspace,\sigma$, where $m$ is the evaluator's choice bit, and $\sigma$ is the proclaimed signature.
    \begin{algorithmic}[1]
        \If {$\CTMAC^\eta.\Vrfy_k(m,\sigma)=1$} {Return $s_m$.}
        \Else{} {Return $\bot$}.
        \EndIf
        
    \end{algorithmic}
\end{algorithm*}
The correctness of the protocol in the noise model $\Noise{2c\eta}$ for any $0\leq c<1$ is trivial.

The security discussed in \cite[Appendix A]{BGZ21} is that of universally composable ($\UC$) security for the sender, as the framework of universally composable security, allows for easy lifting from one-time-memories to that of one-time-programs. We will not repeat the definition for universally composable security here, but refer the reader to \cite{Unr10} for universal composability in a quantum setting.

The main claim used to prove \cref{thm:main BGZ21} is the following:

\begin{theorem}[{\cite[Theorem 3.2]{BGZ21}}  ]
\label{thm: sub linear queries claim}
For $\eta=0$, given a single copy of $\ket{\stamp}$ generated by the transmitter in $\mathcal{P}^\eta$, and the ability to make $r$ (adaptive) queries to the stateless hardware, the probability that an unbounded quantum adversary can force the stateless hardware program $M$ to output both bits $s_0$ and $s_1$ scales as $O(2^{2r-0.228\secpar})$.
\end{theorem}

Notice that the ability to force the token to output both bits corresponds to the adversary submitting $\sigma_0,\sigma_1$ such that the former is a signature for $0$, and the latter is a signature for $1$. Because of this, \cref{thm:Pi^t_W is secure} easily implies an improvement of \cref{thm: sub linear queries claim}:

\begin{theorem}
\label{thm: polynomial queries claim}
For any $0\leq\eta<\frac{1-\alpha}{2}$, where $\alpha=\cos^2(\frac{\pi}{8})$, given a single copy of $\ket{\stamp}$ generated by the transmitter, and the ability to make polynomially many (adaptive) queries to the stateless hardware, the probability that an unbounded quantum adversary can force the stateless hardware program $M$ to output both bits $s_0$ and $s_1$ is negligible.
\end{theorem}

Replacing \cref{thm: sub linear queries claim} with \cref{thm: polynomial queries claim}, results in the following strengthening of \cref{thm:main BGZ21}:

\begin{theorem}
\label{thm:BGZ21 strengthening}
The protocol $\mathcal{P}^{0.07}$ in the stateless hardware model, which is based on conjugate coding, implements the $\OTM$ functionality with statistical security in the universal composability framework against a corrupted receiver making a polynomial number of queries. Moreover, $\mathcal{P}^{0.07}$ is tolerant to up to $14\%$ noise (\cref{def:noise tolerance}).  
\end{theorem}

We repeat the argument in \cite{BGZ21} for completeness and refer the reader to the original work for details regarding the security definitions. 
\begin{proof}[Sketch proof]
It is relatively straightforward to construct a simulator $\Sadv^{\Radv^*}$ in an ideal world (having access to the functionality \FOTM) for any malicious receiver $\Radv^*$: $\Sadv^{\Radv^*}$  will simulate $\Radv^*$, and also simulate by itself the functionality of $\Fwrap$, but with no actual inputs. Instead, the first time a signature is successfully submitted for a bit $b$, $\Sadv^{\Radv^*}$  will call \FOTM to recover $s_b$. $\Sadv^{\Radv^*}$ will keep $s_b$ and answer the same in any subsequent submission of a successful signature for $b$. If no successful signatures for $1-b$ are submitted, this will result in a perfectly indistinguishable view of the transcript. By \cref{thm: polynomial queries claim}, this occurs with an overwhelming probability.
\end{proof}
It is worth noting that even if instantiated with a $\Tmac$ that is secure even against superposition attacks, this would not result in allowing quantum queries to the stateless hardware. In a model where quantum superposition queries to the hardware are allowed, $\OTM$ could not be $\UC$-realized to be secure with statistical security~\cite[Section 4.1]{BGZ21}.

Another similar result is covered in \cite{CGL19}, not only constructing one-time-memories and one-time-programs from stateless hardware, but also providing a general transformation of stateful oracles to stateless oracles via a black box usage of $\UnfDef^{\OneToken,\Vrfy}$ $\Tmac$,\footnote{The terminology used in \cite{CGL19} refers to those as disposable mac, or $\mathsf{DMAC}$.} with various applications, among those the construction of cryptographic disposable back-doors. The results, while very similar, are not directly comparable due to differences in the models discussed. Regardless, instantiating the construction in \cite{CGL19} with $\CTMAC^\eta$ also results in a conjugate coding-based noise-tolerant construction for $\OTM$ in the corresponding model.

\subsection{Quantum Money From TMAC}\label{sec:quantum money}
%-----------------------------------------------------------------------------%
As an application of $\Tmac$, some of the results in~\cite{BS16a} are restated, showing how a $\Tmac$ scheme can be used to construct a classical verification private money scheme. In addition, the theoretical implications of that construction on $\Tmac$ schemes are also reviewed. 
 
\subsubsection{Classically Verifiable Private Quantum Money}

Quantum money was first proposed in~\cite{Wie83} and has since been studied extensively. Here, we only consider the private variant and also restrict ourselves to schemes where verification of the bill is done by an interactive protocol where the bank's side of the protocol is classical. 

%-----------------------------------------------------------------------------%
   
\begin{definition}
A  Classically Verifiable PRivate quantum Money ($\CVPRM$) scheme consists of two $\QPT$ algorithms: $\keygen$ and $\mint$, and an interactive protocol $\Vrfy$ fulfilling the following:
 \begin{enumerate}
     \item On input $1^\secpar$, where $\secpar$ is the security parameter, the algorithm $\keygen$ outputs a classical key $k$.
    \item $\mint_k$ produces a quantum state  $\ket{\$}$, which we refer to as the bill.
      \item $\Vrfy_k(\ket{\$})$ is an interactive protocol with polynomially many rounds between the bank which has access to $k$, and the user who has access to the state $\ket{\$}$. Communication between the sides is only done classically. The bank's algorithm is $\PPT$, and the user's algorithm is $\QPT$.
 \end{enumerate}
We say a classical verification private quantum money scheme is correct if for every $k$ in the range of $\keygen$:
\[\Pr[\Vrfy_k(\mint(k))=1]=1.\]
\end{definition}
\begin{definition}\label{def: unforgeable private money}
A $\CVPRM$ $\Pi$ is $\UnfDef$ if for every $\QPT$ adversary $\adv$

\begin{equation}
    \Pr[{\MoneyForge_{\adv,\Pi}(\secpar)} = 1] \leq \negl,
\end{equation}

if the same holds even for computationally unbounded adversaries (minted only polynomially many bills), the $\CVPRM$ $\Pi$ is said to be $\Unc$ $\UnfDef$ 
\end{definition}

where $\MoneyForge_{\adv,\Pi}(\secpar)$ is the security game described in Game~\ref{exp:The money counterfeiting game}.
\begin{savenotes}
\begin{game}
\caption{Money Forge Game  $\MoneyForge_{\adv,\Pi}(\secpar)$:}\label{exp:The money counterfeiting game}
\begin{algorithmic}[1]
\State $\Cadv$ generates a key $k\gets \keygen(1^\secpar)$. 
\State $\adv$ gets oracle access to $\mint_k$ and $\Vrfy_k$.
  Let $w$ be the number of successful verifications and $\ell$ the number of times that mint was called by the adversary.
\end{algorithmic}

The value of the game $\MoneyForge_{\adv,\Pi}(\secpar)=1$, i.e., the adversary wins if and only if $w > \ell$. 
\end{game}
\end{savenotes}

\subsubsection{Classically Verifiable Quantum Money From TMAC}\label{sec:quantum money from TMAC}
A $\Tmac$ scheme could be used to construct a $\CVPRM$ scheme: let $\TM$ be a $\Tmac$ scheme.

A money scheme $\MM$ is constructed as follows:
The bank runs $\TM.\keygen(1^\secpar)$ to produce a key $k$. In order to mint a bill $\ket{\$}$ to a user, the bank uses $\TM.\tokengen_k$. When the user approaches the bank to verify a bill, the bank will pick a document $m\in \Bitspace^\secpar$ at random and request the user to sign it with $\TM.\Sign_{\ket{\$}}(m)$. In order to verify, the bank will then run $\TM.\Vrfy_k(m,\sigma)$ on the signature supplied by the customer.

\begin{proposition}\label{prp:TMAC->CVPRM}

 If $\TM $ is an $\UnfDef^{\tokengen,\Vrfy}$\footnote{In this context, unforgeability in light of a signing oracle is not required.}  $\Tmac$ scheme, then $\MM$ is an unforgeable $\CVPRM$ scheme. Furthermore, an analogous result holds against unbounded adversaries.
\end{proposition}
\begin{proof}
  
  Correctness is immediate. For unforgeability, assume an adversary was minted $r \in \poly$ bills, so $r$ random documents have been generated. The probability for a collision between two random documents is less than $\frac{r^2}{2^{\secpar}}$, which is negligible if $r$ is polynomial in $\secpar$. Hence, except with a negligible probability, all documents are distinct, so if the adversary is able to generate $w$ bills of which at least $r+1$ bills pass verification, this means that the adversary is able to forge $w$ distinct signed documents, of which at least $r+1$ signatures pass verification, contradicting the security of $\TM$.
 \end{proof}

 While $\sigOrUnrest{0}$ is inspired by the classical verification variant of Wiesner's money~\cite{MVW12}, the $\CVPRM$ resulting by applying the above construction is distinct from it: the bills in the resulting scheme can be thought of as a two-dimensional array of qubits, where each bit of the challenge dictates the measurement of an entire row, rather than of a single qubit.
 \subsubsection{Trade-Off Theorem for Quantum Money and TMAC}
 In \cref{sec:Security of Conjugate TMAC}, we proved that $\CTMAC$ (\cref{alg:Conjugate Tmac} ) is  $\Unc$ $\UnfDef^{\OneToken,\Vrfy}$. In order to construct an $\UnfDef^{\tokengen,\Vrfy}$, computational assumptions were used, as described in \cref{sec:main_results}.
A natural question is then whether a $\Tmac$ can be constructed to be  $\Unc$ $\UnfDef^{\tokengen,\Vrfy}$. Based on a result by Aaronson \cref{thm: tradeoff for quantum money}, it is shown that this is impossible. In fact, even an  $\Unc$ $\UnfDef^{\tokengen}$, $\Tmac$ is impossible to achieve.

\begin{theorem}
[{Trade-off Theorem for Quantum Money,~\cite[Theorem 8]{Aar20}} ]\label{thm: tradeoff for quantum money}
Given any private-key quantum money
scheme,\footnote{And in particular a $\CVPRM$.} with $d$-qubit bills and an $m$-bit secret key held by the bank, a counterfeiter can produce additional bills which pass verification with $1-o(1)$ probability, given $\widetilde{O}(dm^4)$ legitimate bills and 
$exp(d,m)$ computation time. No queries to the bank are needed to produce these bills.
\end{theorem}

The construction of $\CVPRM$ given in \cref{sec:quantum money from TMAC} implies that the same holds for $\Tmac$; otherwise, by \cref{prp:TMAC->CVPRM}, it could be used to construct unconditionally secure quantum money.
This proves \cref{thm: impossibility of unconditionally secure TMAC scheme.}.

\begin{theorem}[Trade-off Theorem for $\Tmac$]
\label{thm: impossibility of unconditionally secure TMAC scheme.}
Given any $\Tmac$
scheme, with $d$-qubit tokens and an $m$-bit secret key, a counterfeiter can produce additional tokens which would produce signatures that pass verification with $1-o(1)$ probability, given $\widetilde{O}(dm^4)$ legitimate tokens and 
$exp(d,m)$ computation time. No queries to the verification oracle are needed to produce these tokens.
\end{theorem}

\subsubsection{Temporarily Memory-Dependent  Non-Interactive Quantum Money}\label{sec: Temporarily Memory-Dependent  Non-Interactive Quantum Money}

While the construction above (\cref{sec:quantum money from TMAC}) is interesting from a theoretical perspective, it does not currently offer anything not already guaranteed by other existing private money schemes based on simple states~\cite{MVW12,PYJ+12,Gav12}. On the other hand, it clearly does not make full use of the power of the $\Tmac$: A $\Tmac$ guarantees unforgeability as long as the documents are not repeated, regardless of them being chosen at random.

Using the above property, \cite{BS16a} devised an alternative non-interactive verification procedure at the expense of losing the statelessness of the scheme: If Alice holds a quantum bill,
one thing she can do is spend it the usual way. However, an alternative thing she can do with the
bill is to use it to sign a document. Such a signature will necessarily consume the bill and, thus, it can
be used as proof that Alice has burned her bill, essentially creating a classical check.

To prevent double spending, the bank would only accept checks that sign a document consisting of the current time and date. The bank would have to keep a database of the cashed checks to prevent double-spending; however, this database could be time-limited. We could keep a short "time frame" during which the bank would need to keep a database of checks received, but it could clear that database at the end of a time frame. The bank can choose simply not to accept any checks whose date has "expired", i.e., the time stamp on the check does not fall within the current time frame.

If there are multiple bank branches, Alice could try to cash the check twice at different bank branches. In order to prevent this, the bank can require the signed document to list the bank branch number as well.

\fi
%  %-----------------------------------------------------------------------------%
\section{Discussion and Open Questions}
% %-----------------------------------------------------------------------------%
\label{sec:open questions}

In this work, we proved that post-quantum one-way functions imply $\Tmac$ schemes. 
It is natural to ask if the converse holds in the quantum setting.  
\begin{openproblem}\label{op:tmac-implies-one-way-functions}
Does $\Tmac$ imply post-quantum one-way functions? 
\end{openproblem}

Note that in the classical setting, we know that $\mac$s exist if and only if one-way functions exist. 
A standard construction of one-way functions from $\mac$~\cite[page 592, Exercise 7]{Gol04} is based on a clever engineering of the classical randomness involved in the signing algorithm. It is not clear whether this approach can be applied in the quantum setting. 

We presented a simple construction for an $\UnfDef^{\tokengen,\Vrfy,\widetilde{\Sign}}$  $\Tmac$ based on simple tensor product states. However, the simplicity of the construction came at the expense of some of the stronger properties that the construction in~\cite[Appendix B.1, Section 6]{BS16a} achieves. One such property is security against adversaries that have quantum access to the verification oracle.
An interesting question would be to ask whether a scheme based on BB84 states, or similar "simple" states, could satisfy this stronger notion.

\begin{openproblem}\label{op:Strong Mac from simple states}
Is there a $\Tmac$ construction based on simple states that satisfies unforgeability, even against adversaries with quantum access to the verification oracle? 
\end{openproblem}
Note that an $\UnfDef^{\tokengen,\ket{\Vrfy},\widetilde{\Sign}}$ (see \cref{sec:A Quantum Superposition Attack}) $ \Tmac$, i.e., a $\Tmac$ secure against adversaries with quantum access to the verification oracle, could potentially result in a secure tokenized signature scheme if augmented with some form of obfuscation~\cite{BGIRSVY12}. This would immediately give rise to a public quantum money construction.  A public scheme based on product states seems to be a difficult task even with a quantum-queries secure $\Tmac$ scheme and obfuscation. A result on quantum state restoration~\cite{FGH+10} proves this is impossible to do with a rank-$1$ projection as the verification procedure. However, it is not known if a similar result applies for higher dimension projections.

Another interesting feature that can be added to a $\Tmac$ is to have classical token-generation, meaning, Alice instead of creating a quantum signing token and giving it to Bob would run an interactive protocol with Bob using classical communication, such that an honest Bob would have a quantum signing token at the end of the protocol.
\begin{openproblem}\label{op:Classical minting TMAC}
Is there a $\Tmac$ construction that allows classical token-generation? 
\end{openproblem}
A natural candidate for such a scheme is a simple variant on semi-quantum money~\cite{RS20}, which is a money scheme with both classical verification and classical minting. In analogy to how $\CTMAC$ signs the single bit message $b$ by "responding" to the challenge $b^\secpar$ in a  Wiesner's money scheme, we can devise a $\Tmac$ scheme that signs the single bit message $b$ by "responding" to the challenge $b^\secpar$ in semi-quantum money. Security against an adversary with no access to verification follows from the security of the money scheme, but we have not been unable to either prove or disprove the security against an adversary that has access to verification.

%%%%%%%%%%%%%%%%%%%%%%%%%%%%
\ifnum\masterthesis=0
\ifnum\anon=0

    \subsection*{Acknowledgments}
We wish to thank Anne Broadbent for discussions related to \cref{sec:One-Time-Memories from Stateless Hardware}.
We also wish to thank Amos Beimel, Shalev Ben-David and Roy Radian for their comments. This work was supported by the Israel Science Foundation (ISF) grant No. 682/18 and 2137/19 and
by the Cyber Security Research Center at Ben-Gurion University.

%%%%%%%%%%%%%%%%%%%%%%%%%%%%
%NON-ANON PART
%ANON PART
The icon $\Stamp$ was downloaded from \url{http://icons8.com}, and is licensed under Creative Commons Attribution-NoDerivs 3.0 Unported.
\fi
\fi

\ifnum\sigconf=1
    \bibliographystyle{ACM-Reference-Format}
\else
    \ifnum\cryptology=1
        \bibliographystyle{abbrv}
    \else
        \ifnum \llncs=1
            \bibliographystyle{splncs04}
        \else
            \bibliographystyle{alphaabbrurldoieprint}
        \fi
    \fi
\fi

\ifnum\masterthesis=0
    \bibliography{main}
\fi
\appendix

%-------------------------------------------------------------------------------------------------------------------%
\ifnum\llncs=1

\chapter*{Supplementary Material}
\fi

\ifnum\llncs=1

\section{Ancillary Proofs}\label{sec:extra proofs}

\subsection{Noisy to Noiseless Reduction Proof}
\label{sec:Noisy to Noiseless Reduction proof}
\begin{proof}[Proof of \cref{lem:noise reduction lemma}]

\end{proof}
\ifnum\llncs=0
\section{Detailed Analysis of Certified Deletion Scheme}\label{sec: analysis of weak certified deletion scheme}
\else
\subsection{Detailed Analysis of Certified Deletion Scheme}\label{sec: analysis of weak certified deletion scheme}
\fi

For convenience, the definitions for the schemes $E_m$ and the security game for weak certified deletion are brought here again.

\begin{algorithm*}
    \caption{The $\QECD$ $E_m$ - A $\QECD$ scheme (for $m\in \Bitspace$)}
    \label{alg: E_m (rep)}
    \begin{algorithmic}[1] % The number tells where the line numbering should start
        \Procedure{$\keygen$}{$1^\secpar$}
                 \State $k\sample\Bitspace^\secpar$
            \State \textbf{Return} $k$.
        \EndProcedure
    \end{algorithmic}
    \begin{algorithmic}[1] % The number tells where the line numbering should start
        \Procedure{$\Enc_k$}{$a$}
            \State \textbf{Return} $H^k(\ket{a})$.
        \EndProcedure

    \end{algorithmic}
	\begin{algorithmic}[1] % The number tells where the line numbering should start
        \Procedure{$\Dec_k$}{$\ket{c}$}
            \State Compute $H^k\ket{c}$ to obtain $\ket{cc}$
            \State Measure $\ket{cc}$ to obtain $a'$
            \State \textbf{Return} $a'$.
                \EndProcedure

    \end{algorithmic}
  	\begin{algorithmic}[1] % The number tells where the line numbering should start
        \Procedure{$\Del$}{$\ket{c}$}
            \State Compute $(H^m)^{\otimes\secpar}\ket{c}$ to obtain $\ket{cc}$
            \State Measure $\ket{cc}$ to obtain $cer$
            \State \textbf{Return} $cer$.
            \EndProcedure
    \end{algorithmic}
      	\begin{algorithmic}[1] % The number tells where the line numbering should start
        \Procedure{$\Vrfy_k$}{$cer$}
            \State Define $\Cons_m=\{i\in [\secpar]| k_i=m\}$.
            \If{$cer=a|_{\Cons_m}$}
            \State \textbf{Return} $1$.
            \Else
                \State \textbf{Return} $0$.
            \EndIf
            \EndProcedure
    \end{algorithmic}
\end{algorithm*}

\begin{savenotes}
\begin{game}
\caption{Weak Certified Deletion Game $\exWCerDel_{\adv,\Pi}(\secpar)$:}\label{exp: weak certified deletion (rep)}  
\begin{algorithmic}[1]
    \State The challenger $\Cadv$ runs $\keygen(1^\secpar)$ to generate $k$, and uniformly samples $a\in\Bitspace^\secpar$. $\Cadv$ then sends $\ket{c}=\Enc_k(a)$.
   \State $\adv$ sends $\Cadv$ some string $cer$.
    \State $\Cadv$ computes  $V\gets\Vrfy_k(cer)$, and  then sends $k$ to the adversary.
    \State $\adv$ outputs $a'$.
\end{algorithmic}
 We say the output of the game is 1 if and only if $V=1$ and $a=a'$.
\end{game}
\end{savenotes}

Presented ahead is one possible strategy for attacking the scheme. The main idea is that if the adversary could know both a valid result of a measurement of the token by the standard basis and a measurement in the Hadamard basis, then it could win the security game by submitting first the result of a measurement in the Hadamard basis as a certificate, and upon receiving the key, output the measurement corresponding to it as $a'$. Therefore, an adversary may act in the following manner: when receiving $\ket{\Stamp}$, it will create two strings
corresponding to guesses for measurement by both bases. After submitting the guessed result for a measurement by the standard basis, it will receive the key and answer accordingly in a deterministic manner. The question is then limited to computing the optimal probability for such an adversary to be correct in its guesses at the first stage of the protocol. This question is exactly the one discussed in~\cite[Section 4.2]{MVW12},  in the form of two distinct challenges sent by the bank to an adversary having a single quantum money bill. The optimal probability achieved by~\cite{MVW12} 
 is $\cos^2(\frac{\pi}{8})\approx0.85355$.

It can be concluded that there is an adversary winning $\exWCerDel_{\adv,E_1}(1)$ with a probability of at least $\cos^2(\frac{\pi}{8})$. A priori, it is not clear why this is the best the adversary can do. It seems reasonable that an adversary might use a more complex strategy to improve their chances instead of guessing ahead at the first stage of the protocol. We show that, in fact, this is an optimal strategy using semi-definite programming techniques.

%-----------------------------------------------------------------------------%
\ifnum\llncs=0
\subsection{Semi-Definite Programming}
\else
\subsubsection{Semi-Definite Programming}
\fi
%-----------------------------------------------------------------------------%

This section discusses standard notions and results of semi-definite programming. Our analysis in this section is similar to the one in Ref.~\cite{MVW12}, and we follow the conventions therein. This preliminary section is taken verbatim from~\cite{MVW12}.

Semi-definite programming is a topic that has found several interesting applications within quantum computing and quantum information theory in recent years. Here, we provide just a brief summary of semi-definite programming that is focused on the narrow aspects that we use. More comprehensive discussions can, for instance, be found in~\cite{VB96,Lov03,Dk02,BV04}. We first cover some notations:

For any finite-dimensional complex Hilbert space $\mathcal{X}$, we write
$L(\mathcal{X})$ to denote the set of linear operators acting on $\mathcal{X}$, 
 $\Herm(\mathcal{X})$ to denote the set of Hermitian operators acting
on $\mathcal{X}$, $\Pos(\mathcal{X})$ to denote the set of positive
semi-definite operators acting on $\mathcal{X}$, and $\Pd(\mathcal{X})$ to denote
the set of positive-definite operators acting on $\mathcal{X}$. For every $A,B\in \Herm(\mathcal{X})$, the notation $A\succeq B$ indicates that $A - B$  is positive semi-definite. $ D(\mathcal{X})$  denote the set of density operators acting on $\mathcal{X}$.

Given operators $A,B\in L(\mathcal{X})$, one defines the inner product
between $A$ and $B$ as $\langle A,B\rangle = \Tr(A^{\ast}B)$.
For Hermitian operators $A,B\in \Herm(\mathcal{X})$, it holds that
$\langle A,B\rangle$ is a real number and satisfies $\langle A,B\rangle = \langle B,A\rangle$.
For every choice of finite-dimensional complex Hilbert spaces $\mathcal{X}$ and
$\mathcal{Y}$, and for a given linear mapping of the form
$\Phi: L(\mathcal{X})\rightarrow L(\mathcal{Y})$, there is a unique mapping
$\Phi^{\ast}:L(\mathcal{Y})\rightarrow L(\mathcal{X})$ (known as the \emph{adjoint}
of $\Phi$) that satisfies
$\langle Y,\Phi(X)\rangle = \langle \Phi^{\ast}(Y),X\rangle$ for all $X\in L(\mathcal{X})$ and
$Y\in L(\mathcal{Y})$.

A semi-definite program is a triple $(\Phi,A,B)$, where
\begin{itemize}
\item
$\Phi: L(\mathcal{X}) \rightarrow L(\mathcal{Y})$ is a Hermiticity-preserving
  linear mapping, and
\item $A\in \Herm(\mathcal{X})$ and $B\in\Herm(\mathcal{Y})$ are Hermitian operators,
\end{itemize}
for some choice of finite-dimensional complex Hilbert spaces $\mathcal{X}$ and $\mathcal{Y}$.
We associate with the triple $(\Phi,A,B)$ two optimization problems,
called the \emph{primal} and \emph{dual} problems, as follows:
\begin{center}
\begin{minipage}{0.5\textwidth}
  \centerline{\underline{Primal problem}}\vspace{-6mm}
    \begin{align}
      \text{maximize:}\;\; & \langle A,X\rangle\\
      \text{subject to:}\;\; & \Phi(X) = B,\\
      & X\in \Pos(\mathcal{X})
    \end{align}
\end{minipage}\hspace*{5mm}
\begin{minipage}{0.4\textwidth}
  \centerline{\underline{Dual problem}}\vspace{-6mm}
    \begin{align}
      \text{minimize:}\;\; & \langle B,Y\rangle \\
      \text{subject to:}\;\; & \Phi^{\ast}(Y) \succeq A,\\
      & Y\in \Herm(\mathcal{Y}).
      \label{eq:recipe_dual}
    \end{align}
  \end{minipage}
\end{center}

The optimal primal value of this semi-definite program is
\[
\alpha = \sup\{\langle A,X\rangle \,:\,X\in \Pos(\mathcal{X}),\,\Phi(X) = B\},
\]
and the optimal dual value is
\[
\beta = \inf\{\langle B,Y\rangle \,:\,Y \in \Herm(\mathcal{Y}),\,\Phi^{\ast}(Y) \succeq A\}.
\]
(It is to be understood that the supremum over an empty set is
$-\infty$ and the infimum over an empty set is $\infty$, so $\alpha$
and $\beta$ are well-defined values in 
$\mathbb{R}\cup\{-\infty,\infty\}$.
In this paper, however, we will only consider semi-definite programs for
which $\alpha$ and $\beta$ are finite.)

It always holds that $\alpha \leq \beta$, which is a fact known as
\emph{weak duality}.
The condition $\alpha = \beta$, which is known as 
\emph{strong duality}, does not hold for every semi-definite program,
but there are simple conditions known under which it does hold.
The following theorem provides one such condition (that has both a
primal and dual form).
\begin{restatable}[Slater's theorem for semi-definite programs]{theorem}{slater}
  \label{theorem:Slater}
Let $(\Phi,A,B)$ be a semi-definite program and let $\alpha$ and
$\beta$ be its optimal primal and dual values.
\begin{itemize}
\item
  If $\beta$ is finite and there exists a positive definite operator
  $X\in \Pd(\mathcal{X})$ for which $\Phi(X) = B$,
  then $\alpha = \beta$ and there exists an operator $Y\in \Herm(\mathcal{Y})$
  such that $\Phi^{\ast}(Y)\succeq A$ and $\langle B,Y\rangle  = \beta$.
\item
  If $\alpha$ is finite and there exists a Hermitian operator
  $Y\in \Herm(\mathcal{Y})$ for which $\Phi^{\ast}(Y) > A$,
  then $\alpha = \beta$ and there exists a positive semi-definite 
  operator $X\in \Pos(\mathcal{X})$ such that $\Phi(X)=B$ and 
  $\langle A,X\rangle  = \alpha$.
\end{itemize}
\end{restatable}
In other words, the first item of this theorem states that if the dual
problem is feasible and the primal problem is 
\emph{strictly feasible}, then strong duality holds and the optimal
dual solution is achievable.
The second item is similar, with the roles of the primal and dual
problems reversed.

%-----------------------------------------------------------------------------%
\ifnum\llncs=0
\subsection{SDP Formulation of a 1-Fold Weak Certified Deletion Scheme}\label{sec:SDP formulation of one fold Wiesner's encryption weak certified deletion forger}
\else
\subsubsection{SDP Formulation of a 1-Fold Weak Certified Deletion Scheme}\label{sec:SDP formulation of one fold Wiesner's encryption weak certified deletion forger}
\fi
%-----------------------------------------------------------------------------%

We begin by addressing the specific security game $\exWCerDel_{\adv,E_1}(1)$. That is, $m=1$, and since $\lambda=1$, the message, the key, and the certificate are all only a single qubit.

Let $\mathcal{X}_{a},\mathcal{X}_{k}$ represent the single-qubit registers designating the  message $a$ and the key $k$. Also, let $\mathcal{X}_{m_1},\mathcal{X}_{m_2}$ be registers used to transport messages, back and forth. As during the protocol, the key is sent to the adversary, we initialize $\mathcal{X}_{m_2}$ as a copy of the key, and $\mathcal{X}_{m_1}$ as the encrypted state itself. Finally, we pad with $\mathcal{X}_{p}$ to get a purified state. For ease of notation, we will also denote $\overline{\mathcal{X}}_{\ast}$ when referring to all register but the one appearing in the subscript.

Denote \[\mathcal{X}= \mathcal{X}_{p}\otimes\mathcal{X}_{m_1} \otimes \mathcal{X}_{m_2}\otimes  \mathcal{X}_{k}\otimes \mathcal{X}_{a},  \]
\[\ket{\psi}=\frac{1}{2}(\ket{00000}+\ket{11001}+\ket{0+110}+\ket{1-111})\in \mathcal{X}, \]
  \[\sigma_0=\ketbra{\psi}.\]
  For any adversary, the security game can be described as follows: at the first stage, the challenger in the hold of $\sigma_0$, sends $\mathcal{X}_{m_1}$ to the adversary. The adversary then applies some quantum channel $Q_1$ on its part of the system and sends back the register $\mathcal{X}_{m_1}$ with its guess of the certificate. The challenger responds by sending over $\mathcal{X}_{m_2}$, which holds a copy of the key. While the challenger should have measured the certificate first, by the deferred measurement principle, it is clear that the outcome of the game is the same as if it were not to be touched, and measured at the end of the protocol. The adversary then applies a second quantum channel $Q_2$, and sends back $\mathcal{X}_{m_2}$ with a guess for $a'$. In order to verify, the challenger checks its registers and verifies that the results are proper. This verification could be instantiated by the following positive operator valued measurement:
\begin{align}
    \Pi =& I_{\mathcal{X}_p}\tensor(\ketbra{0000}+\ketbra{1000}+\ketbra{0101}\\+&\ketbra{1101}+\ketbra{0010}+\ketbra{1111}).
\end{align}

We denote $\rho_1,\rho_2$ respectively as the state of the system after the adversary sends the guess for the certificate and the state of the system when the adversary sends back a guess for $a'$. Since $\rho_1,\rho_2$ must be quantum systems, their trace must be $1$, and since $Q_1,Q_2$ cannot affect the subsystems in possession of the challenger, it is clear that the following constraints are also a necessary condition:
\[\begin{aligned}
\Tr_{\mathcal{X}_{m_1}}(\rho_1) =\Tr_{\mathcal{X}_{m_1} } (\sigma_0), \\ \Tr_{ \mathcal{X}_{m_2}}(\rho_2) = \Tr_{ \mathcal{X}_{m_2}}(\rho_1).
\end{aligned}\]
While not initially clear, these $4$ constraints suffice to  describe the game fully. That is, for any $\rho_1$ and $\rho_2$ fulfilling this, there exist quantum channels $Q_1,Q_2$ operating only on the adversary's subsystem such that $Q_1(\sigma_0)=\rho_1$ and $Q_2(\rho_1)=\rho_2$. To see this, consider the following theorem.

\begin{theorem}[{Freedom in Purifications,~\cite[Exercise 2.81]{NC11}}] 
\label{thm:freedom in purification} Let $\ket{\psi}$ and $\ket{\phi}$ be two
purifications of a state $\rho$ to a composite system $\mathcal{A\otimes B}$. There exists a
unitary transformation $U$ acting on system $\mathcal{B}$ such that:
\[\ket{\phi}= (I_A \otimes U) \ket{\psi}\] 
\end{theorem}
Let $c_1\in D(\mathcal{X}\otimes \mathcal{C})$ be a purification of $\rho_1$. As both $\sigma_0\otimes \ketbra{0}_\mathcal{C}$ and $c_1 $, are purifications of $\Tr_{\mathcal{X}_{m_1}}(\sigma_0)$, by \cref{thm:freedom in purification} there is a unitary $U$ operating only on $\mathcal{X}_{m_1}\otimes \mathcal{C}$ such that $I_{\overline{\mathcal{X}}_{m_1}}\otimes U (\sigma_0\otimes \ketbra{0}_\mathcal{C})=c_1$. Dropping the extra registers, this gives a desired $Q_1(\sigma_0)=\rho_1$. The same argument applies for the existence of a channel $Q'_2$ only acting on the adversary's subsystem such that $Q'_2(c_1)=\rho_2$, perhaps demanding an even larger space. Throwing away the extra subsystem again provides $Q_2$.
We can now maximize over the probability of success:

 \vspace{2mm}
\begin{center}
\begin{minipage}{0.5\textwidth}
  \centerline{\underline{Primal problem}}\vspace{-6mm}
  \begin{align}\label{sdp:1-fold} 
    \text{maximize:}\;\; & \langle \Pi,\rho_2\rangle\\
    \text{subject to:}\;\; & \Tr_{\mathcal{X}_{m_1}}(\rho_1) = \Tr_{\mathcal{X}_{m_1} } (\sigma_0) \\ & \Tr_{ \mathcal{X}_{m_2}}(\rho_2) = \Tr_{ \mathcal{X}_{m_2}} (\rho_1) \\ &
    \Tr(\rho_1)=\Tr(\rho_2)=1\\ &
    \rho_1,\rho_2\in \Pos(\mathcal{X}) \\ &
\end{align}
\end{minipage}\hspace*{5mm}
\end{center}
In order to present this $\SDP$ in the standard form, we combine the two variables into a single variable by adding an auxiliary space $\mathcal{Z}\equiv\mathbb{C}^2$. We can then identify $\rho_1$ with  $\Tr_{\mathcal{Z}}\big((\ketbra{0}\otimes I_\mathcal{X})\rho(\ketbra{0}\otimes I_\mathcal{X} )\big)$ and $\rho_2$ with  $\Tr_{\mathcal{Z}}\big((\ketbra{1}\otimes I_\mathcal{X})\rho(\ketbra{1}\otimes I_\mathcal{X})\big)$.

\begin{center}
\begin{minipage}{0.5\textwidth}
  \centerline{\underline{Primal problem}}\vspace{-6mm}
  \begin{align}\label{sdp:rewrite1}
    \text{maximize:}\;\; & \langle \ketbra{1}\otimes\Pi,\rho\rangle\\
    \text{mark:}\;\; &
    \Gamma_0\equiv\ketbra{0}\otimes I_\mathcal{X}\\ &
    \Gamma_1\equiv\ketbra{1}\otimes I_\mathcal{X}\\ &
    \rho_1\equiv\Tr_{\mathcal{Z}}(\Gamma_0\rho\Gamma_0)\\ &
        \rho_2\equiv\Tr_{\mathcal{Z}}(\Gamma_1\rho\Gamma_1)\\ 
       \text{subject to:}\;\; & \Tr_{ \mathcal{X}_{m_1}}(\rho_1) = \Tr_{\mathcal{X}_{m_1} } (\sigma_0) \\ & \Tr_{ \mathcal{X}_{m_2}}(\rho_1) -\Tr_{ \mathcal{X}_{m_2}}(\rho_2)=0 \\ &
    \Tr(\rho_2)=1 \\ &
    \Tr(\rho_1)=1 \\ &
    \rho\in  \Pos(\mathcal{Z}\otimes \mathcal{X}) 
    \label{eq:primal_problem_ugly}
\end{align}

  \end{minipage}
\end{center}
We further simplify notations in  order to succinctly present the dual: we rewrite our constraints as a linear mapping
\[\Phi:L(\mathcal{Z\otimes X})\rightarrow L(\overline{\mathcal{X}}_{m_1}\oplus \overline{\mathcal{X}}_{m_2}\oplus \mathbb{C}\oplus \mathbb{C}),\] 
defined by
\[\phi(\rho)=\phi_1(\rho)\oplus\phi_2(\rho)\oplus\phi_3(\rho)\oplus\phi_4(\rho),\]

for
\[\phi_1(\rho)=\Tr_{\mathcal{X}_{m_1}, \mathcal{Z}}(\Gamma_0\rho\Gamma_0)\]
\[\phi_2(\rho)=\Tr_{\mathcal{X}_{m_2},\mathcal{Z}}(\Gamma_1\rho\Gamma_1-\Gamma_0\rho\Gamma_0)\]
\[\phi_3(\rho)=\Tr(\Tr_{Z}(\Gamma_0\rho\Gamma_0)).\]
\[\phi_4(\rho)=\Tr(\Tr_{Z}(\Gamma_1\rho\Gamma_1)).\]

We also denote $Q\equiv \ketbra{0}\otimes \Pi$ and notice $Q$ is positive semi-definite. 
For convenience, we also rename the left-hand side of the constraints: $R_1\equiv \sigma_0, R_2\equiv\overline{0}\in \overline{\mathcal{X}}_{m_2},R_3\equiv 1\in \mathbb{C}, R_4\equiv 1\in \mathbb{C}$, and $R\equiv R_1\oplus R_2\oplus R_3\oplus R_4$. Lastly, we define the space  $\mathcal{S}\equiv\overline{\mathcal{X}}_{m_1}\oplus \overline{\mathcal{X}}_{m_2}\oplus \mathbb{C}\oplus \mathbb{C}$.

\begin{center}
\begin{minipage}{0.5\textwidth}
  \centerline{\underline{Primal problem}}\vspace{-6mm}
  \begin{align}\label{sdp:rewrite2}
    \text{maximize:}\;\; & \langle  Q,\rho\rangle\\
    \text{subject to:}\;\; &  
    \Phi(\rho)=R,\\ &
    \rho\in \Pos({\mathcal{Z}\otimes \mathcal{X}}) \\ &
\end{align}
\end{minipage}\hspace*{5mm}
\begin{minipage}{0.4\textwidth}
  \centerline{\underline{Dual problem}}\vspace{-6mm}
    \begin{align}\label{sdp:rewrite2_dual}
      \text{minimize:}\;\; & \langle R, Y\rangle \\
      \text{subject to:}\;\; & \Phi^{\ast}(Y) \succeq Q,\\
      & Y\in \Herm(\mathcal{S})
    \end{align}
  \end{minipage}
\end{center}

\begin{remark}
\label{remark:sdp_value}
The value of SDP~\ref{sdp:rewrite2} for the primal problem is $\alpha$, where $\alpha= \cos^2(\frac{\pi}{8})$ is obtained numerically.\footnote{We used the CVX~\cite{GB08} package on Matlab, and the \href{https://www.dr-qubit.org/matlab.html}{Quantinf} package by Toby Cubbit, the numerical error was relatively large and differed between engines (of magnitude $10^{-5}$ on the "sedumi" solver and $10^{-3}$ on the "SDPT3" solver), leaving some room for doubt if $\alpha$ is exactly $\cos^2(\frac{\pi}{8})$, or a slightly larger value. Regardless, the exact value of $\alpha$ does not affect the results of this work, besides the precise amount of noise tolerance. \ifnum\anon=0 The Matlab source code is available at \url{https://arxiv.org/src/2105.05016/anc}.\fi } As the resulting solution is positive definite, by Slater's condition (\cref{theorem:Slater}), the optimal solution to the dual is also $\alpha$.
\end{remark}
%-----------------------------------------------------------------------------%
\ifnum\llncs=0
\subsection{Parallel Repetition for Weak Certified Deletion Scheme}
\label{sec:parallel repetition}
\else
\subsubsection{Parallel Repetition for Weak Certified Deletion Scheme}
\label{sec:parallel repetition}
\fi
%-----------------------------------------------------------------------------%
Let us now consider the security game $\exWCerDel_{\adv,E_1}(\secpar)$. This game can be thought of as the game of playing $\secpar$ copies of  $\exWCerDel_{\adv,E_1}(1)$  games simultaneously, where $\adv$ wins if and only if it wins in all copies. In particular, the preparation and verification of a  $\secpar$-qubit encrypted message and a certificate is, from the perspective of the challenger, equivalent to the independent preparation
and verification of $\secpar$ independent pairs of single-qubit encrypted messages and certificates. Hence, a successful decryption attack is equivalent to a successful decryption attack against all $\secpar$ of the single-qubit games.

It can be concluded from the semi-definite programming formulation above that an adversary gains no advantage whatsoever by correlating multiple qubits during an attack. The interested reader is encouraged to read~\cite{MS07} for a more comprehensive discussion on the conditions that this holds. In our game, the key factor is that the objective function $Q$ is positive semi-definite. We formulate a $\SDP$ for the $\exWCerDel_{\adv,E_1}(\secpar)$ game, which is the $\secpar$-fold repetitions of $\exWCerDel_{\adv,E_1}(1)$.

\begin{center}
\begin{minipage}{0.5\textwidth}
  \centerline{\underline{Primal problem}}\vspace{-6mm}
  \begin{align}\label{sdp:parallel} 
    \text{maximize:}\;\; & \langle Q^{\otimes\secpar},\rho\rangle\\
    \text{subject to:}\;\; &  
    \Phi^{\otimes\secpar}(\rho)=R^{\otimes\secpar},\\ &
    \rho\in \Pos((\mathcal{Z\otimes X})^{\otimes\secpar}) \\ &
\end{align}
\end{minipage}\hspace*{5mm}
\begin{minipage}{0.4\textwidth}
  \centerline{\underline{Dual problem}}\vspace{-6mm}
    \begin{align}
      \text{minimize:}\;\; & \langle R^{\otimes\secpar}, Y\rangle \\
      \text{subject to:}\;\; & (\Phi^{\otimes\secpar})^\ast(Y) \succeq Q^{\otimes\secpar},\\
      & Y\in \Herm(\mathcal{S}^{\otimes\secpar})
    \end{align}
  \end{minipage}
\end{center}

\begin{proposition}

The value of \cref{sdp:parallel} is $\alpha^\secpar$, where $\alpha= \cos^2(\frac{\pi}{8})$.
\label{prp:parallel-sdp-value}
\end{proposition}
\begin{proof}
By  \cref{remark:sdp_value}, the optimal value for the primal problem (\cref{sdp:rewrite2}) and its dual problem (\cref{sdp:rewrite2_dual}) is $\alpha$. 
It is easy to argue that the solution $P =  \bigotimes^{\secpar}_{i=1} P'$ for $P'$ being an optimal primal solution for the single-repetition semi-definite program is a solution for the primal parallel problem with value $\alpha^\secpar$. 
Similarly, the value of $D=\bigotimes^{\secpar}_{i=1} D'$ for $D'$ being optimal dual solutions for the single-repetition semi-definite program is $\alpha^\secpar$. 
In order to see that $D$ is indeed a feasible solution, we use the fact that for all positive semi-definite $A$ and $B$, $A \succeq B \succeq 0$ implies that  $ A^{\otimes \secpar}\succeq B^{\otimes \secpar}$, since $\Phi^*(D')\succeq Q\succeq 0$ we have that $(\Phi^{\otimes\secpar})^*)(D')=(\Phi^*(D'))^{\otimes\secpar}\succeq Q^{\otimes\secpar}\succeq 0$. This implies that the success probability of an optimal strategy adversary for $\exWCerDel_{\adv,E_1}(\secpar)$, is $\alpha^\secpar$.
\end{proof}
By symmetry, the same holds for the $\SDP$ corresponding to $\exWCerDel_{\adv,E_0}(\secpar)$ as well. As an immediate corollary, we get the following.
\weakCertSecure*
\fi

\section{Expansion To a Full-Blown Scheme}\label{sec:Expansion To a Full-Blown Scheme}
\subsection{Cryptographic Definitions}\label{Cryptographic Definitions}
% -------------------------------------------------------------------------------------------------------------------%
\subsubsection{Universal One-Way Hash Functions}
%-------------------------------------------------------------------------------------------------------------------%

\begin{definition}[{universal one-way hash functions – \nom{woof}{$\Woof$}{Universal one-way hash functions}, adapted from~\cite[Definition 6.4.18]{Gol04}}]
Let  $\ell: \mathbb{N} \rightarrow \mathbb{N}$. A
collection of functions $\{h_s : \Bitspace^\ast \rightarrow \Bitspace^{\ell(|s|)}\}_{s\in\Bitspace^\ast}$ is called universal one-way
hashing ($\Woof$) if there exists a probabilistic polynomial-time algorithm $I$ so that
the following holds:
\begin{enumerate}
\item For all sufficiently large $\secpar$ and all $s$ in the range of $I$, it holds that $|s|=\secpar$.
\item (efficient evaluation): There exists a polynomial-time algorithm that, given $s$ and $x$,
returns $h_s(x)$.
\item (hard-to-form designated collisions):
For every $\QPT$ $\adv_0$, outputting a quantum state $\ket{S}$, and a string x, and a $\QPT$ $\adv$ receiving $\ket{S}$, as well as the description of a hash function $s$ sampled by $I(1^\secpar)$, that outputs a string $y$:
\[\Pr[h_s(x) = h_s(y)
\textit{and } x\neq y]\leq \negl\]
where the probability is taken over the randomness of $I,\adv_0$ and $\adv$.
\end{enumerate}
\end{definition}
\begin{remark}
The existence of post-quantum $\Woof$'s, for any polynomially bounded integer function $\ell$ is implied by the existence of post-quantum one-way functions (\cite{Son14}).
\end{remark}
% %-------------------------------------------------------------------------------------------------------------------%
\subsubsection{Authenticated Encryption}\label{sec:Authenticated Encryption}
%-------------------------------------------------------------------------------------------------------------------%
An authenticated encryption is an encryption scheme that satisfies both chosen cipher-text indistinguishability security and authentication guarantees analogous to the unforgeability guarantees of $\mac$, i.e.,  an adversary cannot create a valid encryption for a document that it has not seen encryptions for. We first recall the definition of $\cca$ security:
\begin{definition}[{Chosen-Cipher-Text Attack Security, \cite[Definition 3.33]{KL14}}]
\label{def: cca security}

A private-key encryption scheme $\Pi = (\keygen, \Enc, \Dec)$  has indistinguishable encryptions under a chosen-cipher-text-attack, or is $\cca$-secure, if for all
$\PPT$ adversaries $\adv$:
\[\Pr[\PrivKcca_{\adv,\Pi}(\secpar) = 1]\leq \frac{1}{2}+\negl,\]
where the probability is taken over all random coins used in the game.
\end{definition}
\begin{savenotes}
\begin{game}
\caption{The $\cca$ Indistinguishability Game $\PrivKcca_{\adv,\Pi}(\secpar)$:}
\label{exp:cca Indistinguishability Game}
\begin{algorithmic}[1]
    \State A random key $k$ is generated by running $\keygen(1^\secpar)$.
    \State The adversary $\adv$ is given input $1^\secpar$ and oracle access to $\Enc_k$ and $\Dec_k$. It outputs a pair of messages $(m_0,m_1)$ of the same length.
    \State A uniform bit $b\in \Bitspace$  is chosen, and then a cipher-text $c\gets   \Enc_k(m_b)$ is given to $\adv$. We call $c$ the challenge cipher-text.
    \State The adversary $\adv$ continues to have oracle access to $\Enc_k$ and $\Dec_k$, but is not allowed to query the latter on the challenge cipher-text itself, Eventually, $\adv$ outputs a bit $b'$.
\end{algorithmic}
 The output of the game is defined to be $1$ if $b = b'$, and
$0$ otherwise. If the output of the game is $1$, we say that $\adv$ succeeds.

\end{game}
\end{savenotes}

As it is in the case of $\mac$, there are two variants of unforgeability: strong unforgeability where even a fresh distinct encryption for an already encrypted plain-text counts as forgery, as opposed to the weaker, albeit more common security notion where the plain-texts are required to be distinct. We require the strong variant, and so we modify the definition in~\cite{KL14} accordingly.

\begin{definition}[{Unforgeable Encryption Scheme, based on~\cite[Definition 4.17]{KL14}} ]
\label{def: Unforgeable Encryption Schem}
 A private-key encryption  scheme  $\Pi = (\keygen, \Enc, \Dec)$ is strongly unforgeable if for all
$\PPT$ adversaries $\adv$:
\[\Pr[\EncForge_{\adv,\Pi}(\secpar) = 1]\leq \negl,\]
where the probability is taken over all random coins used in the game.
\end{definition}
\begin{savenotes}
\begin{game}
\caption{The Strong Unforgeable Encryption Game $\EncForge_{\adv,\Pi}(\secpar)$:}\label{exp:unforgeable encryption experiment}
\begin{algorithmic}[1]
    \State Run $\keygen(1^\secpar)$ to obtain a key $k$.
    \State The adversary $\adv$ is given input $1^\secpar$ and access to an encryption oracle $\Enc_k$. The adversary outputs a cipher-text $c$.

    \State Let $m\equiv\Dec_k(c)$, and let $Q$ denote the set of all responses to queries $\adv$ made to the encryption oracle.
\end{algorithmic}
 The output of the game is $1$  if and only if $(1)$ $m\neq\bot$ and $(2)$ $c\notin Q$.
 \end{game}
\end{savenotes}

A strong authenticated encryption is an encryption scheme that fulfills both of the security notions discussed above:

\begin{definition}[Strong Authenticated Encryption]\label{def:Post Quantum Authenticated Encryption}
A private key encryption scheme is an authenticated encryption if it is $\cca$-secure and strongly unforgeable.

Moreover, we call an encryption scheme a Post-Quantum Classical-Queries Strong Authenticated Encryption if it satisfies both of the security notions ($\cca$ indistinguishability and strong unforgeability), even against $\QPT$ adversaries. 

\end{definition}

Next, we give a brief discussion regarding the assumptions required for post-quantum classical-queries strong authenticated encryption.

A standard construction for authenticated encryption, as detailed in~\cite[Theorem 4.9]{KL14}, is an \emph{encrypt-then-sign} algorithm using a $\cpa$-secure encryption scheme and an unforgeable strong $\mac$ scheme.
Careful inspection of the proof brought in~\cite{KL14} shows that the resulting authenticated encryption is actually a strong authenticated encryption, and also that quantum adversaries with classical access to the oracles, would not have any substantial advantage, as long as the $\mac$ and encryption schemes are also secure against quantum adversaries with classical access to the oracles.
A  result in~\cite{BZ13} proves that post-quantum one-way functions imply the existence of a $\qa\cpa$ encryption scheme, and a $\qa\cma$ $\mac$ scheme, in which quantum adversaries are even allowed to perform quantum queries to the oracles, let alone classical queries. Such schemes thus suffice to construct a post-quantum classical-queries strong authenticated encryption. Hence, we have the following remark summarizing the discussion above.

\begin{remark}\label{rem: OWF->strong authenticated encryption}
If post-quantum one-way functions exist, so does a post-quantum classical-query strong authenticated encryption scheme.
\end{remark}

%-------------------------------------------------------------------------------------------------------------------%
\subsection{Security Proof for the Full Scheme}\label{sec:Security proof for full blown}
%-------------------------------------------------------------------------------------------------------------------%
In this section, we prove the lifting proposition, \cref{prp:oneRtofb}. \cref{fig:structure} illustrates the structure of the expansion of $\CTMAC$ to a full blown scheme, $\sigOrUnrest{\eta}$.
\begin{proof}[Proof of \cref{prp:oneRtofb}]\label{pf:prp:oneRtofb}
The proof is immediate by combining \cref{lem: 1 restricted onetime implies ell-restricted onetime,lem: ell restricted onetime implies unrestricted onetime,lem:-> signing oracle,lem: unrestricted onetime implies full blown unrestricted}, which we prove in this section. These lifting lemmas are achieved via a series of constructions, all of which preserve the tokens being BB84 states, perhaps appending them with classical data as well.
\end{proof}
 
 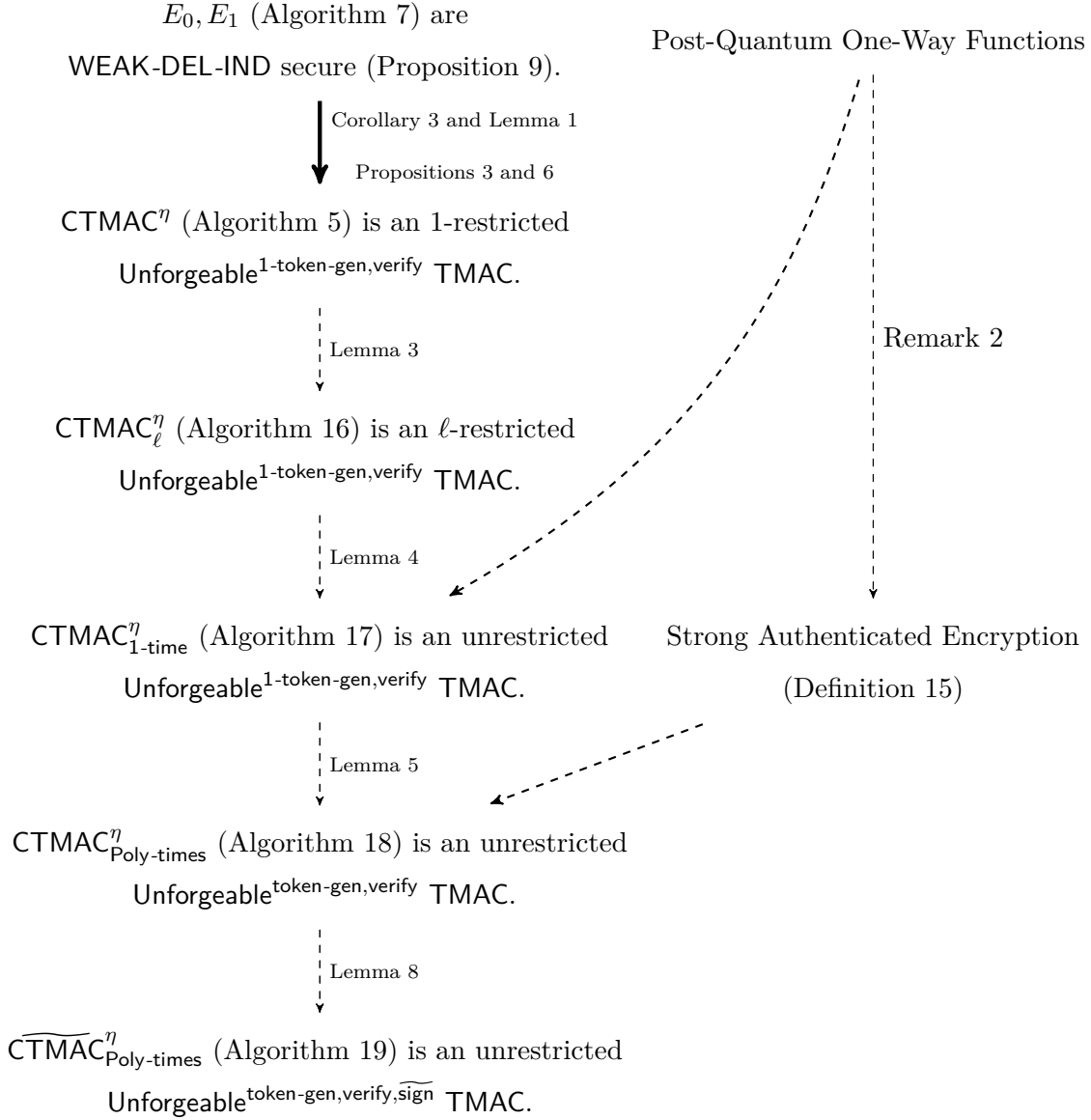
\begin{figure*}[htp]
\centering
\begin{adjustbox}{max width=1.4 \textwidth,center}
    \begin{tikzpicture}%[scale=0.5, align = left]
  \matrix (m) [matrix of nodes, row sep=3em, column sep=0em, align=center]
    {
    $\begin{matrix} \text{$E_0,E_1$ (\cref{alg: E_m}) are}\\ \text{$\exWCerDel$ secure (\cref{prp:parallel-sdp-value}).} \end{matrix}$
    & $\begin{matrix}\text{Post-Quantum One-Way Functions}\end{matrix}$ \\
    %   \text{ $\CTMAC$ (\cref{alg:Conjugate Tmac}) is $\UnfDef^{\OneToken,\VR^*}$ 
    %   }  &\\
    %   \text{$\CTMAC$ is $\UnfDef^{\OneToken,\VR}$} & \\
      $\begin{matrix}\text{$\CTMAC^\eta$ (\cref{alg:noise_tolerant Conjugate Tmac}) is an $1$-restricted}\\ \text{ $\UnfDef^{\OneToken,\Vrfy}$ $\Tmac$.}\end{matrix}$ &\\
      $\begin{matrix}
\text{$\CTMAC^\eta_{\ell}$ (\cref{alg: OTL}) is an $\ell$-restricted}\\
\text{ $\UnfDef^{\OneToken,\Vrfy}$ $\Tmac$.}
\end{matrix}$ & \\
      $\begin{matrix}
\text{$\oneunrest{\eta}$ (\cref{alg: OT}) is an unrestricted}\\
\text{ $\UnfDef^{\OneToken,\Vrfy}$ $\Tmac$.}
\end{matrix}$%\text{$\OT$ (\cref{alg: OT}) is unrestricted and $\UnfDef^{\OneToken,\Vrfy}$}
&$\begin{matrix}\text{Strong Authenticated Encryption}\\ \text{(\cref{def:Post Quantum Authenticated Encryption})}\end{matrix}$%\text{Strong Authenticated Encryption (\cref{def:Post Quantum Authenticated Encryption})}
\\
      $\begin{matrix}
\text{$\polyunrest{\eta}$ (\cref{alg: Tom}) is an unrestricted}\\
\text{ $\UnfDef^{\tokengen,\Vrfy}$ $\Tmac$.}
\\ 
\end{matrix}$%\text{$\TM$ (\cref{alg: Tom}) is unrestricted and $\UnfDef^{\tokengen,\Vrfy}$ } 
& \\
$\begin{matrix}
\text{$\sigOrUnrest{\eta}$ (\cref{alg: TMS}) is an unrestricted} \\
\text{ $\UnfDef^{\tokengen,\Vrfy,\widetilde{\Sign}}$  $\Tmac$.} 
\\

\end{matrix}$
%\text{$\TMS$ (\cref{alg: TMS}) is unrestricted and $\UnfDef^{\tokengen,\Vrfy,\widetilde{\Sign}}$} 
& \\
     };  
     
  { [start chain] \chainin (m-2-1);
    \chainin (m-3-1) [join={node[right,labeled] {\text{\cref{lem: 1 restricted onetime implies ell-restricted onetime}}}  }];
    \chainin (m-4-1) [join={node[right,labeled] {\text{\cref{lem: ell restricted onetime implies unrestricted onetime}}}}];
    \chainin (m-5-1) [join={node[right,labeled] {\text{\cref{lem: unrestricted onetime implies full blown unrestricted}}}  }];
    \chainin (m-6-1) [join={node[right,labeled] {\text{\cref{lem:-> signing oracle}}}  }];
    
     }

  { [start chain] \chainin (m-1-2);
    \chainin (m-4-2) [join={node[auto] {\text{\cref{rem: OWF->strong authenticated encryption}}} }];
    
    }

    \path[->,draw,dashed]
    (m-1-2) edge [pil, bend left=22] node[left,labeled] [scale=0.65] {} (m-4-1);
    
    \path[->,draw,dashed]
    (m-4-2) edge [pil] node[left,labeled] [scale=0.65] {} (m-5-1);

%      }
   \path[->,draw,ultra thick]
   (m-1-1) edge [auto] node[labeled]  [scale=1] 
   {\text{$\begin{matrix}\text{\cref{cor: noise-tolerant weak  basis supply security implies noise-tolerant security,lem:noise reduction lemma}}\\ \text{\cref{prp:weak_to_strong_supplied basis,prp: weak supplied basis security}}
%   \\
   \end{matrix}$
    }} (m-2-1) ;
%hacky-fix    
     
%   { [start chain] \chainin (m-1-2);
%     \chainin (m-7-1) [join={node[right,labeled] {} }];
    
%      }
\end{tikzpicture}
\end{adjustbox} 

\caption[Expansion to Full Blown Scheme]{
The above diagram summarizes the expansion to a full blown scheme showing the different steps in our construction and the reductions in the security proofs. Each node represents a result, and two arrows arising from two different nodes pointing at the same destination node means that the results in the two nodes together imply the result in the destination node. For any $\eta\leq0.07$, all the $\Tmac$ schemes in the diagram above are $\tolerant{2c\eta}$, for every $0\leq c<1$. The bolded arrow is the main technical novelty of this work, and is expanded upon in \cref{fig:proof idea}. 
}
\label{fig:structure}
\end{figure*}

\cref{lem: 1 restricted onetime implies ell-restricted onetime,lem: ell restricted onetime implies unrestricted onetime,lem: unrestricted onetime implies full blown unrestricted} are an adaptation of corresponding lemmas in~\cite[Appendix C]{BS16a}, with few differences. One, unlike~\cite{BS16a}, we extend the results to take into account adversaries that have access to a verification oracle, and two, we only assume the existence of post-quantum one-way functions, instead of post-quantum collision resistant hash functions. Lastly the scheme in \cite{BS16a} was unforgeable even against adversaries with a signing oracle (see discussion in \cref{sec:Comparison with Vanilla Unforgeability}), since it was a strong $\Tmac$ (see \cref{sec:A break of Strong Unforgeability}). The scheme presented in this work is not such a priori, and hence we provide a general lift to unforgeability against adversaries with a signing oracle in \cref{lem:-> signing oracle}.
%-------------------------------------------------------------------------------------------------------------------%
\subsubsection{Expansion to a Length-Restricted Scheme}
%-------------------------------------------------------------------------------------------------------------------%
We extend a $1$-restricted scheme to an $\ell$-restricted scheme (see \cref{def: ell-restricted TMAC}) for $\ell\in \poly$ using a standard construction. This is done by signing the $\ell$ bits of the document using $\ell$ independent instances of the $1$-restricted scheme, each corresponding to a different bit of the document (see \cref{alg: OTL}):

\begin{restatable}{lemma} {OneRtolR}\label{lem: 1 restricted onetime implies ell-restricted onetime}
There is a noise-tolerance preserving lift (\cref{def:noise tolerance}) of any $1$-restricted $\UnfDef^{\OneToken,\Vrfy}$ $\Tmac$ (see \cref{def: ell-restricted TMAC,def:Unforgeability}) to an $\ell$-restricted  $\UnfDef^{\OneToken,\Vrfy}$ $\Tmac$, where $\ell(\secpar)\in \poly$, as shown in \cref{alg: OTL}.
\end{restatable}

\begin{algorithm*}
    \caption{$\OTL$.\\ 
    The resulting scheme obtained by instantiating $\OTL$ with the specific scheme $\CTMAC^\eta$ (\cref{alg:noise_tolerant Conjugate Tmac}), is called $\CTMAC^\eta_{\ell}$. }
    \label{alg: OTL}
    \textbf{Assumes:} $\OTO$ is a $1$-restricted $\Tmac$, and $\ell=\ell(\secpar) \in \poly$.

    \begin{algorithmic}[1] % The number tells where the line numbering should start
        \Procedure{$\keygen$}{$1^\secpar$}
            \For {$i=1,2,\ldots,\ell$}
                 \State $k_i\gets \OTO.\keygen(1^\secpar)$
             \EndFor
            \State \textbf{Return} $(k_1,\ldots,k_\ell)$.
        \EndProcedure
    \end{algorithmic}
    \begin{algorithmic}[1] % The number tells where the line numbering should start
        \Procedure{$\tokengen_k$}{}
            \State Interpret $k=(k_1,\ldots,k_\ell)$.
            \For {$i=1,2,\ldots,\ell$}
                 \State   $\ket{\Stamp_i}\gets\OTO.\tokengen_{k_i}$.
             \EndFor
            \State \textbf{Return} $(\ket{\Stamp_1}\otimes\ldots\otimes\ket{\Stamp_\ell})$.
        \EndProcedure

    \end{algorithmic}
	\begin{algorithmic}[1] % The number tells where the line numbering should start
        \Procedure{$\Sign_{\ket{\Stamp}}$}{$m$}
            \State Interpret $\ket{\Stamp}=(\ket{\Stamp_1}\otimes\ldots\otimes\ket{\Stamp_\ell})$ and $m=(m_1,\ldots,m_\ell$).
            \For {$i=1,2,\ldots,\ell$}
                 \State $\sigma_i\gets\OTO.\Sign$ 
             \EndFor
            \State \textbf{Return} $(\sigma_1,\ldots,\sigma_\ell)$.
                \EndProcedure

    \end{algorithmic}
  	\begin{algorithmic}[1] % The number tells where the line numbering should start
        \Procedure{$\Vrfy_k$}{$m,\sigma$}
            \State Interpret $k$ as $(k_1,\ldots,k_\ell)$, $m$ as $(m_1,\ldots,m_\ell)$ and $\sigma$ as $(\sigma_1,\ldots,\sigma_\ell)$.
            \For {$i=1,2,\ldots,\ell$}
                 \State   $v_i\gets\OTO.\Vrfy_{k_i}(m_i,\sigma_i)$.
             \EndFor
            \State \textbf{Return} $\bigwedge_i v_i$
                \EndProcedure
    \end{algorithmic}
\end{algorithm*}

\begin{proof}

Correctness follows from the correctness of the $\OTO$ scheme, that is, the verification of a valid signature for $m$ accepts if and only if all the $\ell$ verifications for the $\ell$ copies of the $\OTO$ scheme accepts.  Next, we show that the noise-tolerance property is preserved in the transformation from $\OTO$ to $\OTL$. Suppose $\OTO$ is $\tolerant{\delta}$ where $\delta \geq 0$ is some constant. We claim that $\OTL$ is also $\tolerant{\delta}$, i.e., correctness holds for $\OTL$ up to negligible error, in the noise model $\Noise{\delta}$ (see \cref{def:noise-model}). 
Let $m=(m_1,\ldots,m_\ell)\in \Bitspace^\ell$ be arbitrary. 
Let $k\gets  \OTL.\keygen(1^\secpar)$ and $\ket{\Stamp}\gets \OTL.\tokengen_{k}$ be a valid token with respective to the secret key $k$.
The secret key $k$ and $\ket{\Stamp}$ can be viewed as a concatenation of $\ell$ $\OTO$ secret keys and the respective tokens, i.e., 
\[k=k_1\|\cdots\|k_\ell,\quad\ket{\Stamp}=\ket{\Stamp_1}\otimes\cdots \otimes\ket{\Stamp_\ell},\] where $\ket{\Stamp_j}$ is the $j^{th}$ $\OTO$ token with respect to the $\OTO$ secret key $k_j$. Similarly, let $\ket{\Stamp'}$, and $\ket{\Stamp'_1}\otimes\ldots\otimes\ket{\Stamp'_\ell}$ denote the quantum state of the respective tokens after the corruption of some of its qubits due to noise. 
Let $\sigma\gets \OTL.\Sign_{\ket{\Stamp'}}(m)$. Hence, $\sigma=\sigma_1\|\cdots\|\sigma_\ell$, where $\sigma_j\gets \OTO.\Sign_{\ket{\Stamp'_j}}(m_j)$, for every $j\in \ell$.

Let $Good$ be the event that $\OTL.\Vrfy_{k}(\sigma)=1$.
It is enough to show that \[\Pr[Good]\geq 1-\negl,\] where $\negl$ is some negligible function in the security parameter $\secpar$.

Note that by construction,
\begin{equation}\label{eq:Good_wedge}
Good= \bigwedge_{j\in \ell} Good_j,
\end{equation}
where for each $j\in [\ell]$, $Good_j$ denotes the event that in the $\OTL$ verification $\OTO.\Vrfy_{k_j}(\sigma_j)=1$.

By the $\delta$-noise tolerance of $\OTO$, for every $j\in [\ell]$,

\begin{equation}\label{eq:indvl_bound}
\Pr[Good_j]\geq 1-\widetilde{\negl},
\end{equation}
for some negligible function $\widetilde{\negl}$.
Recalling \cref{eq:Good_wedge,eq:indvl_bound} and using the fact that for each of the $\ell$ tokens, $\tokengen$ and the respective $\keygen$ were done independent of each other,
\[
\Pr[Good]=\Pr[\bigwedge_{j\in [\ell]}Good_j]=\prod_{j\in \ell}\Pr[Good_j]\geq(1-\widetilde{\negl})^{\ell(\secpar)}\geq 1 - {\ell(\secpar)}\cdot\widetilde{\negl}.
\] Since, ${\ell(\secpar)}\cdot\widetilde{\negl}$ is a negligible function, this concludes the proof regarding the noise tolerance.

We prove unforgeability as follows. Suppose there is an adversary $\adv$ winning $\UnfExp_{\adv,\OTL}^{\OneToken,\Vrfy}(\secpar)$ (see Game~\ref{exp:UnfExp}) with non-negligible probability $\epsilon(\secpar)$. In other words, $\adv$ gets one signing token and submits $\{m_i,\sigma_i\}_w$ in the $\UnfExp^{\OneToken,\Vrfy}(\secpar)$ game, where the  documents $m_i$ are of length $\ell$, such that with a probability of $\epsilon(\secpar)$, there exists $m_i\neq m_j$, and verification (for all the $\ell$ bits) is successful for both documents. We will construct a new adversary $\Badv$ that attempts to win the $\UnfExp^{\OneToken,\Vrfy}(\secpar)$ game against $\OTO$. $\Badv$ receives one signing token from the challenger. $\Badv$ then runs $\OTO.\keygen$ and $\OTO.\tokengen$ $\ell-1$ times to get additional $\ell-1$ tokens. $\Badv$ shuffles these together, i.e., inserts the token from the challenger to an index, chosen uniformly at random from $[\ell]$, and simulates $\adv$ to obtain $\{m_i,\sigma_i\}_w$. $\Badv$ would then extract from each of the $w$ signed documents, the index of the document and signature corresponding to the token he received from the challenger. $\Badv$ can simulate every oracle query from $\adv$ as follows. $\Badv$ queries the verification oracle for the $\OTO$ at the index corresponding to the challenge token and verifies the other indices by themselves. $\Badv$ returns $1$ if and only if $\OTO.\Vrfy$ outputs $1$, and all the other verifications $\Badv$ made were successful. Since the view of $\adv$ is the same as in $\UnfExp^{\OneToken,\Vrfy}_{\adv,\OTL}(\secpar)$, with non-negligible probability $\epsilon(\secpar)$, $count\geq2$ in 
$\UnfExp^{\OneToken,\Vrfy}_{\adv,\OTL}(\secpar)$, where $count$ is the random variable counting the number of successful verifications, as defined in Game~\ref{exp:UnfExp}. 

Conditioning on the event mentioned above, there exists $m_i\neq m_j$ such that both pass verification on all indices. Since $m_i\neq m_j$, they must differ on at least one of the bits. With a probability of at least $\frac{1}{\ell}$, this would be the same bit corresponding to the input token given to $\Badv$. Conditioning on the event that it is indeed the right index, $\Badv$ would have managed to get alleged signatures for two different documents with respect to $\OTO$, both of which would pass verification with a probability of at least $\epsilon(\secpar)$ (by the success guarantee of $\adv$). 
Here, we also use the fact that the event of $\adv$ winning the game against $\OTL$ is independent of the index in which the real token was inserted because $\adv$ always sees the same distribution no matter how the real token was inserted. Hence, with a probability of at least $\epsilon(\secpar)/\ell$, which is non-negligible, $\Badv$ would win $\UnfExp^{\OneToken,\Vrfy}_{\Badv,\OTO}(\secpar)$, thus, giving us a contradiction.
\end{proof}

%-------------------------------------------------------------------------------------------------------------------%
\subsubsection{An Unrestricted Scheme Using Hash-and-Sign}
%-------------------------------------------------------------------------------------------------------------------%
Next, we use the hash-and-sign paradigm
to extend a length-restricted (see \cref{def: ell-restricted TMAC}) scheme to an unrestricted scheme that can sign documents of arbitrary length. The main idea is to hash the documents to a fixed size and sign their hash. However, a slightly more contrived approach is necessary to achieve unforgeability, since we are using a $\Woof$ and not a $\mathsf{CRH}$ as the hash function. For each message $m$, we sample a fresh hash function, hash the message to a string of length $\ell$, sign the hash value concatenated with the description of the hash function, and then attach the description of the hash function as a part of the signature. 
The construction is described formally in \cref{{alg: OT}}.

\begin{restatable}{lemma}{lRtoRstar}\label{lem: ell restricted onetime implies unrestricted onetime}
Assuming the existence of post-quantum one-way functions, there is a noise-tolerance preserving lift (\cref{def:noise tolerance}) of any $\ell$-restricted $\UnfDef^{\OneToken,\Vrfy}$ $\Tmac$ for $\ell(\secpar)\in \poly$ (see \cref{def: ell-restricted TMAC,def:Unforgeability}), to an unrestricted  $\UnfDef^{\OneToken,\Vrfy}$ $\Tmac$, as shown in \cref{alg: OT}.
\end{restatable}

\begin{savenotes}

\begin{algorithm*}
    \caption{$\OT$.\\ The resulting scheme obtained by instantiating $\OT$ with the specific scheme $\CTMAC^\eta_{\ell}$ (\cref{alg: OTL}), is called $\oneunrest{\eta}$.}
    \label{alg: OT}
    \textbf{Assumes:} $\OTL$ is a $(\ell+\secpar)$-restricted $\Tmac$ for some $\ell=\ell(\secpar) \in \poly$, $\{h_r : \Bitspace^\ast \rightarrow \Bitspace^{\ell(|r|)}\}_{r\in\Bitspace^\ast}$ is a $\Woof$ with indexing function $I$.
    \begin{algorithmic}[1] % The number tells where the line numbering should start
        \Procedure{$\keygen$}{$1^\secpar$}

            \State \textbf{Return} $\OTL.\keygen(1^\secpar)$. 
        \EndProcedure
    \end{algorithmic}
    \begin{algorithmic}[1] % The number tells where the line numbering should start
        \Procedure{$\tokengen_k$}{}
            \State \textbf{Return} $\OTL.\tokengen_{k}$.
        \EndProcedure

    \end{algorithmic}
	\begin{algorithmic}[1] % The number tells where the line numbering should start
        \Procedure{$\Sign_{\ket{\Stamp}}$}{$m$}
            \State $s\gets I(1^\secpar)$ 
            \State \textbf{Return} $(s,\OTL.\Sign_{\ket{\Stamp}}(s||h_s(m)))$.
        \EndProcedure

    \end{algorithmic}
  	\begin{algorithmic}[1] % The number tells where the line numbering should start
    \Procedure{$\Vrfy_k$}{$m,\sigma$}
        \State Interpret $\sigma=(s,\sigma')$
        \State \textbf{Return} $\OTL.\Vrfy_{k}(s||(h_s(m)),\sigma')$.
    \EndProcedure
    \end{algorithmic}
\end{algorithm*}
\end{savenotes}

  \begin{proof}
  
The correctness of the scheme is trivial and, hence, we have omitted the proof. The claim about preservation of noise tolerance is also immediate since all the algorithms are the same in $\OTL$ and $\OT$ up to hashing the message which is a classical procedure, and therefore, is unaffected by the noise in the quantum communication channel.

We now prove that the unrestricted scheme $\OT$ is $\UnfDef^{\OneToken,\Vrfy}$. Suppose by way of contradiction, there exists a $\QPT$ adversary $\adv$ that wins $\UnfExp^{\OneToken,\Vrfy}_{\OT}(\secpar)$ with non-negligible probability. Let $\{m_i,(s_i,
\sigma'_i)\}_{w}$ be the documents submitted by $\adv$.
First, notice that except with negligible probability, for any $i<j$, $(s_i||h_{s_i}(m_i))\neq (s_j||h_{s_j}(m_j))$. Otherwise, $s_i=s_j$, and $\adv$ could be used to find a designated collision in $h_{s_i}$, by first uniformly guessing the index corresponding to $i$  and outputting $m_i$, and then simulating the algorithm to gather $m_j$. Since $\adv$ wins with non-negligible probability, the event that $\adv$ wins, and that there is no collision between the hashes of distinct documents, occurs with non-negligible probability. Let this probability be $\epsilon(\secpar)$. Conditioning on the above-mentioned event, there is a set $S'$ of size greater then $r+1$ such that for all $i,j\in S'$:
\begin{align}
&(s_i||h_{s_i}(m_i))\neq (s_j||h_{s_j}(m_j)),\\
&\OTL.\Vrfy_k((s_i||h_{s_i}(m_i)),\sigma'_i))=1,\\
&\OTL.\Vrfy_k((s_j||h_{s_j}(m_j)),(\sigma'_j))=1.
\end{align}

It is thus easy to use $\adv$ to construct an adversary $\Badv$ for $\OTL$ winning with probability $\epsilon(\secpar)$, contradicting that $\OTL$ is $\UnfDef^{\OneToken,\Vrfy}$.

\end{proof}
   
%-------------------------------------------------------------------------------------------------------------------%
 \subsubsection{Unforgeability Against Polynomial Tokens Attacks}\label{sec:Unforgeability Against Polynomial Tokens Attacks}
%-------------------------------------------------------------------------------------------------------------------%
In this section, we show how to lift any $\UnfDef^{\OneToken,\Vrfy}$ unrestricted scheme, such as $\OT$ (\cref{alg: OT}), to an $\UnfDef^{\tokengen,\Vrfy}$ unrestricted scheme $\TM$ (\cref{alg: Tom}), using an authenticated encryption scheme ($\Aenc$). In this full-blown scheme, $\TM$, the secret key is the key for the authenticated scheme. A token is generated by sampling a fresh secret key for $\OT$ and then running $\OT.\tokengen$ on the $\OT$ secret key to generate an $\OT$ token. The $\OT$ token, along with the encrypted secret key, forms the token for the full-blown scheme $\TM$.

A document is signed by first signing it under the $\OT$ token to get a $\OT$ signature for it, and appending the encrypted secret key to it. 
Verification of an alleged document-signature pair is done by first decrypting the cipher corresponding to the $\OT$ secret key, which is then used to run $\OT.\Vrfy$ on the other part of the signature. Verification accepts if and only if the authentication of the cipher accepts, and the $\OT$ verification succeeds. 

We demand that the $\Aenc$ scheme used here satisfies a property analogous to that of strong $\mac$. Assuming the existence of post-quantum one-way functions, authenticated encryptions schemes satisfying this strong property do exist, see \cref{rem: OWF->strong authenticated encryption}. 

\begin{restatable}{lemma}{RstarTofb}\label{lem: unrestricted onetime implies full blown unrestricted} 
Assuming the existence of a post-quantum classical-queries strong authenticated encryption scheme, there is a noise-tolerance preserving lift (\cref{def:noise tolerance}) of any unrestricted  $\UnfDef^{\OneToken,\Vrfy}$ $\Tmac$ (see \cref{def: ell-restricted TMAC,def:Unforgeability}) to an unrestricted  $\UnfDef^{\tokengen,\Vrfy}$ $\Tmac$, as shown in \cref{alg: Tom}.

\end{restatable} 

\begin{algorithm} %\label{alg:TOM}
    \caption{$\TM$.\\ The resulting scheme obtained by instantiating $\TM$ with the specific scheme $\oneunrest{\eta}$ (\cref{alg: OT}) is called $\polyunrest{\eta}$.}
    \label{alg: Tom}
    \textbf{Assumes:} $\OT$ is an unrestricted $\Tmac$, and $\Aenc$ is a post-quantum classical-queries strong authenticated encryption scheme. 
    \begin{algorithmic}[1] % The number tells where the line numbering should start
        \Procedure{$\keygen$}{$1^\secpar$}
            \State \textbf{Return} $\Aenc.\keygen(1^\secpar)$.
        \EndProcedure
    \end{algorithmic}
    \begin{algorithmic}[1] % The number tells where the line numbering should start
        \Procedure{$\tokengen_k$}{}
            \State  Run  $\OTO.\keygen(1^\secpar)$ to obtain $\kappa$.
         \State    Run $\OTO.\tokengen_{\kappa}$ to obtain $\ket{\widetilde{\Stamp}}$.
           \State \textbf{Return} $(\ket{\widetilde{\Stamp}}, \Aenc.\enc_k(\kappa))$.
        \EndProcedure

    \end{algorithmic}
	\begin{algorithmic}[1] % The number tells where the line numbering should start
        \Procedure{$\Sign_{\ket{\Stamp}}$}{$m$}
            \State Interpret $\ket{\Stamp}$ as $(\ket{\widetilde{\Stamp}},e)$
            \State \textbf{Return} $(\OTO.\Sign_{\ket{\widetilde{\Stamp}}}(m),e)$.
        \EndProcedure

    \end{algorithmic}
  	\begin{algorithmic}[1] % The number tells where the line numbering should start
    \Procedure{$\Vrfy_k$}{$m,\sigma$}
            \State Interpret $\sigma$ as $(s,e)$
            \State  $\kappa\gets\Aenc.\Dec_{k}(e)$.
            \If{$\kappa=\bot$} 
                \State \textbf{Return 0}
            \Else
                \State \textbf{Return} $\OTO.\Vrfy_{\kappa}(m,s)$.
            \EndIf
    \EndProcedure
    \end{algorithmic}

\end{algorithm}

\begin{proof}
Given an unrestricted scheme $\UnfDef^{\OneToken,\Vrfy}$ $\Tmac$  $\OT$, we construct another unrestricted scheme $\TM$ given in \cref{alg: Tom} that is $\UnfDef^{\tokengen,\Vrfy}$.

Correctness: It is clear that in the $\TM$ scheme (\cref{alg: Tom}), the verification procedure accepts the output of the signing procedure  (assuming the same holds for $\OT$). 

The claim about preservation of noise tolerance is also immediate and, hence, we omit the proof.

Next, we show that $\TM$ is $\UnfDef^{\tokengen,\Vrfy}$. Assume towards contradiction, that there is a $\QPT$ adversary with oracle access to verification, $\adv$,  getting $r$ signing tokens, and submitting $w$ signed documents such that with non-negligible probability $r+1$ distinct documents all pass verification, meaning $\UnfExp^{\tokengen,\Vrfy}_{\adv,\TM}(\secpar)=1$ with non-negligible probability. Without loss of generality, assume that $\adv$ always outputs a fixed (per $\secpar$) number of signed documents $w$.

Let $\kaut$ be the secret key generated for $\TM$, which by construction, is a key for the authenticated encryption scheme $\Aenc$. Let $d_1,d_2,\ldots,d_{w}$ be the $w$ signed documents submitted by $\adv$, and let $S$ be the set of indices that the challenger accepts. Without loss of generality, each $d_i$ must have the form $d_i=(m_i,s_i,e_i)$, where $m_i$ is the document that was signed, $e_i$ is an encryption of the associated secret key for the $\OT$ scheme (whose decryption we will refer to as $\kappa_i$), and each $s_i$ is the  signature for $m_i$ by $\OT$. For all $i\in S$, since $\Vrfy_{\OT}(d_i)= 1$,   $e_i$ is a valid encryption under $\Aenc$, and 
$\Dec_{\kaut}(e_i)$ outputs a string $p_i$ such that $\OT.\Vrfy_{p_i}(m_i,s_i)$ accepts.

  Denote $\mathcal{E}_{Win}^\secpar$ to be the event that $\UnfExp^{\tokengen,\Vrfy}_{\adv,\TM}(\secpar)=1$, i.e.,
 \[ count\geq r+1,\]  where $count$ is the random variable counting the number of successful verifications as defined in Game~\ref{exp:UnfExp}. 
 Similarly, let $\mathcal{E}_{Dup}^\secpar$ be the event that for every $i\in S$, $e_i$ appeared in one of the tokens sent by the challenger in response to the $\tokengen$ queries of $\adv$. 
  %\[\mathcal{E}_{Dup}^\secpar\equiv\{ \forall i\in S, \textit{ $e_i$ was appended within the tokens}\},\]
  Let,
  \begin{equation}\label{eq:evt: E_DUP and E_win}
\mathcal{E}^\secpar\equiv\mathcal{E}^\secpar_{Dup}\bigwedge \mathcal{E}_{win}^\secpar.
\end{equation}
We will prove the following lemma:
  \begin{lemma}
    \label{clm: y(secpar is non-negligible)}
    $\Pr[\mathcal{E}^\secpar]$ is non-negligible.
\end{lemma}
  \begin{proof}[Sketch proof]
       Suppose not, and  $\Pr[\mathcal{E}^\secpar]$ is negligible. Since $\Pr[\mathcal{E}_{Win}^\secpar]$ is non-negligible by assumption, the supposition implies $\Pr[\mathcal{E}_{Win}^\secpar\wedge (\mathcal{E}_{Dup}^\secpar)^c]$ is non-negligible. Denote by $\epsilon(\secpar)$ the non-negligible probability that both of the following events occur:
    \begin{enumerate}
        \item $(\mathcal{E}_{Win}^\secpar)$: $\UnfExp_{\adv,\TM}^{\tokengen,\Vrfy}(\secpar)=1$.
        \item $(\mathcal{E}_{Dup}^\secpar)^c$: There is an $i\in S$ such that $e_i$ did not appear in the tokens sent by the challenger.
    \end{enumerate}
    It is straightforward to construct a $\QPT$  $\Badv_{frg}$ with oracle access to $\Aenc.\Enc_{\kaut}(m,\sigma)$, that guesses the right index for $i$, and succeeds in  winning  $\EncForge_{\Badv_{frg},\Aenc}(\secpar)$ (Game~\ref{exp:unforgeable encryption experiment}) with a probability that is greater than $\epsilon(\secpar)$ times an inverse polynomial, implying that $\Pr[\mathcal{E}^\secpar]$ must be non-negligible. Here we use the fact that $\Aenc$ is a strong \emph{authenticated} encryption.
    \end{proof}

Conditioning on $\mathcal{E}^\secpar$, there are at most $r$ distinct values for the $e_i$'s. Hence, by the pigeonhole principle for some $i,j\in S $   $e'\equiv e_i=e_j $, and since $e'$ was appended in one of the tokens, it must be an encryption of  one of the secret keys for $\OT$ generated by the challenger during $\TM.\tokengen$. Denote $p\equiv\Dec_{\kaut}(e')$.

 Since $a_i=(m_i,s_i,e_i)$ and $a_j=(m_j,s_j,e_j)$ are accepted, we know that $\OT.\Vrfy_p(m_i, s_i)=1$ and $\OT.\Vrfy_p(m_j, s_j)=1$. This is already very close to forgery of the $\OT$ scheme, only $\adv$ has access to some extra information, namely the encrypted secret keys. In order to complete our proof via contradiction, we will construct an adversary $\Badv$ against the one-time scheme $\OT$ (\cref{alg: OT}) in the $\UnfExp^{\OneToken, \Vrfy}(\secpar)$ game (see Game~\ref{exp:unforgeable encryption experiment}), and using the $CCA$ encryption property of the underlying $\Aenc$ scheme, show that the adversary wins with non-negligible probability.
\paragraph{An adversary for $\OT$.}

$\Badv_{\OT}$ will act in the following manner: given a single token $\ket{\Stamp'}$, $\Badv_{\OT}$ will run $\OT.\keygen$ and $\OT.\tokengen$ $r-1$ times to produce $r-1$ secret keys and the corresponding tokens, respectively. $\Badv_{\OT}$ will then run $\Aenc.\keygen$ to generate the secret key $\kaut$ and use it to encrypt all the secret keys and append them accordingly to the tokens. Lastly, $\Badv_{\OT}$ will call $\OT.\keygen$ one more time to create a bogus secret key $\Bog$, encrypt it, and append it to $\ket{\Stamp'}$ (the original challenge token from the challenger) as the $r^{th}$ token. We denote the secret keys generated by $\Badv_{\OT}$ as the set $\{\kappa_i\}_i$. $\Badv_{\OT}$ will then shuffle those tokens and feed them to $\adv$. Whenever $\adv$ queries the oracle with $(m,sig,\kenc)$, $\Badv_{\OT}$ will do the following steps:
\begin{enumerate}
    \item Obtain $\Dec_{\kaut}(\kenc)=\kappa$, and if this fails, reject.
    \item If $\kappa$ is the bogus key $\Bog$, query the oracle on $(m,sig)$ and return its answer. Otherwise, check that $\OT.\Vrfy_{\kappa}(m,sig)=1$.

\end{enumerate}

 If all tests pass $\adv$ answers $1$.
 
 Upon receiving $\adv$'s answer $\{d_i\}_{i\in S}$, $\Badv_{\OT}$ will guess $\hat{i},\hat{j}\in [w]$ and output $(m_{\hat{i}},s_{\hat{i}}), (m_{\hat{j}},s_{\hat{j}})$.
 
 In a similar manner as before, we define analogous events in the simulation of $\adv$ in $\Badv_{\OT}$.
 Let $\mathcal{F}_{Win}^\secpar$ be the event within the $\Badv_\OT$ simulation, $\adv$ produces at least $r+1$ signed documents which would pass verification, i.e.,
 \[\mathcal{F}_{Win}^\secpar\equiv\{\text{within the simulation, }  count\geq r+1\}, \]
  where $count$ is the random variable counting the number of successful verifications as defined in Game~\ref{exp:UnfExp}.
Similarly, let $\mathcal{F}_{Dup}^\secpar,\mathcal{F}^\secpar$ be the events defined as follows:
  \[\mathcal{F}_{Dup}^\secpar\equiv\{\text{within the simulation, $\forall i\in S$, $e_i$ appeared in the tokens provided to $\adv$}\}. \]
 \begin{equation}\label{evt: F_DUP and F_win}
\mathcal{F}^\secpar\equiv\mathcal{F}_{Dup}^\secpar\bigwedge \mathcal{F}_{win}^\secpar \end{equation}
If $\mathcal{F}^\secpar$ occurs with non-negligible probability, then clearly $\Badv_{\OT}$ succeeds with non-negligible probability. The deduction holds due to the same arguments that hold for the  event in \cref{eq:evt: E_DUP and E_win}. The main arguments are that there must be some $i\neq j$ such that $e'\equiv e_i=e_j$, and 
 $\OT.\Vrfy_p(m_j, s_j)=1,\OT.\Vrfy_p(m_i, s_i)=1$.
Since the tokens were shuffled by $\Badv_{\OT}$ before being submitted to $\adv$, $e'$ corresponds to any given copy of the $\OT$ scheme with probability $\frac{1}{r}$. Clearly, $\Badv_{\OT}$ would pick the correct $i,j$ with a probability of at least $\frac{1}{w^2}$.
Hence, $\UnfExp_{\Badv_{\OT},\OT}^{\OneToken,\Vrfy}(\secpar)=1$ with probability at least $\frac{1}{rw^2}\Pr[\mathcal{F}^\secpar]$, which is non-negligible if $\Pr[\mathcal{F}^\secpar]$ is non-negligible.

Hence, the last thing that we need to complete the proof of \cref{lem: unrestricted onetime implies full blown unrestricted} (by reaching the desired contradiction), is to prove that $\Pr[\mathcal{F}^\secpar]$ is non-negligible, which we do next.
 
 \begin{lemma}
     \label{lem: z(secpar) is non-negligible}
        $\Pr[\mathcal{F}^\secpar]$ is non-negligible.
 \end{lemma}

\begin{proof} 
Assume to the contrary, that  $\Pr[\mathcal{F}^\secpar]$ is negligible. We will construct an efficient distinguisher for $\Aenc$ that wins the indistinguishability game in \cref{def: cca security} with non-negligible advantage.

\paragraph{A Distinguisher for $\Aenc$.}
$\Badv_{ind}$ is a distinguisher for Game~\ref{exp:cca Indistinguishability Game}. $\Badv_{ind}$ acts in the following manner: 
$\Badv_{ind}$ runs $\OT.\keygen, \OT.\tokengen$ $r$ times to create secret keys $\{\kappa_i\}_i$ and corresponding tokens   $\{\ket{\Stamp_i}\}_i$. $\Badv_{ind}$ then asks the encryption oracle to encrypt all but one of the secret keys $\{\kappa_i\}_{2\leq i\leq r}$. Then $\Badv_{ind}$ appends the corresponding encryptions to the respective tokens, barring the first token, $\ket{\Stamp_1}$.
 
$\Badv_{ind}$ runs $\OT.\keygen$ one last time to get a fresh secret key $p$, and sends $p$ and $\kappa_1$ as the challenge messages, and gets back $c$ which is the challenge cipher. $\Badv_{ind}$ appends $c$ to $\ket{\Stamp_1}$ to get a $\TM$ token. 

$\Badv_{ind}$ will then shuffle those tokens and feed them to $\adv$. Whenever $\adv$ queries the oracle with $(m,sig,\kenc)$ $\Badv_{ind}$ will:
\begin{enumerate}
    \item  If $\kenc\neq c$, call a decryption oracle to extract $\hat{\kappa}\equiv \Dec_{\kaut}(\kenc)$, or reject if not valid, and return $\OT.\Vrfy_{\hat{\kappa}}(m,sig)=1$
    \item If $\kenc=c$, check that $\OT.\Vrfy_{\kappa_1}(m,sig)=1$.
\end{enumerate}
 
 When $\adv$ submits $\{d_i\}_w$, $\Badv_{ind}$ would verify all of them in the exact same manner as that of oracle queries, construct the subset $S$, verify that for all $i\in S$, $e_i$'s are such that were appended in the tokens, and that $\abs{S}\geq r+1$. If all of the above-mentioned verification holds, $\Badv_{ind}$ would output $1$, and otherwise outputs $0$. 

Notice that if the challenge cipher $c$ is an encryption of $p$, the key unrelated to $\ket{\Stamp_1}$, then $\adv$'s view is exactly the same as it is in the simulation of $\Badv_{\OT}$. Hence, in that case, so $\Badv_{ind}$ will output $1$ with a probability of $\Pr[\mathcal{F}^\secpar]$. On the other hand, if $c=\enc_{\kaut}(\kappa_1)$, then the view of $\adv$ is the same as in $\UnfExp_{\adv,\TM}^{\tokengen,\Vrfy}(\secpar)$, so it will succeed with probability $\Pr[\mathcal{E}^\secpar]$. Hence,
\[|\Pr[\Badv_{ind}(enc_{\kaut}(p),\secpar)=1]-\Pr[\Badv_{ind}(enc_{\kaut}(\kappa_1),\secpar)=1]|=\Pr[\mathcal{F}^\secpar]-\Pr[\mathcal{E}^\secpar]|,\] which is non-negligible since $\Pr[\mathcal{F}^\secpar]$ is negligible as per assumption, but $\Pr[\mathcal{E}^\secpar]$ is non-negligible (\cref{clm: y(secpar is non-negligible)}). This is a contradiction to the security of $\Aenc$. 
Hence, it must be the case that $\Pr[\mathcal{F}^\secpar]$ is non-negligible, thus concluding the proof of the lemma.
 \end{proof}
\cref{lem: z(secpar) is non-negligible}
in turn finishes the proof of \cref{lem: unrestricted onetime implies full blown unrestricted} as mentioned above. 
  \end{proof}
 %-------------------------------------------------------------------------------------------------------------------%
 \subsubsection{Unforgeability in the Presence of a Signing Oracle}\label{sec:Unforgeability In The Presence of a Signing Oracle}
%-------------------------------------------------------------------------------------------------------------------%
Lastly, an $\UnfDef^{\tokengen,\Vrfy}$ $\Tmac$ scheme is wrapped with randomness to get an $\UnfDef^{\tokengen,\Vrfy,\widetilde{\Sign}}$ $\Tmac$ scheme.  The idea is relatively simple. One can think of a naive reduction from $\UnfDef^{\tokengen,\Vrfy,\widetilde{\Sign}}$ to $\UnfDef^{\tokengen,\Vrfy}$ by simulating signing queries using additional tokens. However, the reduction fails essentially because the signing oracle ($\widetilde{\Sign}_k$) may be queried several times with the same document, providing different responses each time (see  \cref{sec:Comparison with Vanilla Unforgeability} for more details). 
A way to enforce that there are no multiple signing queries for the same document is to concatenate the document to be signed with fresh randomness. The randomness is then provided as part of the signature. Verification is done by simply verifying the document concatenated with the proclaimed randomness. In this way, we circumvent the issue mentioned above, while maintaining all previous properties. The full description of the construction is provided in \cref{alg: TMS}.

 \begin{restatable}{lemma}{fbToSIGN}\label{lem:-> signing oracle}

 There is a noise-tolerance preserving lift (\cref{def:noise tolerance}) of any unrestricted $\UnfDef^{\tokengen,\Vrfy}$ $\Tmac$ (see \cref{def: ell-restricted TMAC,def:Unforgeability}), to an unrestricted  $\UnfDef^{\tokengen,\Vrfy,\widetilde{\Sign}}$ $\Tmac$, as shown in \cref{alg: TMS}.
\end{restatable}

\begin{algorithm} 
    \caption{$\TMS$.\\ The resulting scheme obtained by instantiating $\TMS$ with the specific scheme $\polyunrest{\eta}$ (\cref{alg: Tom}), is called $\sigOrUnrest{\eta}$.  }
    \label{alg: TMS}
    \textbf{Assumes:} $\TM$ is an unrestricted $\Tmac$.     
    \begin{algorithmic}[1] % The number tells where the line numbering should start
        \Procedure{$\keygen$}{$1^\secpar$}
            \State \textbf{Return} $\TM.\keygen(1^\secpar)$.
        \EndProcedure
    \end{algorithmic}
    \begin{algorithmic}[1] % The number tells where the line numbering should start
        \Procedure{$\tokengen_k$}{}
         \State  \textbf{Return}   $\TM.\tokengen_{k}$
        \EndProcedure

    \end{algorithmic}
	\begin{algorithmic}[1] % The number tells where the line numbering should start
        \Procedure{$\Sign_{\ket{\Stamp}}$}{$m$}
            \State $\rand\sample \Bitspace^\secpar$
            \State \textbf{Return} $(\rand,\TM.\Sign_{\ket{\Stamp}}(m||\rand))$.
        \EndProcedure

    \end{algorithmic}
  	\begin{algorithmic}[1] % The number tells where the line numbering should start
    \Procedure{$\Vrfy_k$}{$m,\sigma$}
            \State Interpret $\sigma$ as $(\rand,sig)$.
            \State \textbf{Return} $\TM.\Vrfy_{k}(m||\rand,sig)$.
    \EndProcedure
    \end{algorithmic} 

\end{algorithm}

\begin{proof}
The correctness is immediate. The claim about preservation of noise tolerance is also immediate since all the algorithms are the same in $\TM$ and $\TMS$ up to adding randomness to the message, which is a classical procedure and is, hence, unaffected by noise in the quantum communication channel.

For unforgeability, let $\adv$ be an adversary winning $\UnfExp^{\tokengen,\Vrfy,\widetilde{\Sign}}_{\adv,\TMS}(\secpar)$ with non-negligible probability $\epsilon(\secpar)$. Without loss of generality, it can be assumed that $\adv$  always  makes exactly $j$ oracle queries to $\widetilde{\Sign}$ for some $j\in \poly$, and asks for $r\in \poly$ tokens. 

% Note that by construction, for $i\in[j]$, the $i^{th}$ signing query is made with respect to a document $m_i$, the returned signature is of the form $(\rand_i,\TM.\widetilde{\Sign}(m_i|| \rand_i))$ where $\rand_i$ is chosen uniformly and independently at random. Hence, with overwhelming probability, all $m_i\equivm_i|| \rand_i$ are distinct. Denote this probability by $\delta(\secpar)$

Next, we construct a corresponding adversary $\Badv$ winning  $\UnfExp^{\tokengen,\Vrfy}_{\Badv,\TM}(\secpar)$  with non-negligible probability. $\Badv$ will use $r+j$ tokens, and run $\adv$, supplied with $r$ of those tokens. If $\adv$ makes a query of the form $(m,(\rand,\sigma))$ to the verification oracle, $\Badv$ would query $(m||\rand,\sigma)$ to its verification oracle in $\UnfExp^{\tokengen, \widetilde{sign},\Vrfy}_{\TM}$ and answer accordingly. If $\widetilde{\Sign}_{k}$ is called with a document query $m_i$, $\Badv$ will sample a random $\rand_i\in \Bitspace^\secpar$, compute $\sigma_i\gets\Sign_{\ket{\Stamp_i}}(m_i||\rand_i)$ with one of its remaining tokens $\ket{\Stamp_i}$, and return to $\adv$ the response $(\rand_i,\sigma_i)$. As $\adv$ would make $j$ calls to $\widetilde{\Sign}_k$, $\Badv$ will not run out of tokens. $\adv$ outputs $\hat{m}_1,\ldots,\hat{m}_{w}$ and corresponding  signatures $(\hat{\rand}_1,\hat{\sigma}_1),\ldots,(\hat{\rand}_w,\hat{\sigma}_{w})$. $\Badv$ will then extract from that the documents $\hat{m}_1||\hat{\rand}_1,\ldots,\hat{m}_{w}||\hat{\rand}_{w}$, and  $\hat{\sigma}_1,\ldots,\hat{\sigma}_{w}$ as the corresponding signatures for them. $\Badv$ would output those, in addition to those signed documents that he generated by itself $(m_i||\rand_i,\sigma_i)_{i\in [j]}$. Due to the randomness of $(\rand_i)_{i\in [j]}$, they  are all distinct with overwhelming probability $\delta(\secpar)$, meaning  $(m_i||\rand_i)_{i\in [j]}$ are also distinct with that probability. These documents are bound to pass verification, as they were signed by the use of a token. In the winning event for $\adv$, there is a $r+1$  subset of   $(\hat{m}_i||\hat{\rand}_i)_{i\in [w]}$, which are both distinct from $(m_i||\rand_i)_{i\in [j]}$ and successfully pass verification (see the winning condition for  Game~\ref{exp:UnfExp}). As the view of $\adv$ is the same as in the true game $\UnfExp_{\adv,\TMS}^{\tokengen, \widetilde{sign},\Vrfy}$, this occurs with probability $\epsilon(\secpar)$. $\Badv$ could only lose only if $\adv$ loses in the simulation, or if the randomness sampled was not distinct(or both). By the union bound this means $\Badv$ loses  $\UnfExp^{\tokengen,\Vrfy}_{\Badv,\TM}(\secpar)$ with probability at most $1-\delta(\secpar)+1-\epsilon(\secpar)$, or alternatively, $\Badv$ wins with probability $\epsilon(\secpar)-(1-\delta(\secpar))$, which is a  non-negligible function, meaning such  $\adv$ cannot exist.
\end{proof}
\ifnum\llncs=0

\fi

%-------------------------------------------------------------------------------------------------------------------%
\section{Drawbacks of the Conjugate TMAC}\label{sec:Drawbacks Of the Conjugate Tmac} 
%-------------------------------------------------------------------------------------------------------------------%
In this section, we bring two attacks against $\CTMAC$ (see \cref{sec:Conjugate TMAC}): one by quantum access to the verification oracle, and the other violating strong unforgeability. 
%-------------------------------------------------------------------------------------------------------------------%
\subsection{A Quantum Superposition Attack}\label{sec:A Quantum Superposition Attack}
%-------------------------------------------------------------------------------------------------------------------%
As briefly discussed in \cref{sec: notions of security}, for $\Tmac$s with deterministic verifications, as $\CTMAC$ is, it is straightforward to define a stronger security notion where security is preserved even if the adversary is allowed quantum access to the verification function.
\begin{definition}\label{def:quantum queries unforgeable TMAC}
$\UnfExp^{(\ket{\Vrfy},\cdot,\cdot)}(\secpar)$ is the same as $\UnfExp^{(\Vrfy,\cdot,\cdot)}(\secpar)$, except that $\adv$ instead has access to the quantum unitary $\mathsf{U}_{\Vrfy_k}$ (which it can query polynomially many times) instead of $\Vrfy_k$.
\end{definition}

It is emphasized that $\UnfExp^{(\ket{\Vrfy},\cdot,\cdot)}(\secpar)$ is defined only when verification is deterministic, as $\mathsf{U}_{f}$ is only defined when $f$ is a function.
\begin{proposition}\label{prp: forgeability in light of super position attacks}
For all $\secpar$, there exists an adversary $\adv$ winning $\UnfExp^{\OneToken,\ket{\Vrfy}}_{\adv,\CTMAC}(\secpar)$ with certainty.

\end{proposition}
\begin{proof}
The  attack is very similar to the one in~\cite{Lut10}.
The idea is to learn the token one qubit at a time: given a token $\ket{\Stamp}$, to learn the $i^{th}$ qubit $\adv$ would query the verification oracle with $\ket{\left(0,X_i\Stamp\right)}\ket{0}$ for $X_i\equiv I\otimes\ldots \otimes X\otimes\ldots I$, where $X$ operates on the $i^{th}$, and $I$ is the identity. $\adv$ would then measure the result register. If verification passed, the $i^{th}$ qubit is either $\ket{+}$ or $\ket{-}$, and the signature register submitted has not been damaged at all by the measurement. If verification failed, the $i^{th}$ qubit is either $\ket{0}$ or $\ket{1}$. The adversary then restores $\ket{\Stamp}$ by applying $X_i$. In either case, measuring the qubit on the correct basis that the attacker now knows would reveal the exact quantum state of the $i^{th}$ qubit of $\ket{\Stamp}$. 
$\adv$ could then reassemble $\ket{\Stamp}$ and repeat the process above for all qubits, essentially uncovering the secret key.
\end{proof}
%-------------------------------------------------------------------------------------------------------------------%
\subsection{A Break of Strong Unforgeability}\label{sec:A break of Strong Unforgeability}%------------------------------------------------------------------------------------- ------------------------------%
 We can also consider a stricter definition regarding what counts as forgery. The security definition brought in \cref{sec: notions of security} only guarantees the inability of an adversary provided with $r$ token to produce $r+1$ signatures for distinct fresh documents. A stronger security notion would prohibit even the creation of $r+1$ fresh distinct signed documents, even if the documents themselves are not distinct, or not fresh. Often this discussion is null, as vanilla $\mac$ schemes usually have unique signatures. The $\Tmac$ schemes we present, however, have exponentially many signatures for every document. In analogy to strong $\mac$ schemes~\cite[Definition 4.3]{KL14}, we can thus also consider a strong variant of all of the above, and define corresponding security notions $\str\textit{-}\UnfDef^{(\cdot,\cdot,\cdot)}$, and security games $\str\textit{-}\UnfExp^{(\cdot,\cdot,\cdot)}(\secpar)$.
\begin{definition} \label{def:strong TMAC}
$\str\textit{-}\UnfExp^{(\cdot,\cdot,\cdot)}(\secpar)$ differs from $\UnfExp^{(\cdot,\cdot,\cdot)}(\secpar)$ (Game~\ref{exp:UnfExp}) by the following:

\begin{enumerate}
    \item The set $Q$ is defined to be all query-response pairs made to $\widetilde{\Sign}$, instead of only the queries.
    \item $count=\abs{\{(m_i,\sigma_i)| i\in S\wedge (m_i,\sigma_i)\notin Q\}}$.
\end{enumerate}
\end{definition}

\begin{proposition}\label{prp: strong forgery of Conjugate Tmac}
For all $\secpar$, there exists an adversary $\adv$ winning $\str\textit{-}\UnfExp^{\OneToken,\Vrfy}_{\adv,\CTMAC}(\secpar)$ with probability $\frac{1}{2}$.
\end{proposition}
\begin{proof}
 For a key $k\equiv (a,b)$ generated by $\Cadv$, $\adv$ receives a token $\ket{\Stamp}=H^b\ket{a}$, and measures the token in the standard basis to obtain a classical string $\sigma$, representing a signature for $0$. There are, however, many other valid signatures this process could have generated. A change of any bit in $\sigma$ corresponding to a coordinate $i$ for which $b_i=1$ would result in such a signature. On average,  half the coordinates  are such, and all an adversary must do to uncover a fresh signature, is pick one bit of the valid signature and perturb it. The adversary thus wins with a probability of $\frac{1}{2}$.
\end{proof}

%-------------------------------------------------------------------------------------------------------------------%
\section{TMAC as a MAC}\label{sec:Comparison with Vanilla Unforgeability}

%-------------------------------------------------------------------------------------------------------------------%
In this section, we show how a $\Tmac$ can function as an unforgeable $\mac$ scheme, provided the $\Tmac$ is unforgeable against attacks with signing and verification oracles:
\begin{theorem}
Let $\TM$ be an $\UnfDef^{\widetilde{\Sign},\Vrfy}$ (see \cref{def:Unforgeability}) $\Tmac$. Then the $\mac$ scheme $\Cmac=(\TM.\keygen,\TM.\widetilde{\Sign},\TM.\Vrfy)$ derived from $\TM$, is $\cma$ unforgeable.
\end{theorem}
The proof of this theorem follows by definition and, hence, is omitted.

It is tempting to argue that an $\UnfDef^{\tokengen,\Vrfy}$ $\Tmac$ is sufficient for the above theorem. An equivalent question is whether an $\UnfDef^{\tokengen,\Vrfy}$ $\Tmac$ is also an $\UnfDef^{\widetilde{\Sign},\Vrfy}$ $\Tmac$. The na\"ive intuition being that any document signed by the signing oracle could simply be signed by an extra token instead. However, this turns out to be not necessarily true.

\begin{proposition}

There exists an $\UnfDef^{\tokengen,\Vrfy}$ $\Tmac$ scheme which is not $\UnfDef^{\widetilde{\Sign},\Vrfy}$.
\end{proposition}
\begin{proof}
Let $\Pi$ be a deterministic verification $\str\textit{-}\UnfDef^{\tokengen,\Vrfy}$ scheme such that for some message $m_1$, there are at least two distinct signatures which $\widetilde{\Sign}$ generates with non-negligible probability (The scheme presented in~\cite{BS16a} satisfies these requirements; Our scheme is not \emph{strongly} unforgeable). We define a new $\Tmac$ scheme $\Pi'$ which is the same as $\Pi$ up to a modification to the verification algorithm.
Fix $m_0\neq m_1$. In addition to what $\Pi.\Vrfy$ accepts, $\Pi'.\Vrfy$ also accepts two distinct signatures for $m_1$ concatenated to each other, as a valid signature for $m_0$.  $\Pi'$ is clearly not $\UnfDef^{\widetilde{\Sign},\Vrfy}$. An adversary could easily request the oracle to sign $m_1$ twice, resulting (with non-negligible probability) in two distinct signatures for $m_1$. Concatenating those is a valid signature for $m_0$. 

Next we claim that $\Pi'$ is  $\UnfDef^{\tokengen,\Vrfy}$. Suppose not, then there exists $\adv$ asking for $\ell$ tokens, and submitting $(a_1,\sigma_1),\ldots,(a_{\ell+1},\sigma_{\ell+1})$ such that with non-negligible probability all are valid signed documents, and the documents are all distinct. Let $p_0(\secpar)$ denote the probability that $\adv$ succeeds, and among the signatures submitted there is a signature for $m_0$ which is not valid in $\Pi$, and $p_1(\secpar)$ the probability that $\adv$ succeeds, but all submitted signatures are also valid in $\Pi$. Consider $\Badv$ is an adversary against $\Pi$,  simulating $\adv$ until it receives its final reply. With probability $p_1(\secpar)$, the signatures submitted by $\adv$ are all also valid in $\Pi$, immediately implying $p_1(\secpar)$ is negligible. With probability $p_0(\secpar)$ and W.L.O.G, $a_{\ell+1}=m_0$ and $\sigma_{\ell+1}$ is a concatenation of two signatures for $m_1$. In such a scenario, (which is easily recognizable by $\Badv$), $\Badv$ could simply submit the two distinct signatures for $m_1$, alongside $(a_1,\sigma_1)\ldots(a_{\ell-1},\sigma_{\ell-1})$. Here W.L.O.G., it is assumed that if there originally were some $a_i=m_1$, then it was the $\ell^{th}$ document. As all the signed documents submitted by $\Badv$ are valid and distinct, this comprises a win in $\str\textit{-}\UnfExp_{\Badv,\Pi}^{\tokengen,\Vrfy}$ (\cref{def:strong TMAC}); hence, $p_1(\secpar)$ must be negligible as well, contradicting the assumption that  $p_0(\secpar)+p_1(\secpar)$ was non-negligible, and providing a separation between the security notions.
\end{proof}%

The flaw with the intuitive reduction is that the documents queried to the signing oracle might not be distinct. The intuitive reduction can be formalized  if we can prevent this event from occurring.\footnote{In~\cite{BS16a}, the authors point out that such events cannot occur except with negligible probability, in $\Tmac$ schemes that satisfy the %two properties. This provides conditions on $\Tmac$ schemes in which this issue does not occur: demanding the 
strong variant of unforgeability as well as a property called unpredictability. Unpredictability means that if we sign (by $\widetilde{\Sign}$) a document twice, this results in two different signatures with overwhelming probability. While the second property is fulfilled by our construction, the first is not; hence, this approach is of little use in this work.}
In practice, any $\UnfDef^{\tokengen,\Vrfy}$ $\Tmac$ can be used to construct an $\UnfDef^{\widetilde{\Sign},\Vrfy}$ $\Tmac$, which is in fact, an $\UnfDef^{\tokengen,\widetilde{\Sign},\Vrfy}$ $\Tmac$, as discussed  in~\cref{sec:Unforgeability In The Presence of a Signing Oracle}.

%------------------------------------------------------------------------------------------------------------------------%

\ifnum\llncs=0
\section{Unconditional Unforgeability for a Fixed Number of Tokens}\label{sec:Fixed J Secure Tokens}
  
%------------------------------------------------------------------------------------------------------------------------%
In previous sections, we saw constructions of $\Unc$ $\UnfDef^{\OneToken,\Vrfy}$ and  $\UnfDef^{\tokengen,\Vrfy}$ $\Tmac$s. Moreover, we know by \cref{thm: impossibility of unconditionally secure TMAC scheme.} that unconditional security cannot be achieved for polynomially many tokens. This leads us to an intermediate question of achieving unconditional security for a fixed constant number of tokens. 

\begin{definition}
 For $j\in \mathbb{N}$, $\UnfExp^{\JToken,\Vrfy}_{\adv,\Pi}$ $\Tmac$ is identical to $\UnfExp^{\OneToken,\Vrfy}_{\adv,\Pi}$, but the adversary in the security game is allowed to make up to $j$ queries to $\tokengen$. The corresponding security notion is called $\UnfDef^{\JToken,\Vrfy}$ $\Tmac$.\label{def: J Tmac}
\end{definition}

We use the same ideas as in  \cref{sec:Unforgeability Against Polynomial Tokens Attacks} to achieve this notion of security unconditionally, for length-restricted schemes. The idea is to replace (computationally) secure authenticated encryption used in  \cref{sec:Unforgeability Against Polynomial Tokens Attacks} with a direct application of the "encrypt then sign" paradigm, using  $\mac$ and encryption schemes that satisfy information-theoretic variants of secrecy and unforgeability respectively, for a fixed number of documents, i.e., up to $j$ documents.
\begin{definition}
\label{def: J secure encryption}
 For $j\in \mathbb{N}$, a deterministic encryption scheme  $\Pi=(\keygen,\enc,\dec)$ is said to be   $\mathsf{statistically}$  $j\textit{-}\mathsf{secure}$ if for any two $j$-tuples of messages  $(m_1,m_2,...m_j)$ and $(m'_1,m'_2,...m'_j)$ and two $j$-tuples of cipher-texts $(c_1,c_2,...c_j)$ and $(c'_1,c'_2,...c'_j)$s  \[\abs{\Pr[\enc_k(m_1)=c_1,...\enc_k(m_1)=c_1]-\Pr[\enc_k(m'_1)=c'_1,...\enc_k(m'_1)=c'_1]}\leq \negl,\]

where the probability is taken over the distribution of $k\gets \keygen(1^\secpar)$, as well as the randomness of the encryption. 
\end{definition}

\begin{definition}
\label{def: J unforgeable mac}
 For $j\in \mathbb{N}$, a $\mac$ scheme $\Pi=(\keygen,\Sign,\Vrfy)$ is said to be $\Unc$  $j\textit{-}\UnfDef$ if for any $\secpar$, and every key $k$ generated by $\keygen(1^\secpar)$ and any  $j'< j$, the probability of any computationally unbounded adversary in hold of $j'$ signed document, cannot submit a signature for a fresh document except, with negligible probability.
\end{definition}

\begin{remark}
There exists a length-restricted deterministic  $\mac$ (with canonical verification) scheme satisfying \cref{def: J unforgeable mac} (\cite{WC81}). The computational cost of this scheme is polynomial in the security parameter, the length of the message and in $j$.
\end{remark}

\begin{remark}
There exists a length-restricted encryption scheme satisfying \cref{def: J secure encryption}.\footnote{The existence of such a scheme is a known folklore. The construction is as follows. A $j$-wise independent hash function $h$ sampled as the key, the domain and range of which is same as the message space. The encryption of a message $m$ is  $(h(r)\oplus m,r)$, where $r$ is sampled uniformly from the message space. See for example~\cite{DW17} for a private case of this result.} The computational cost of this scheme is polynomial in the security parameter, the length of the message and in $j$.
\end{remark}

\begin{restatable}{lemma}{RstarToJ}\label{lem: restricted onetime implies j times restricted} 
For any $\ell\in \NN$, there is a noise-tolerance preserving lift (\cref{def:noise tolerance}) of an $\ell$-restricted $\Unc$ $\UnfDef^{\OneToken,\Vrfy}$ $\Tmac$ (see \cref{def: ell-restricted TMAC,def:Unforgeability}), to an $\ell$-restricted  $\Unc$ $\UnfDef^{\JToken,\Vrfy}$ $\Tmac$ (see \cref{def: J Tmac}), as shown in \cref{alg: TomJ}.
\end{restatable} 

The proof of security follows the same principles as \cref{lem: unrestricted onetime implies full blown unrestricted}.
The resulting scheme obtained by instantiating $\TM$ with the specific scheme $\oneunrest{\eta}$ (\cref{alg: OT}), is called $\Junrest{\eta}$.

\begin{algorithm} %\label{alg:TOM}
     \caption{$\TM$ - An $\Unc$ $\UnfDef^{\JToken,\Vrfy}$ $\Tmac$, for parameter $j$.}
    \label{alg: TomJ}
    \textbf{Assumes:} $\OT$ is an $\ell$-restricted $\Tmac$, $\Cmac$, $\Cenc$ are families of encryption and $\mac$ schemes respectively, where if $\keygen$ is provided with the parameter $j$, the schemes constitute   a $j\textit{-}\UnfDef$ $\mac$, and a $j\textit{-}\Sec$ encryption, respectively. The domain of  $\Cenc$ is assumed to be the message space, and the domain of $\Cmac$ is assumed to be the cipher-text space $\Cenc$.
    \begin{algorithmic}[1] % The number tells where the line numbering should start
        \Procedure{$\keygen$}{$1^\secpar,j$}
            \State \textbf{Return} $(\Cenc.\keygen(1^\secpar,j),\Cmac.\keygen(1^\secpar,j))$.
        \EndProcedure
    \end{algorithmic}
    \begin{algorithmic}[1] % The number tells where the line numbering should start
        \Procedure{$\tokengen_k$}{}
            \State Interpret $k=(k_{\Cenc},k_{\Cmac})$. 
            \State $
            \kappa \gets \OT.\keygen(1^\secpar) $.
            \State   $\kenc\gets \Cenc.\enc_{k_{\Cenc}}(\kappa)$, $\ket{\widetilde{\Stamp}}\gets\OT.\tokengen_{\kappa}$.
            
            \State \textbf{Return}
            $(\ket{\widetilde{\Stamp}},\kenc,\Cmac.\Sign_{k_{\Cmac}}(\kenc))$.
        \EndProcedure

    \end{algorithmic}
	\begin{algorithmic}[1] % The number tells where the line numbering should start
        \Procedure{$\Sign_{\ket{\Stamp}}$}{$m$}
            \State Interpret $\ket{\Stamp}$ as $(\ket{\widetilde{\Stamp}},\kenc, \ksig)$.
            \State \textbf{Return} $(\OTO.\Sign_{\ket{\widetilde{\Stamp}}}(m),\kenc,\ksig)$.
        \EndProcedure

    \end{algorithmic}
  	\begin{algorithmic}[1] % The number tells where the line numbering should start
    \Procedure{$\Vrfy_k$}{$m,\sigma$}
    
            \State Interpret $\sigma$ as $(s,\kenc,\ksig)$, and $k$ as $(k_{\Cenc},k_{\Cmac})$.
            \State $Res\gets \Cmac.\Vrfy_{k_{\Cmac}}(\kenc,\ksig)$.
            \If{Res=0}
             \State \textbf{Return 0}.
            \Else
                \State $\kappa\gets \Cenc.\Dec_{k_{\Cenc}}(\kenc)$.
                \State \textbf{Return} $\OTO.\Vrfy_{\kappa}(m,s)$.
            \EndIf

    \EndProcedure
    \end{algorithmic}

\end{algorithm}

\fi

\printnomenclature
\label{pg:nomenclature}
\ifnum\masterthesis=1
    \bibliography{main}
\fi

\end{document}